\documentclass[10pt]{article}

\usepackage{amsthm}
\usepackage{graphicx} 
\usepackage{array} 

\usepackage{amsmath, amssymb, amsfonts, verbatim}
\usepackage{hyphenat,epsfig,subcaption,multirow}
\usepackage{nicefrac}
\usepackage{paralist}

\usepackage[font=small, labelfont=bf]{caption}

\usepackage[usenames,dvipsnames]{xcolor}
\usepackage[ruled]{algorithm2e}

\DeclareFontFamily{U}{mathx}{\hyphenchar\font45}
\DeclareFontShape{U}{mathx}{m}{n}{
      <5> <6> <7> <8> <9> <10>
      <10.95> <12> <14.4> <17.28> <20.74> <24.88>
      mathx10
      }{}
\DeclareSymbolFont{mathx}{U}{mathx}{m}{n}
\DeclareMathSymbol{\bigtimes}{1}{mathx}{"91}

\usepackage{tcolorbox}
\tcbuselibrary{skins,breakable}
\tcbset{enhanced jigsaw}

\usepackage[normalem]{ulem}
\usepackage[compact]{titlesec}

\definecolor{DarkRed}{rgb}{0.5,0.1,0.1}
\definecolor{DarkBlue}{rgb}{0.1,0.1,0.5}

\usepackage{nameref}
\definecolor{ForestGreen}{rgb}{0.1333,0.5451,0.1333}
\definecolor{Red}{rgb}{0.9,0,0}
\usepackage[linktocpage=true,
	pagebackref=true,colorlinks,
		urlcolor=black,
	linkcolor=DarkRed,citecolor=ForestGreen,
	bookmarks,bookmarksopen,bookmarksnumbered]
	{hyperref}
\usepackage[noabbrev,nameinlink]{cleveref}
\crefname{property}{property}{Property}
\creflabelformat{property}{(#1)#2#3}
\crefname{equation}{eq}{Eq}
\creflabelformat{equation}{(#1)#2#3}

\usepackage{bm}
\usepackage{url}
\usepackage{xspace}
\usepackage[mathscr]{euscript}

\usepackage{tikz}
\usetikzlibrary{arrows}
\usetikzlibrary{arrows.meta}
\usetikzlibrary{shapes}
\usetikzlibrary{backgrounds}
\usetikzlibrary{positioning}
\usetikzlibrary{decorations.markings}
\usetikzlibrary{patterns}
\usetikzlibrary{calc}
\usetikzlibrary{fit}
\usetikzlibrary{snakes}

\usepackage{mdframed}

\usepackage[noend]{algpseudocode}
\makeatletter
\def\BState{\State\hskip-\ALG@thistlm}
\makeatother

\usepackage{cite}
\usepackage{enumitem}

\usepackage[margin=1in]{geometry}

\newtheorem{theorem}{Theorem}
\newtheorem{lemma}{Lemma}[section]

\newtheorem{proposition}[lemma]{Proposition}

\newtheorem{claim}[lemma]{Claim}

\newtheorem{problem}{Problem}

\newtheorem*{claim*}{Claim}
\newtheorem*{proposition*}{Proposition}
\newtheorem*{lemma*}{Lemma}
\newtheorem*{problem*}{Problem}

\crefname{lemma}{Lemma}{Lemmas}
\crefname{claim}{Claim}{Claims}

\newtheorem{mdresult}{Result}

\newtheorem*{mdresult*}{Main Result}
\newenvironment{Result}{\begin{mdframed}[backgroundcolor=lightgray!40,topline=false,rightline=false,leftline=false,bottomline=false,innertopmargin=2pt]\begin{mdresult*}}{\end{mdresult*}\end{mdframed}}

\newtheorem{assumption}{Assumption}
\newtheorem{definition}[lemma]{Definition}
\theoremstyle{definition}
\newtheorem{remark}[lemma]{Remark}

\newtheorem{mdinvariant}[lemma]{Lemma}
\newenvironment{Lemma}{\begin{mdframed}[hidealllines=false,backgroundcolor=gray!10,innertopmargin=5pt]\begin{mdinvariant}}{\end{mdinvariant}\end{mdframed}}

\theoremstyle{definition}
\newtheorem{mdalg}{Algorithm}
\newenvironment{Algorithm}{\begin{tbox}\begin{mdalg}}{\end{mdalg}\end{tbox}}

\allowdisplaybreaks

\renewcommand{\qed}{\nobreak \ifvmode \relax \else
      \ifdim\lastskip<1.5em \hskip-\lastskip
      \hskip1.5em plus0em minus0.5em \fi \nobreak
      \vrule height0.75em width0.5em depth0.25em\fi}

\newcommand{\Qed}[1]{\ensuremath{\qed_{\textnormal{~#1}}}}

\newcommand{\eps}{\ensuremath{\varepsilon}}
\newcommand{\Paren}[1]{\Big(#1\Big)}

\newcommand{\bracket}[1]{\left[#1\right]}
\newcommand{\paren}[1]{\ensuremath{\left(#1\right)}\xspace}
\newcommand{\card}[1]{\left\vert{#1}\right\vert}

\newcommand{\IR}{\ensuremath{\mathbb{R}}}

\newcommand{\norm}[1]{\ensuremath{\|#1\|}}

\newcommand{\ceil}[1]{{\left\lceil{#1}\right\rceil}}

\newcommand{\prob}[1]{\Pr\paren{#1}}
\newcommand{\expect}[1]{\Exp\bracket{#1}}
\newcommand{\var}[1]{\textnormal{Var}\bracket{#1}}
\newcommand{\cov}[1]{\textnormal{Cov}\bracket{#1}}

\newcommand{\set}[1]{\ensuremath{\left\{ #1 \right\}}}
\newcommand{\poly}{\mbox{\rm poly}}
\newcommand{\polylog}{\mbox{\rm  polylog}}

\newcommand{\opt}{\textnormal{\ensuremath{\mbox{opt}}}\xspace}

\newcommand{\alg}{\ensuremath{\mathcal{A}}\xspace}

\DeclareMathOperator*{\Exp}{\ensuremath{{\mathbb{E}}}}
\DeclareMathOperator*{\Prob}{\ensuremath{\textnormal{Pr}}}
\renewcommand{\Pr}{\Prob}

\newenvironment{tbox}{\begin{tcolorbox}[
		enlarge top by=5pt,
		enlarge bottom by=5pt,
		 breakable,
		 boxsep=0pt,
                  left=4pt,
                  right=4pt,
                  top=10pt,
                  arc=0pt,
                  boxrule=1pt,toprule=1pt,
                  colback=white
                  ]
	}
{\end{tcolorbox}}

\newcommand{\supp}[1]{\ensuremath{\textnormal{\text{supp}}(#1)}}

\newcommand{\II}{\ensuremath{\mathbb{I}}}

\newcommand{\mireal}[1][]{
  \ifx\relax#1\relax%
    \II(\mione \,; \mitwo)%
  \else%
    \II(\mione \,; \mitwo\mid #1)%
  \fi
}

\newcommand{\Prot}{\Pi}

\newcommand{\MM}{\ensuremath{{M}}}

\newcommand{\RR}{\ensuremath{\mathbf{R}}}

\title{An Asymptotically Optimal Algorithm for Maximum Matching \\ in Dynamic Streams}
\author{
	Sepehr Assadi\footnote{(\href{mailto:sepehr.assadi@rutgers.edu}{sepehr.assadi@rutgers.edu)} Department of Computer Science, Rutgers University.  Research supported in part by a NSF CAREER Grant CCF-2047061, and a gift from Google Research.} \and 
Vihan Shah\footnote{(\href{mailto:vihan.shah98@rutgers.edu}{\text{vihan.shah98@rutgers.edu}}) Department of Computer Science, Rutgers University.  Research supported in part by a NSF CAREER Grant CCF-2047061. }  
}

\date{}

\begin{document}
\maketitle

\pagenumbering{roman}


\begin{abstract}
	We present an algorithm for the maximum matching problem in dynamic (insertion-deletions) streams with \textbf{asymptotically optimal} space complexity: for any $n$-vertex graph, 
	our algorithm with high probability outputs an $\alpha$-approximate matching in a single pass  using $O(n^2/\alpha^3)$ \emph{bits} of space. 
	
	\medskip
	
	A long line of work on the dynamic streaming matching problem has reduced the gap between space upper and lower bounds  first to $n^{o(1)}$ factors [Assadi-Khanna-Li-Yaroslavtsev; SODA 2016] and subsequently to $\polylog{(n)}$ factors [Dark-Konrad; CCC 2020]. Our upper bound now matches the Dark-Konrad lower bound up to $O(1)$ factors, thus completing this research direction. 
	
	\medskip
	
	Our approach consists of two main steps: we first (provably) identify a family of graphs, similar to the  instances used in prior work to establish the lower bounds for this problem, as 
	the only ``hard'' instances to focus on. These graphs include an induced subgraph which is both sparse and contains a large matching. We then design a dynamic streaming algorithm for this family of graphs 
	which is more efficient than prior work. The key to this efficiency is a novel sketching method, which bypasses the typical loss of $\poly\!\log{(n)}$-factors in space compared to standard $L_0$-sampling primitives, and 
	can be of independent interest in designing  optimal algorithms for other streaming problems.

\end{abstract}

\clearpage
\bigskip

\setcounter{tocdepth}{3}
\tableofcontents

\clearpage

\pagenumbering{arabic}
\setcounter{page}{1}


\newcommand{\diag}{\textnormal{\textbf{diag}}\xspace}

\newcommand{\vect}{\textnormal{\textbf{vec}}\xspace}

\newcommand{\Lsampler}{\textnormal{\texttt{L0-Sampler}}\xspace}
\newcommand{\sprec}{\textnormal{\texttt{Sparse-Recovery}}\xspace}

\newcommand{\slz}{s_{\textnormal{\texttt{L0}}}}
\newcommand{\ssr}{s_{\textnormal{\texttt{SR}}}}

\newcommand{\XX}{\ensuremath{\mathcal{X}}}
\section{Introduction}\label{sec:intro} 

We study the maximum matching problem in the \emph{dynamic} streaming model. In this problem, the edges of an input graph $G=(V,E)$ are presented to the algorithm as a sequence of both edge insertions and deletions. The goal 
is to recover an approximate maximum matching of $G$ at the end of the stream using a limited space smaller than the input size, namely, $o(n^2)$ space where $n$ is the number of vertices. The dynamic graph streaming model is highly motivated by applications to processing massive graphs 
and has been studied extensively in recent years; see, e.g.~\cite{AhnGM12,AhnGM12b,KapralovLMMS14,Konrad15,BhattacharyaHNT15,McGregorTVV15,AssadiKLY16,AssadiKL17,KallaugherKP18,NelsonY19,AssadiCK19a,BeraCG19,KapralovMMMNST20,DarkK20,Konrad21} and references therein.  

A brief note on the history of dynamic streaming matching  is in order. Initiated by a breakthrough result of~\cite{AhnGM12},  
for most graph problems studied in insertion-only streams, researchers were able to subsequently obtain algorithms with similar guarantees in dynamic streams as
well; this includes connectivity~\cite{AhnGM12}, cut sparsifiers~\cite{AhnGM12b}, spectral sparsifiers~\cite{KapralovLMMS14}, densest subgraph~\cite{McGregorTVV15}, subgraph counting~\cite{AhnGM12b}, $(\Delta+1)$-vertex coloring~\cite{AssadiCK19a}, among many others. This placed the maximum matching problem in a rather unique position in the literature: while there is a straightforward $2$-approximation algorithm for this problem in insertion-only streams 
using only $O(n\log{n})$ space~\cite{FeigenbaumKMSZ05}, no non-trivial approximation algorithms were developed for this problem even in $o(n^2)$ space, despite significant attention; see, e.g.~\cite{TurnstileMatchingOP,ChitnisCHM15,ChitnisCEHM15}. 

This problem was addressed in a series of (independent and concurrent) work~\cite{Konrad15,AssadiKLY16,ChitnisCEHMMV16}. In particular,~\cite{AssadiKLY16} proved that any $\alpha$-approximation algorithm for matching
in dynamic streams requires $(n^{2-o(1)}/\alpha^3)$ space and designed an $\alpha$-approximation algorithm with $O(n^2/\alpha^3 \cdot \poly\!\log{(n)})$ space for this problem (\!\cite{Konrad15} gave a slightly weaker lower and upper bounds for this problem and~\cite{ChitnisCEHMMV16} obtained an algorithm with similar performance as~\cite{AssadiKLY16}). The work of~\cite{AssadiKLY16} thus brought the gap between space upper and lower bounds on this problem down to an $n^{o(1)}$ factor. 
The lower bound of~\cite{AssadiKLY16} relied on a remarkable characterization of dynamic streaming algorithms due to~\cite{LiNW14,AiHLW16} that allows for transforming linear sketching lower bounds to dynamic streams. 
However, this characterization  requires making strong requirements from the streaming algorithms (such as processing doubly exponentially long streams); see~\cite{KallaugherP20} for a detailed discussion on this topic.
More recently,~\cite{DarkK20} bypassed this characterization step entirely and along the way, improved the lower bound for this problem to $\Omega(n^2/\alpha^3)$ space directly in dynamic streams. This constitutes the state-of-the-art 
for the dynamic streaming matching problem. 

In parallel to this line of work on the matching problem that focused on determining the ``high order terms'' in the space complexity of this problem (namely, up to $n^{o(1)}$ or $\poly\!\log{(n)}$ factors), there has also been substantial 
work on determining the ``lower order terms'' on space complexity of other dynamic graph streaming problems~\cite{SunW15,KapralovNPWWY17,NelsonY19,Yu21}. For instance,~\cite{NelsonY19}, building on~\cite{KapralovNPWWY17}, 
proved that any dynamic streaming algorithm for connectivity requires $\Omega(n\cdot\log^3\!{(n)})$ space which matches the algorithm of~\cite{AhnGM12} up to constant factors. This quest for obtaining 
asymptotically optimal bounds is common in the streaming literature beyond graph streams such as in frequency moment estimation~\cite{KaneNW10a,KaneNW10b,LiW13,AndoniNPW13,BravermanKSV14}, 
empirical entropy~\cite{HarveyNO08,ChakrabartiCM10,JayramW13}, numerical linear algebra~\cite{ClarksonW09}, compressed sensing~\cite{PriceW11,PriceW13}, and sampling~\cite{KapralovNPWWY17}. 

This state-of-affairs is the  motivation behind our work: \emph{Can we determine the space complexity of the maximum matching problem down to its lower order terms?} We resolve this question in the affirmative by presenting 
an improved algorithm for this problem. 

\begin{Result}[Formalized in~\Cref{thm:main}]
	There is a dynamic streaming algorithm that with high probability outputs an $\alpha$-approximation to maximum matching using $O(n^2/\alpha^3)$ space for any $\alpha \ll n^{1/2}$. 
\end{Result}
Let us right away note that the condition of $\alpha \ll n^{1/2}$ in our main result is not arbitrary\footnote{The same condition is used in all prior lower bounds in~\cite{Konrad15,AssadiKLY16,DarkK20} as well as algorithms~\cite{Konrad15,ChitnisCEHMMV16}
with the exception of algorithm of~\cite{AssadiKLY16}.}: for $\alpha > n^{1/2}$, we have $n^{2}/\alpha^3 < n/\alpha$, while one needs $\Omega((n/\alpha) \cdot \log{n})$ space simply to store an $\alpha$-approximate matching! As a result, 
our algorithm now matches the lower bound of~\cite{DarkK20} up to constant factors in almost the entirety of its meaningful regime for parameter $\alpha$, thus completely resolving the space complexity of the maximum matching problem 
in dynamic streams. We now discuss further aspects of our work. 

\paragraph{Beyond $L_0$-samplers.} The key technique in dynamic graph streams is the use of $L_0$-samplers\footnote{We are only aware of a single work~\cite{KapralovLMMS14} in dynamic graph streams that does not use $L_0$-samplers.} that allow for sampling an edge from an stream that contains both insertions and deletions of the edges (see~\Cref{sec:prelim-toolkit}). 

Previously-best algorithms of~\cite{AssadiKLY16,ChitnisCEHMMV16} for dynamic streaming matching sample $O(n^2/\alpha^3)$ edges from the input graph (from a carefully-designed non-uniform distribution) and show that this sample contains an $\alpha$-approximate matching. For the sampling, they need to use $L_0$-samplers that will bring in an additional $\poly\!\log{(n)}$ factor overhead in the space. At the same time, a careful examination of the lower bound of~\cite{DarkK20} suggests that one needs to recover $\Omega(n^2/\alpha^3)$ \emph{edges} from the graph (not only \emph{bits}, assuming one only communicates edges). 
On top of this, the lower bound of~\cite{NelsonY19} for the connectivity problem is based on showing that recovering $(n-1)$ edges of a spanning forest in the input, essentially require paying the cost of $(n-1)$ $L_0$-samplers as well, 
leading to their $\Omega(n\cdot\log^3\!{(n)})$ lower bound. Putting all this together, it is natural to conjecture that one also needs $\Omega(n^2/\alpha^3 \cdot \poly\!\log\!{(n)})$ space for the matching problem\footnote{This was in fact the authors' conjecture 
at the beginning of this project.}. 

Our algorithm in this paper is still based on finding $\Theta(n^2/\alpha^3)$ edges from the input graph. It turns out however that one can do this more efficiently than using the same number of $L_0$-samplers. 
In particular, we show a way of recovering these edges with only $O(1)$ bit overhead per edge \emph{on average}. 
This is achieved using a novel sketching primitive in this paper (\Cref{sec:snr}). On a high level, this sketch allows us to recover \emph{sparse induced subgraphs} of the input graph, specified to the algorithm only at the end stream, in a more efficient manner than recovering them one edge at a time via $L_0$-samplers\footnote{Let us note that our sketch cannot do magic: The problem of finding sparse induced subgraphs is at the core of the lower bound approaches for dynamic streaming matching in~\cite{AssadiKLY16,DarkK20}, thus there is no hope of solving it ``efficiently''. Our sketch shows that one can recover these graphs without paying any extra cost \emph{over} the lower bounds of these work.}. We believe this idea can be useful 
for obtaining asymptotically optimal algorithms for other dynamic streaming problems as well. 

\paragraph{Classifying input graphs.} Another key idea  in our paper is a way of roughly classifying input graphs into ``easy'' and ``hard'' instances. Informally speaking, the easy instances are the ones that one can recover a large matching from them 
by sampling $\ll n^2/\alpha^3$ edges (again, in a non-uniform way). Such a graph can then be handled in $O(n^2/\alpha^3)$ space even if we use $L_0$-samplers for our sampling given we now need much fewer number of samples than before. 
One of our two main lemmas (\Cref{lem:ms}) gives one characterization of these graphs: essentially, any ``hard'' graph, i.e., a one not solvable by the above approach, includes a subgraph on $n-o(n/\alpha)$ vertices with only $\approx n$ edges
and a matching of size $\approx n-o(n/\alpha)$ (they essentially have an induced matching of size $\approx n-o(n/\alpha)$). A reader familiar with~\cite{Konrad15,AssadiKLY16,DarkK20}  may notice that this family precisely captures the graphs  in  prior lower bounds for dynamic streaming matching problem. 

Our next main lemma (\Cref{lem:sc}) then gives an algorithm for solving these hard graphs. The idea behind the algorithm is as follows. Let $S$  denote the vertices in the induced sparse subgraph of the input and let $T$ be the 
remaining vertices (we will be able to recover an approximate version of this partitioning \emph{at the end} of the stream). \emph{If} we are able to recover edges inside $S$, we will be done as there is a large matching in $S$ and it does not have too many extra edges. The problem is that we will not know this set until the end of the stream and by that point we should have collected all the required information. This is where our main sketching tool  mentioned earlier comes into place. 
Informally, the sketch allows us to, for any vertex $v \in S$, recover the neighbors $N(v)$ of $v$ using roughly $\card{N(v) \cap T} + \paren{\card{N(v)-T} \cdot \poly\log{(n)}}$ bits (as opposed to $\card{N(v)} \cdot \poly\!\log{(n)}$ bits via $L_0$-samplers). 
As the total number of edges outside $T$ is quite small, i.e., $\approx n$ in total, this is a  huge saving for us that allows for obtaining our desired $O(n^2/\alpha^3)$ bit upper bounds. 

We shall remark that in this discussion, we have been imprecise to give a rough intuition of our approach; the actual details turn out to be considerably more challenging as described in~\Cref{sec:sc}  and~\Cref{sec:ms}. 

\paragraph{``Shaving'' log-factors?} Finally, our improvement over prior work in~\cite{AssadiKLY16,ChitnisCEHMMV16} at no place is obtained via ``shaving log-factors''. 
Indeed there is a considerable gap of $n^{\Omega(1)}$ factor between the parameters that our easy-graph algorithms and hard-graph algorithms can still handle within $O(n^2/\alpha^3)$ space. 
This in turn allowed us to be quite cavalier with the parameters (e.g., using $\log{n}$-factors or $n^{\Omega(1)}$-factors where constant or $\poly\!\log\!{(n)}$ sufficed) and still recover an optimal space bound.


\section{Preliminaries}\label{sec:prelim}

 \paragraph{Notation.}
For a graph $G=(V,E)$, we write $\vect(E)$ to denote the ${{n} \choose {2}}$-dimensional vector where $\vect(E)_{i}$ denotes the multiplicity of the edge $e_i$ in $G$. We use $\deg(v)$ and $N(v)$ for each vertex $v \in V$ 
to denote the degree and neighborhood of $v$, respectively. For a subset $F$ of edges in $E$, we use $V(F)$ to denote the vertices incident on $F$; similarly, for a set $U$ of vertices in $V$, $E(U)$ denotes the edges incident on $U$. 

Throughout, we will use the term ``with high probability" to mean with probability at least $1-1/n^c$ for some large constant $c > 0$. The constant $c$ can be made arbitrarily large by only increasing the space of our algorithms with a constant factor and thus within the same asymptotic bounds. Moreover, for our purpose, this probability is large enough that one can always do a union bound over at most $\poly{(n)}$ different events that we consider in this paper; so we do not necessarily mention this each time.

 \subsection{Dynamic (Graph) Streams and Linear Sketches}\label{sec:prelim-model}
 
The dynamic streaming model is defined formally as follows. 
 
 \begin{definition}[Dynamic (graph) streams]\label{def:dynamic-stream}
 	A dynamic stream $\sigma = (\sigma_1,\ldots,\sigma_N)$ defines a vector $x \in \IR^m$. Each entry of the stream is a tuple $\sigma_i = (j_i,\Delta_i)$ for $j_i \in [m]$ and $\Delta_i \in \set{-1,+1}$. The vector $x$ is  defined as:
	\[
		\textnormal{for all $j \in [m]$}: \quad x_j = \sum_{\sigma_i: j_i = j} \Delta_i. 
	\]
	A dynamic graph stream is a dynamic stream wherein $m = {{n}\choose{2}}$ and $x = \vect(E)$ for the graph $G=(V,E)$ with $V = [n]$. Each update to $\vect(E)$  corresponds to inserting or deleting the specified
	edge from the graph. 
	
	A dynamic streaming algorithm makes a single pass over updates to $x$ and uses a limited memory, measured in \underline{number of bits}, and outputs an answer to the given problem at the end of the stream. 
 \end{definition}
 
Similar to virtually all other dynamic streaming algorithms, our algorithms will also be based on linear sketches, defined as follows. 
 
 \begin{definition}[{Linear sketch}]\label{def:linear-sketch}
 	Let $\Prot$ be a problem defined over vectors $x \in \IR^m$ (e.g., return the $\ell_2$-norm of $x$). A \emph{{linear sketch}} for $\Prot$ is an algorithm defined by the following pair: 
	\begin{itemize}
		\item \emph{\textbf{sketching matrix:}} A matrix $\Phi \in \poly{(m)}^{s \times m}$ that can be chosen randomly and implicitly; 
		\item \emph{\textbf{recovery algorithm:}} An algorithm that given the sketching matrix $\Phi$ and the vector $\Phi \cdot x$, returns a solution to $\Prot(x)$. 
	\end{itemize}
	We refer to the vector $\Phi \cdot x$ 	as a \underline{\emph{sketch}} of $x$, and to the number of bits needed to store $\Phi$ (implicitly) and $\Phi \cdot x$ as the \underline{\emph{size}} of the linear sketch. 

	The linear sketch for an input $x$ then consists of sampling a sketching matrix $\Phi$ (independent of $x$), computing the sketch $\Phi \cdot x$, and running the recovery algorithm on the sketch to solve the problem. 
	
	(We note that the computations can be on the set of integers (or reals) as well as on finite fields.)
 \end{definition}
 
For our purpose in this paper, we typically focus on graph problems for the choice of $\Prot$ in~\Cref{def:linear-sketch} and then set $x = \vect(E)$ where $E$ is the edge-set of the input graph. 
  The following proposition is well-known. 
 
 \begin{proposition}\label{prop:linear-sketch}
 	Let $\Prot$ be a problem defined over vectors $x \in \IR^m$. Suppose there exists a linear sketch of size $s(n)$ for $\Prot$ with probability of success $p(n)$. 
	Then, there is also a streaming algorithm for solving $\Prot$ on dynamic streams defining $x$ with probability of success $p(n)$ using $O(s(n)+\log{m})$ bits of space. 
 \end{proposition}
\begin{proof}
 	Let $\Phi$ be the sketching matrix for the linear sketch. It is enough to show that we can compute $\Phi \cdot x$ for vector $x$ in a dynamic stream as in~\Cref{def:dynamic-stream}; the rest then follows by running the recovery algorithm of the sketch
	at the end on $\Phi \cdot x$. This can be done easily by linearity of the sketch by maintaining  $\Phi \cdot x$ and for each update $\sigma_i = (j_i,\Delta_i)$ to $x$, updating it to 
	\[
	\Phi \cdot x + \Phi \cdot \Delta_i \cdot \mathbf{1}_{j_i} = \Phi \cdot (x+ \Delta_i \cdot \mathbf{1}_{j_i}) \tag{$\mathbf{1}_{j_i}$ is the vector that is 
	$1$ on $j_i$-th entry and zero elsewhere}
	\]
	to get the sketch of the updated vector $x$. As each update can be computed from the input and (implicit access to) the sketching matrix $\Phi$ with additional $O(\log{m})$ space for book-keeping, we are done.  
\end{proof}
 
 Given~\Cref{prop:linear-sketch}, in the rest of the paper, we simply focus on designing linear sketches for our dynamic streaming problems. 
 
 \subsection{Standard Sketching Toolkit}\label{sec:prelim-toolkit}
 
 We will also use $L_0$-samplers, a powerful tool used by most dynamic graph streaming algorithms, in our paper. The goal of $L_0$-samplers is to solve the following basic problem. 
 
 \begin{problem}[{$L_0$-Sampling}]\label{prob:l0}
 	Given a vector $x \in \IR^m$ specified in a dynamic stream, sample $x_i$ uniformly at random from the support of $x$ at the end of the stream. 
 \end{problem}
 
 We will typically use $L_0$-samplers by applying them to different pre-specified \emph{subsets} of edges (pairs of vertices) of the underlying graph to sample a uniform edge from those subsets.

 \begin{proposition}[\!\!\cite{JowhariST11,KapralovNPWWY17}]\label{prop:l0-sampler}
 	There is a linear sketch, called $\Lsampler$, for~\Cref{prob:l0}  with size 
	\[
	\slz = \slz(m,\delta_F,\delta_E) = O(\log m \cdot (\log m \cdot \log (1/\delta_F) + \log (1/\delta_E))
	\]
	 bits, that
	outputs \textnormal{FAIL} with probability at most $\delta_F$ and outputs a wrong answer with probability at most $\delta_E$. 
 \end{proposition}

Another standard tool we use is sparse recovery  to solve the following problem. 

\begin{problem}[Sparse Recovery]\label{prob:sparse-rec}
	Given an integer $k \geq 1$ and a vector $x \in \IR^m$ specified in a dynamic stream with the promise that $\norm{x}_0 \leq k$, recover all of $x$ at the end of the stream. 
\end{problem}

We use the following  result on sparse recovery over finite fields. 

\begin{proposition} [c.f.~\cite{DasV13}] \label{prop:sparse-rec}
	Let $q$ be any prime number and $k \geq 1$ be an arbitrary integer. There is a deterministic (poly-time computable) linear sketch, called $\sprec$, for~\Cref{prob:sparse-rec} for $k$-sparse vectors $x \in \mathbb{F}^m_q$, with size 
	\[
		\ssr = \ssr(m,k,q) = O(k \cdot \log{m} \cdot \log{q})
	\]
	bits that always outputs the correct answer on $k$-sparse vectors. Moreover, all computations of this linear sketch are also performed over the field $\mathbb{F}_q$. 
\end{proposition}

We shall note that for our application, we actually need the `moreover' part of~\Cref{prop:sparse-rec} (which limits the use of more standard sparse recovery approaches). 
 
 \subsection{Probabilistic Tools} 
 
 We use the following standard concentration inequalities. The first is a standard form of Chernoff bounds. 
 
 \begin{proposition}[Chernoff bound; c.f.~\cite{dubhashi2009concentration}]\label{prop:chernoff}
 	Suppose $X_1,\ldots,X_m$ are $m$ independent random variables with range $[0,1]$ each. Let $X := \sum_{i=1}^m X_i$ and $\mu_L \leq \expect{X} \leq \mu_H$. Then, for any $\eps > 0$, 
 	\[
 	\Pr\paren{X >  (1+\eps) \cdot \mu_H} \leq \exp\paren{-\frac{\eps^2 \cdot \mu_H}{3+\eps}} \quad \textnormal{and} \quad \Pr\paren{X <  (1-\eps) \cdot \mu_L} \leq \exp\paren{-\frac{\eps^2 \cdot \mu_L}{2+\eps}}.
 	\]
 \end{proposition}
 
 We also need McDiarmid's inequality when there is non-trivial correlation between random variables. 
 \begin{proposition}[McDiarmid's inequality~\cite{McDiarmid89}]\label{prop:mcdiarmid}
 	Let $X_1,\ldots,X_m$ be $m$ independent random variables where each $X_i$ has some range $\XX_i$. Let $f: \XX_1 \times \cdots \times \XX_m \rightarrow \IR$ be any $c$-Lipschitz function meaning that for  all $i \in [m]$ 
	and all choices of $(x_1,\ldots,x_m),(x'_1,\ldots,x'_m) \in  \XX_1 \times \cdots \times \XX_m$, 
	\[
		\card{f(x_1,\cdots , x_{i-1},x_i,x_{i+1},\cdots,x_m) - f(x_1,\cdots , x_{i-1},x'_i,x_{i+1},\cdots,x_m)} \leq c. 
	\]
	Then, for all $b > 0$, 
	\[
		\Pr\paren{\card{f(X_1,\ldots,X_m) - \expect{f(X_1,\ldots,X_m)}} \geq b} \leq 2 \cdot \exp\paren{-\frac{2 \, b^2}{m \cdot c^2}}. 
	\]
 \end{proposition}
 
 Finally, in certain places, we also use limited independence hash functions in our algorithms to reduce their space complexity. 
 
  \begin{definition}[Limited-independence hash functions]\label{def:hash}
 For integers $n,m,k \geq 1$, a family $\mathcal{H}$ of hash functions from $[n]$ to $[m]$ is called a 
 $k$-wise independent hash function iff for any two $k$-subsets $a_1,\ldots,a_k \subseteq [n]$ and $b_1,\ldots,b_k \subseteq [m]$,  
 \[
 \Pr_{h \sim \mathcal{H}}\paren{h(a_1)=b_1 \wedge \cdots \wedge h(a_k) = b_k} = \frac1{m^k}. 
 \]
 \end{definition}

 Roughly speaking, a $k$-wise independent hash function behaves like a totally random function when considering at most $k$ elements. We use the following standard result for $k$-wise independent hash functions.
 \begin{proposition}[\!\!\cite{MotwaniR95}]\label{prop:k-wise}
 	For every integers $n,m,k \geq 2$, there is a $k$-wise independent hash function $\mathcal{H} = \set{h: [n] \rightarrow [m]}$ so that sampling and storing a function $h \in \mathcal{H}$ takes $O(k \cdot (\log n + \log m))$ bits of space.
 \end{proposition}
 
 We shall also use the following concentration result on an extension of Chernoff-Hoeffding bounds for limited independence hash function. 
 \begin{proposition}[\!\!\cite{SchmidtSS95}]\label{prop:chernoff-limited}
 	Suppose $h$ is a $k$-wise independent hash function and $X_1,\ldots,X_m$ are $m$ random variables in $\set{0,1}$ where $X_i = 1$ iff $h(i) = 1$. 
	Let $X := \sum_{i=1}^{m} X_i$. Then, for any $\eps > 0$, 
	\[
		\Pr\paren{\card{X - \expect{X}} \geq \eps \cdot \expect{X}} \leq \exp\paren{-\min\set{\frac{k}{2},\frac{\eps^2}{4+2\eps} \cdot \expect{X}}}. 
	\] 
 \end{proposition}


\newcommand{\NES}{\textnormal{\texttt{NE-Sampler}}\xspace}
\newcommand{\NEC}{\textnormal{\texttt{NE-Counter}}\xspace}

\newcommand{\NET}{\ensuremath{\textnormal{\texttt{NE-Tester}}}\xspace}

\newcommand{\SNR}{\textnormal{\texttt{SN-Recovery}}\xspace}

\newcommand{\ssnr}{s_{\textnormal{\texttt{SNR}}}}
\newcommand{\snet}{s_{\textnormal{\texttt{NET}}}}

\newcommand{\eq}{\ensuremath{\textnormal{\textsc{Equality}}_N}\xspace}  
\newcommand{\skeq}{sk_{eq}}

\newcommand{\sam}{A}

\newcommand{\snes}{s_{\textnormal{\texttt{NES}}}}

\newcommand{\sett}{v_{sett}}
\newcommand{\Sett}{V_{sett}}

\newcommand{\nsett}{v_{sett}^{new}}

\newcommand{\vbad}{V_{bad}}
\newcommand{\badsize}{v_{bad}}

\newcommand{\inside}{\tilde{b}}

\section{New Sketching Toolkit}\label{sec:tools}

We present two novel linear sketches in this section that are needed for our main algorithm. The first one is a simple way of sampling random edges from a group of vertices to obtain an edge to a random neighbor of this set (as opposed to a random edge). The second (and main\footnote{The reason we consider this the most important of our sketches is that essentially \emph{all} our saving of $\poly\log{(n)}$ factors comes from the  efficiency of this sketch. 
For the first sketch, even a somewhat loose (in terms of extra $\poly\log{(n)}$ factors) bound in the space suffices for our purpose.}) linear sketch is a sparse-recovery-type sketch that allows for finding neighborhood of a vertex (or group of vertices) assuming we already know a set that intersects largely 
with the neighborhood.

\subsection{Neighborhood-Edge Sampler }\label{sec:nei-edge-sampler}

Suppose we have a group $S$ of vertices, and we want to sample a vertex $v$ from the neighborhood of $S$. If we want the probability of sampling $v$ to be proportional to $\deg(v)$, we can  sample an edge incident on $S$ (using an $\Lsampler$) and return the other endpoint; but what if we would like to sample $v$ uniformly at random from $N(S)$? There is a simple (and standard) solution for this problem using an $\Lsampler$ \emph{if we do not need to recover the edge incident on $v$}\footnote{Create an $n$-dimensional vector where entry $i$ denotes the number of edges incident on $v_i$ from $S$; then use an $\Lsampler$ to return an element from the support of this vector uniformly at random.}. However, for our purpose, we crucially need the edge as well therefore just an $\Lsampler$ will not work. We formulate the following problem to address this formally. 

\begin{problem}\label{prob:nei-edge-sampler}
	Given a graph $G=(V,E)$ specified in a dynamic stream, and a set $S \subseteq V$ of vertices at the start of the stream, output an edge $(u,v)$ such that $u \in S$ and $v$ is sampled uniformly at random from $N(S)$. 
\end{problem}
 
We design a linear sketch for solving this problem. 
\begin{lemma}\label{lem:nei-edge-sampler}
	There is a linear sketch, called $\NES(G,S)$, for~\Cref{prob:nei-edge-sampler} with size 
	\[
	\snes= \snes(n) = O(\log^3{n})
	\]
	 bits, that outputs \textnormal{FAIL} with probability at most $\nicefrac{1}{100}$ and gives a wrong answer with probability at most $n^{-8}$.
\end{lemma}

To solve~\Cref{prob:nei-edge-sampler}, we first need the following standard lemma. The proof of this lemma is known and is presented only for completeness. 
\begin{lemma}\label{lem:size-test}
	There is a linear sketch called \emph{\NEC} of size $O((\log{n})\cdot\log{(1/\delta_E}))$ bits that given a graph $G=(V,E)$ presented in a dynamic stream, and  any two sets $S$ and $T$ of vertices at the beginning of the stream, outputs
	whether or not $\card{N(S) \cap T} = 1$ with probability of error at most $\delta_E$. 
\end{lemma}
\begin{proof}
	The algorithm is simply as follows. For $i=1$ to $t=\log{(1/\delta_E)}$ iterations: 
	\begin{enumerate}[label=$(\roman*)$]
		\item Pick a pair-wise independent hash function $h_i: [n] \rightarrow \set{1,2}$. For $j \in \{1,2\}$, let $T_{ij} := \set{v \in T \mid h_i(v) = j}$. 
		\item Count the number of edges from $S$ to $T_{i1}$ and from $S$ to $T_{i2}$ using counters $c_{i1}$ and $c_{i2}$. If both counters are zero or both non-zero, then return $\card{N(S) \cap T} \neq 1$.
	\end{enumerate}
	If the algorithm never terminated up until here, output $\card{N(S) \cap T} = 1$. 
	
	The algorithm uses $O(\log{n} \cdot \log{(1/\delta_E)})$ space as by~\Cref{prop:k-wise}, it only needs $O(\log{n})$ bits per iteration to store each hash function (and another $O(\log{n})$ bits for the counters). 
	Moreover, whenever the algorithm returns $\card{N(S) \cap T} \neq 1$, the answer is correct: Either $S$ has zero neighbors in $T_{i1} \cup T_{i2} = T$, or it has non-zero number of edges from $S$ to both $T_{i1}$ and $T_{i2}$ implying $S$ has more than one neighbor in $T$. The only case in which the algorithm can make an error is when $\card{N(S) \cap T} > 1$, but it does not detect it. 
	
	Consider any pair of vertices $u \neq v$ in $N(S) \cap T$. In each iteration $i \in [t]$, the probability that $u$ and $v$ hash to the same value is half since $h_i$ is a pairwise-independent hash function. Thus, the probability that this event happens in all 
	$t = \log{(1/\delta_E)}$ iterations is $\delta_E$, which means the algorithm can only err with probability at most $ \delta_E$. 
\end{proof}

We are now ready to prove~\Cref{lem:nei-edge-sampler} using the following linear sketch. 

\begin{Algorithm}
	$\NES(G,S)$: A  linear sketch for~\Cref{prob:nei-edge-sampler}. 
	
	\medskip
	
	\textbf{Input:} A graph $G=(V,E)$; a set $S \subseteq V$ of vertices. 
	
	\medskip
	
	\textbf{Output:} An edge from $S$ to a uniformly random vertex of $N(S)$. 
	
	\medskip
	
	\textbf{Sketching matrix:} 
	
	\begin{enumerate}
		\item Repeat the following for $k=50$ iterations:		
		\item Do in parallel  $i \in \set{ 0, 1, \ldots,  \ceil{2 \log n}}$:
		\begin{enumerate}[leftmargin=10pt]
		\item Let $h_i: V \rightarrow [2^i]$ be a pairwise independent hash function. Define $T_i := \set{v \in V \mid h_i(v)=1}$ so that each vertex belongs to $T_i$ with probability $1/2^i$ (with pairwise independence across vertices). 
		\item Store an $\Lsampler$ $\mathcal{L}_i$ with parameters $\delta_E=n^{-10}$ and $\delta_F=1/100$ for edges between $S$ and $T_i$. Also, store a \NEC (\Cref{lem:size-test}) for $(S,T_i)$ with $\delta_E=n^{-10}$.
		\end{enumerate}
	\end{enumerate}
	
	\medskip
	
	\textbf{Recovery:} 
	
	\begin{enumerate}
		\item Go over all copies of \NEC in an arbitrary order and find one that returns ``$1$" as the size of the intersection. Extract an edge from the corresponding $\Lsampler$ and output it. 
		\item If no \NEC returns ``$1$" then output ``FAIL".
	\end{enumerate}
	
\end{Algorithm}

We now show the correctness of $\NES$. Let $d:=\card{N(S)}$ and let $i^*$ be such that $2^{i^*} \leq 4d < 2^{i^*+1}$. Consider the iteration $i^*$ wherein we sample each vertex in $T_{i^*}$ with probability $1/2^{i^*}$. 
Let $X$ be a random variable denoting the number of elements of $N(S)$ that are sampled.
Let $X_j$ be the random variable which is $1$ if the $j$-th vertex of $N(S)$ is sampled and $0$ otherwise. We have $X= \sum_{j \in N(S)} X_j$. We want to find the probability that $X=1$:
\begin{align*} 
	\prob{X=1} &= \sum_{j \in N(S)} \prob{X_j=1 \wedge X_{-j}=0} = \sum_{j \in N(S)} \prob{X_j=1} \cdot \prob{X_{-j}=0 \mid X_j=1} \notag \\
	&= \sum_{j \in N(S)} \frac{1}{2^{i^*}} \cdot \paren{1-\prob{X_{-j} \geq 1 \mid X_j=1}} \tag{by the choice of $T_{i^*}$} \\
	&\geq \sum_{j \in N(S)} \frac{1}{2^{i^*}} \cdot \paren{1-\expect{X_{-j} \mid X_j=1}} \tag{by Markov inequality} \\
	&= \sum_{j \in N(S)} \frac{1}{2^{i^*}} \cdot \paren{1-\sum_{j' \in N(S) \setminus {j}} \expect{X_{j'}}} \tag{by linearity of expectation and pairwise independence of $X_j,X_{j'}$ for $j \neq j' \in N(S)$} \\
	&= \frac{d}{2^{i^*}} \cdot \paren{1-\frac{d-1}{2^{i^*}}} \tag{as $\card{N(S)} = d$} \\
	&\geq \frac{1}{4} \cdot \paren{1-\frac{1}{2}} = \frac{1}{8}. \tag{as $2^{i^*} \leq 4d < 2^{i^*+1}$}
\end{align*}

For the purpose of analysis we say that the algorithm fails if the parallel iteration $i^*$ fails. This could happen if $X \neq 1$ or if the corresponding $\Lsampler$ fails. By the above calculation and the bound on $\delta_F = 1/100$, 
we get that the failure probability is at most $7/8 + 1/100 < 9/10$. Thus, the probability that all $k=50$ iterations fail is at most $(0.9)^{50} < 0.01$.
Therefore, exactly one element $v$ from $N(S)$ is picked in some iteration, and we can find an edge from $S$ to $v$ with probability at least $0.99$. 
Also, we can union bound over the error probabilities of $O(k \cdot \log{n})$ copies of $\Lsampler$ and \NEC giving a total error probability of at most $n^{-8}$. This proves the correctness of \NES as required in~\Cref{lem:nei-edge-sampler}. 

The space taken by an $\Lsampler$ is $O(\log^2 n)$ bits since $\delta_E=n^{-10}$ and $\delta_F=1/100$ (using \Cref{prop:l0-sampler}). The space taken by one \NEC (\Cref{lem:size-test}) is $O(\log^2 n)$ bits, and by a hash function is $O(\log n)$ bits. As we run $O(\log n)$ copies of $\Lsampler$, \NEC and hash functions in parallel, the total space taken by $\NES$ is $O(\log^3 n)$ bits implying \Cref{lem:nei-edge-sampler}.

\subsection{Sparse-Neighborhood Recovery}\label{sec:snr}
The second problem we would like to tackle is a sparse recovery type problem: suppose we have a group $S$ of vertices, and at the end of the stream, we (somehow) managed to find a \emph{superset} $T$ of all but a ``tiny'' fraction of vertices in $N(S)$. 
Can we recover the remainder of $N(S)-T$ efficiently  using our sketch? Formally, 

\begin{problem}[Sparse-Neighborhood Recovery]\label{prob:sparse-recovery}
	Let $a,b \geq 1$ be known integers such that $a \geq 100b$. Consider a graph $G=(V,E)$ specified in a dynamic stream and let $S \subseteq V$ be a known subset of vertices. The goal is to, given a set $T \subseteq V$ at the \underline{end} of the stream, 
	return the set $N(S) - T$, assuming the following promises: 
	\begin{enumerate}[label=$(\roman*)$]
		\item\label{promise1} size of $T$ is at most $a$;
		\item\label{promise2} size of $N(S) - T$ is at most $b$; 
		\item\label{promise3} for every vertex $v \in N(S) - T$, we have $\card{S \cap N(v)} < c$.
	\end{enumerate}
\end{problem}
In words, in~\Cref{prob:sparse-recovery}, we have a set $S$ of vertices, known at the \emph{start} of the stream, and we are interested in their neighbors \emph{outside} a given set $T$, specified at the \emph{end} of the stream. 
Our guarantees are roughly that $T$ is not ``too large'' (parameter $a$), neighborhood of $S$ outside $T$ is  ``small'' (parameter $b$), and each vertex outside $T$ only has ``few'' neighbors inside $S$ (parameter $c$). 
See~\Cref{fig:sketch-nbhd-rec} for an illustration. 

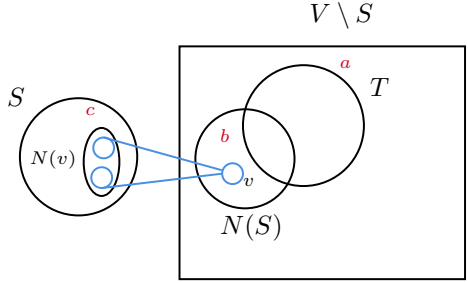
\begin{figure}[h!]
	\centering
	{\tikzset{every picture/.style={line width=0.75pt}} 

\begin{tikzpicture}[x=0.75pt,y=0.75pt,yscale=-1,xscale=1]
	
	\draw   (367,94) -- (511,94) -- (511,211.5) -- (367,211.5) -- cycle ;
	\draw   (286.5,149.75) .. controls (286.5,133.41) and (299.74,120.17) .. (316.08,120.17) .. controls (332.42,120.17) and (345.67,133.41) .. (345.67,149.75) .. controls (345.67,166.09) and (332.42,179.33) .. (316.08,179.33) .. controls (299.74,179.33) and (286.5,166.09) .. (286.5,149.75) -- cycle ;
	\draw   (399,134) .. controls (399,117.16) and (412.66,103.5) .. (429.5,103.5) .. controls (446.34,103.5) and (460,117.16) .. (460,134) .. controls (460,150.84) and (446.34,164.5) .. (429.5,164.5) .. controls (412.66,164.5) and (399,150.84) .. (399,134) -- cycle ;
	\draw   (375,150.75) .. controls (375,136.94) and (386.19,125.75) .. (400,125.75) .. controls (413.81,125.75) and (425,136.94) .. (425,150.75) .. controls (425,164.56) and (413.81,175.75) .. (400,175.75) .. controls (386.19,175.75) and (375,164.56) .. (375,150.75) -- cycle ;
	\draw  [color={rgb, 255:red, 74; green, 144; blue, 226 }  ,draw opacity=1 ] (323.5,145.25) .. controls (323.5,142.35) and (325.85,140) .. (328.75,140) .. controls (331.65,140) and (334,142.35) .. (334,145.25) .. controls (334,148.15) and (331.65,150.5) .. (328.75,150.5) .. controls (325.85,150.5) and (323.5,148.15) .. (323.5,145.25) -- cycle ;
	\draw  [color={rgb, 255:red, 74; green, 144; blue, 226 }  ,draw opacity=1 ] (388.5,158.25) .. controls (388.5,155.35) and (390.85,153) .. (393.75,153) .. controls (396.65,153) and (399,155.35) .. (399,158.25) .. controls (399,161.15) and (396.65,163.5) .. (393.75,163.5) .. controls (390.85,163.5) and (388.5,161.15) .. (388.5,158.25) -- cycle ;
	\draw  [color={rgb, 255:red, 74; green, 144; blue, 226 }  ,draw opacity=1 ] (322.5,160.25) .. controls (322.5,157.35) and (324.85,155) .. (327.75,155) .. controls (330.65,155) and (333,157.35) .. (333,160.25) .. controls (333,163.15) and (330.65,165.5) .. (327.75,165.5) .. controls (324.85,165.5) and (322.5,163.15) .. (322.5,160.25) -- cycle ;
	\draw   (318.67,152.33) .. controls (318.67,142.85) and (322.7,135.17) .. (327.67,135.17) .. controls (332.64,135.17) and (336.67,142.85) .. (336.67,152.33) .. controls (336.67,161.81) and (332.64,169.5) .. (327.67,169.5) .. controls (322.7,169.5) and (318.67,161.81) .. (318.67,152.33) -- cycle ;
	\draw [color={rgb, 255:red, 74; green, 144; blue, 226 }  ,draw opacity=1 ]   (328.75,140) -- (388.33,156.5) ;
	\draw [color={rgb, 255:red, 74; green, 144; blue, 226 }  ,draw opacity=1 ]   (327.75,165.5) -- (388.5,158.25) ;
	
	\draw (278.5,113.5) node [anchor=north west][inner sep=0.75pt]   [align=left] {$S$};
	\draw (431,70) node [anchor=north west][inner sep=0.75pt]   [align=left] {$V \setminus S$};
	\draw (397.67,159.17) node [anchor=north west][inner sep=0.75pt]  [font=\scriptsize] [align=left] {$v$};
	\draw (290,144.67) node [anchor=north west][inner sep=0.75pt]  [font=\scriptsize] [align=left] {$N(v)$};
	\draw (461.67,108.33) node [anchor=north west][inner sep=0.75pt]   [align=left] {$T$};
	\draw (386,177) node [anchor=north west][inner sep=0.75pt]   [align=left] {$N(S)$};
	\draw (318.08,123.17) node [anchor=north west][inner sep=0.75pt]  [font=\scriptsize,color={rgb, 255:red, 208; green, 2; blue, 27 }  ,opacity=1 ] [align=left] {$c$};
	\draw (446.33,99.33) node [anchor=north west][inner sep=0.75pt]  [font=\scriptsize,color={rgb, 255:red, 208; green, 2; blue, 27 }  ,opacity=1 ] [align=left] {$a$};
	\draw (386,134.33) node [anchor=north west][inner sep=0.75pt]  [font=\scriptsize,color={rgb, 255:red, 208; green, 2; blue, 27 }  ,opacity=1 ] [align=left] {$b$};

\end{tikzpicture}} 
	\caption{This figure shows $S$ and its neighborhood $N(S)$ which intersects with $T$. $T$ has size at most $\textcolor{red}{a}$ and $N(S)-T$ has size at most $\textcolor{red}{b}$. Also, every vertex $v$ in $N(S)-T$ has at most $\textcolor{red}{c}$ neighbors in $S$.  \label{fig:sketch-nbhd-rec}}
\end{figure}

\begin{lemma}\label{lem:sparse-recovery}
	 There is a linear sketch, called $\SNR(G,S)$, for~\Cref{prob:sparse-recovery} that uses sketch and randomness of size, respectively, 
	\[
	\ssnr = \ssnr(n,a,b,c) = O(a \cdot \log{c} + b \cdot \log{n} \cdot \log{c})  \quad \textnormal{and} \quad O(a \cdot \log{n})
	\]
	bits and outputs a wrong answer with probability at most $1-4\exp(-\dfrac{b}{24\log{n}})$ for $b \geq 24\log{n}$. 
\end{lemma}

The key part of~\Cref{lem:sparse-recovery} is that the dependence on $\log{n}$ is only on the (much) smaller $b$-term, as opposed to the $a$-term (otherwise, this result would be immediate by~\Cref{prop:sparse-rec}\footnote{What makes~\Cref{prob:sparse-recovery} particularly different from sparse-recovery is  that since we only know $T$ is a superset of $N(S)$ and not equal to it, our underlying vector is only $(a+b)$-sparse as opposed to $b$-sparse.}). This saving is a key factor in the success of our algorithms in achieving asymptotically optimal bounds for the matching problem. 

We present two algorithms for solving~\Cref{lem:sparse-recovery}. The first one is very simple and already achieves the asymptotic optimal bounds on the sketch size (in terms of parameters $a,b$); the problem with this approach however is that the recovery algorithm for the sketch requires an exhaustive search of all options and thus requires exponential time in the worst case; the sketching matrix of the algorithm also requires $O(a \cdot n)$ bits of space to store which is prohibitively large for our purpose. 
Thus, we present this sketch as a warm-up in~\Cref{app:sparse-recovery}. Our second sketch is more involved but uses a near-linear time recovery algorithm and not too much randomness\footnote{The randomness used by this sketch is still larger than
the sketch size which is problematic on the surface for us. However, we will be able to reuse this randomness across multiple sketches and thus achieve our desired bounds on the space overall.}. We present this algorithm in the remainder of 
this section.

To continue, we give a different representation of~\Cref{prob:sparse-recovery} that makes the exposition simpler. 

\paragraph{Vector-representation of~\Cref{prob:sparse-recovery}.} For any graph $G=(V,E)$ and set $S \subseteq S$ of vertices, define the $n$-dimensional vector $x = x(G,S)$, indexed by vertices in $V$, such that for all $i \in [n]$, 
\begin{align}
	x_i := 
	\begin{cases} 
		0 \quad & \text{if $v_i \in S$} \\ 
		\sum_{u \in S} \mathbb{I}[(u,v_i) \in E] \quad & \text{otherwise}
	\end{cases}. \label{eq:x-GS}
\end{align}
A basic observation is that $\supp{x(G,S)}$ (set of non-zero entries of $x(G,S)$) corresponds to $N(S)$ in $G$. The second observation is that each update to an edge $(u,v)$ in a dynamic stream $\vect(E)$ 
can be directly used to update $x$  as well. Finally, throughout the proof, we will take $q$ to be the smallest prime larger than $c$ and work with the field $\mathbb{F}_q$, i.e., the field of integers modulo $q$. 
Given Promise~\ref{promise3} of~\Cref{prob:sparse-recovery}, $x$ will have the same non-zero entries among the coordinates in $S$ still even in $\mathbb{F}_q$.   
As such, our goal is to design a linear sketch $\Phi$ such that one can recover $\supp{x}$ from $\Phi \cdot x$ (interpreted in $\mathbb{F}_q$)  with high probability.

\newcommand{\IndexR}{\ensuremath{\textnormal{\texttt{Index-Recovery}}}\xspace}
\newcommand{\PartialR}{\ensuremath{\textnormal{\texttt{Partial-Recovery}}}\xspace}

\paragraph{Proof of~\Cref{lem:sparse-recovery}.} The high-level overview of the proof is as follows. We will first design a sketch of size $O(\log{q})$ bits only that can recover the value of $x_i$ for some random $i \in T$. This sketch however 
may fail with constant probability and err with probability roughly $b/a$ introduced by elements in the support of $x$ outside $T$. We will then use $O(a)$ of these sketches in parallel with each other to recover a constant fraction of $x$ projected on $T$
using a sketch of size $O(a\log{q})$. As a result, this effectively shrinks the size of the set $T$ for us that we need to focus next. However, due to the potential error introduced by the sketches, this means that we may now need to recover $O(b)$ \emph{additional} elements from outside of this new $T$. A bit more formally, this approach allows us to find a vector $y$ and shrink the set $T$ to another set $T_y$ such that $\card{T_y} \leq \card{T}/4$ and $x-y$ only has $b+O(b)$ elements 
outside $T_y$. This means that we need to solve the original problem, on the vector $x-y$ now, with a smaller parameter $a$ but a larger parameter $b$. 

Our approach is thus to run this recursive algorithm \emph{non-adaptively} by storing appropriate sketches of $x$ 
only and since we know $y$, use linearity of sketches to compute the sketch of $x-y$ also. The key part of the proof is  to ensure that we shrink size of $T$, i.e., parameter $a$, \emph{rapidly} while grow $b$ \emph{slowly} in this process, all while keeping the 
total sketch size still only $O(a\log{q})$. The slow growth of $b$ thus allows us reach a situation where the resulting vector we have to work with becomes $\Theta(b)$-sparse overall (inside and outside of the current set $T$). We can then use standard sparse recovery on this vector to recover it entirely using another sketch
of size $O(b\log{q} \cdot \log{n})$ which will give us the desired bound on the sketch size.

We now start the formal proof. The first step is a
subroutine for recovering a single index in $T$ with a small error probability. In the rest of the proof, we assume that size of $T$ is exactly $a$ without loss of generality (say, by increasing $n$ slightly and adding dummy elements to $T$).

\begin{lemma}\label{lem:indexR}
	There is a linear sketch, $\IndexR(x,T,a,b,\gamma)$, that given a vector $x \in [n]$ and set $T \subseteq [n]$ and parameters $a \geq 100b$ (as specified in~\Cref{prob:sparse-recovery}), plus a confidence parameter $\gamma \in (0,1/4)$, 
	returns an index $i \in T$ chosen uniformly at random together with the value of $x_i$. The probability of failure of the algorithm, $\delta_F$, and the probability it outputs a wrong value for $x_i$, $\delta_E$, 
	are 
	\[
		\delta_F \leq \frac{4}{5} \qquad \text{and} \qquad \delta_E \leq \gamma \cdot \frac{b}{2a} + \frac{b^2}{2a^2}. 
	\]
	The sketch has size $O(\log{(1/\gamma)} \cdot \log{q})$ bits and requires $O(\log{(1/\gamma)} \cdot \log{n})$ random bits.  
\end{lemma}
\begin{proof}
	Let $h: [n] \rightarrow [2a]$ be a pair-wise independent hash function and define $H := \set{i \in [n] \mid h(i) = 1}$ (so that each index belongs to $H$ with probability $1/2a$ and the choice of vertices is pairwise independent). 
	Compute $z := \sum_{i \in H} x_i$. At the end of the stream, if $\card{H \cap T} \neq 1$, terminate and output FAIL. 
	Moreover, for $j=1$ to $t:= \log{(1/\gamma)}$ iterations: 
	\begin{enumerate}[label=$(\roman*)$]
		\item Pick a pairwise independent hash function $h_j: H \rightarrow \set{1,2}$. For $k \in \{1,2\}$, let $H_{jk} := \set{i \in H \mid h_j(i) = k}$ (so that $H_{j1}$ and $H_{j2}$ form a partition of $H$ with each index in $H$ having sent to each one uniformly). 
		\item For $k \in \set{1,2}$, compute  $z_{jk} := \sum_{i \in H_{jk}} x_i$. If both $z_{jk}$ are non-zero terminate and output FAIL. 
	\end{enumerate}
	If the algorithm never terminated, output $i$ where $\set{i} := H \cap T$ and $x_i = z$ as the answer. This concludes the description of the $\IndexR$ algorithm. 
	
	We now analyze  correctness and  bound  size and randomness of $\IndexR$.	Firstly, 
	\begin{align*}
		\Pr\paren{\card{H \cap T} = 1} = \Pr\paren{\text{there exists a \underline{unique} index $i \in T$ with $h(i) = 1$}} \geq \frac{1}{2} \cdot \paren{1-\frac{1}{2}} = \frac{1}{4},
	\end{align*}
	using the pairwise independence of $h$, exactly as in the proof of~\Cref{lem:nei-edge-sampler}. Note that the algorithm can detect this event exactly as it knows both $H$ and $T$ at the end of the stream. 
	Moreover, by union bound, 
	\begin{align*}
		\Pr\paren{\card{H \cap (N(S) \setminus T)} = 0} = 1-\Pr\paren{\text{there is an index $i \in (N(S) \setminus T)$ with $h(i) = 1$}} \geq 1-\frac{b}{2a} \geq 1-\frac{1}{200}. 
	\end{align*}
	Conditioned on both events above, for  $i \in T$ where $\set{i} = H \cap T$, we get that $z = x_i$. Moreover, for every iteration $j \in [t]$ in the algorithm, it can never be that both $z_{j1}$ and $z_{j2}$ are non-zero. 
	Thus, conditioned on these events, the algorithm outputs a correct answer. This in particular means that the probability the algorithm outputs FAIL is at most $3/4+{1}/{200} < 4/5$ as required. 
	
	Now note that the only way the algorithm may output a wrong answer is when $\card{H \cap T} = 1$ but $\card{H \cap (N(S) \setminus T)} > 0$. Firstly, using that $h$ is a pairwise independent hash function, 
	we have that, 
	\begin{align*}
		\Pr\paren{\card{H \cap (N(S) \setminus T)} \geq 2} \leq {\card{N(S) \setminus T} \choose 2} \cdot \frac{1}{4a^2} \leq \paren{\frac{e \cdot b}{2}}^2 \cdot \frac{1}{4a^2} < \frac{b^2}{2a^2}. \tag{as $\card{N(S) \setminus T} \leq b$}
	\end{align*}
	Thus, whenever this event happens, we can simply charge the error to the $b^2/2a^2$ term in $\delta_E$ in the lemma statement. In the following, we only need to handle the case when $\card{H \cap (N(S) \setminus T)} = 1$. 
	
	For the error to happen 
	in this case, we should have that the single element in $H \cap T$ and the single element in $H \cap N(S) \setminus T$ are always mapped the same by the hash function $h_j$ for every iteration $j \in [t]$. The probability of this event 
	happening is exactly $1/2^t = \gamma$ by the pairwise independence of each $h_j$ (and their independence across $t$ iterations). Thus, conditioned on the event that ${H} \cap N(S) \setminus T$ has size $1$, which itself happens with probability at most $b/2a$ calculated above, the probability of error is $\gamma$. 
	This is also accounted for  in the $\gamma \cdot b/2a$ term in $\delta_E$, concluding the proof of the correctness.
	
	As for the space, storing each of the hash functions requires $O(\log{n})$ and each sum $O(\log{q})$ bits. Given that we have $O(\log{(1/\gamma)})$ iterations, we get the desired space bound. 
\end{proof}

We now build on the sketch $\IndexR$ in~\Cref{lem:indexR} to recover a constant fraction of the indices of $x \in \IR^n$ in $T$ again with small error.  

\begin{lemma}\label{lem:partialR}
	There is a linear sketch, $\PartialR(x,T,a,b,\gamma)$, that given vector $x \in \IR^n$ and set $T \subseteq [n]$ and parameters $a \geq 100b$ (as specified in~\Cref{prob:sparse-recovery}), plus a confidence parameter $\gamma \in (0,1/4)$, 
	returns a vector $y \in \IR^n$ and a set $T_y \subseteq [n]$ with the following properties with probability at least $1-2\exp\paren{-\gamma b}$: 
	\begin{enumerate}[label=$(\roman*)$]
		\item $T_y$ is a subset of $T$ and has size at most $a/4$. 
		\item $x-y$ has at most $b+12 \cdot (\gamma b+\frac{b^2}{a})$ non-zero elements outside $T_y$. 
	\end{enumerate}
	The sketch has size $O(a \cdot \log{(1/\gamma)} \cdot \log{q})$ bits and requires $O(a \cdot \log{(1/\gamma)} \cdot \log{n})$ random bits.  
\end{lemma}
\begin{proof}
	Initialize $T_y \leftarrow T$ and $y \leftarrow \bm{0}^n$. For $t:= (5a\cdot \ln{8}) < 12a$ iterations: Run $\IndexR(x,T,a,b,\gamma)$ independently and, if the output is \emph{not} FAIL, let $y_i = x_i$ for $i \in T$ and $x_i$ returned by the sketch, and
	update $T_y \leftarrow T_y \setminus \set{i}$. This concludes the description of $\PartialR$. 
	 
	Fix any index $i \in T$. In any iteration that $\IndexR$ does not output FAIL, the probability that $i$ is returned as the index is $1/a$ by~\Cref{lem:indexR}. Thus, the probability that this index $i$ is never returned through all the iterations is 
	at most
	\begin{align*}
		\paren{1-\paren{\frac{1}{5} \cdot \frac{1}{a}}}^{t} \leq \exp\paren{-\frac{1}{5} \cdot \frac{t}{a}} = \exp\paren{-\ln{8}} = \frac{1}{8}. \tag{by the choice of $t = (5a \cdot \ln{8})$} 
	\end{align*}
	
	For any $i \in T$, define an indicator random variable $X_i$ which is $1$ iff $i$ still belongs to $T_y$ at the end of $\PartialR$. Let $X = \sum_{i \in T} X_i$ denote the size of $T_y$.
	By the above calculation, we have $\expect{X} \leq a/8$. Moreover, consider the $t$ \emph{independent} random variables denoting the randomness of each iteration of the algorithm. 
	We have that $X$ is a $1$-Lipschitz function of these variables (as changing randomness of one iteration, can only make a single index $i$ join or leave $T_y$). Thus, 
	by McDiarmid's inequality (\Cref{prop:mcdiarmid}), 
	\begin{align*}
	\Pr\paren{\card{T_y} > a/4} \leq \Pr\paren{\card{X - \expect{X}}> a/8} \leq 2 \cdot \exp\paren{-\frac{2 \, a^2}{64 t}} < \exp\paren{-\frac{a}{400}}. \tag{by the choice of $t < 12a$}
	\end{align*}
	
	Now again consider an index $i \in T$. For $x_i \neq y_i$ to happen, we should have that the iteration in which $\IndexR$ returns $i$ makes an error. Thus, the number of indices of $x-y$ that are different (outside $T_y$), 
	is upper bounded by the number of iterations wherein $\IndexR$ errs. Each iteration makes an error with probability at most $\gamma \cdot \frac{b}{2a} + \frac{b^2}{2a^2}$ by~\Cref{lem:indexR}. Thus, letting $Y$ denote the random variable
	for the number of erroneous iterations, we get that 
	\begin{align*}
		\expect{Y} \leq t \cdot \paren{\gamma \cdot \frac{b}{2a} + \frac{b^2}{2a^2}} \leq 6 \cdot (\gamma \cdot b + \frac{b^2}{a}).  \tag{by the choice of $t < 12a$}
	\end{align*}
	Moreover, since $Y$ is a sum of $t$ independent random variables (one per  iteration), by Chernoff bound, 
	\begin{align*}
		\Pr\paren{Y \geq 12 \cdot (\gamma \cdot b + \frac{b^2}{a})} \leq \Pr\paren{Y \geq 2 \expect{Y}} \leq \exp\paren{-{\gamma \cdot b}{}}. 
	\end{align*}
	As there are at most $b$ indices outside $T$ that are non-zero in $x$ (but zero in $y$), we get that $x-y$ has at most $b + 12 \cdot (\gamma \cdot b + \frac{b^2}{a})$ non-zero entries outside $T_y$ 
	with probability at least $1-\exp\paren{-{\gamma \cdot b}{}}$. 
	
	Combining the above two bounds, and since $a \geq 100b$ and $\gamma < 1/4$, we get the final probability bound on the properties of the algorithm. 
	The space and randomness are also $t = O(a)$ times that of $\IndexR$, which implies the lemma by the bounds in~\Cref{lem:indexR}. 
\end{proof}

We now present the sketching matrix of $\SNR$ and postpone its recovery algorithm to later. 

\begin{Algorithm} \label{alg:snr-sketch} 
	The \textbf{sketching matrix} of $\SNR(G,S)$ and its corresponding sketches. 
	\begin{enumerate}[leftmargin=15pt]
		\item Define the recursive sequences $\set{\gamma_j}$, $\set{a_j}$, and $\set{b_j}$ for any $j \geq 1$: 
		\[
			\gamma_1 = \frac{1}{24}, ~ \gamma_{j} := \frac{\gamma_{j-1}}{2} \quad \text{and} \quad a_1 := a  \text{,} ~ a_j := \frac{a_{j-1}}{4} \quad \text{and} \quad b_1 = b  \text{,} ~ b_{j} := b_{j-1} + 12 \cdot (\gamma_{j-1} \cdot b_{j-1} + \frac{b_{j-1}^2}{a_{j-1}}). 
		\]
		\item For $j =1$ to $t$ iterations where $t$ is the \underline{minimum} of $\log\log{n}$ and the \underline{last} index where $a_t \geq 100 \cdot b_t$:
		\begin{enumerate}
			\item Let $\Phi_j$ be the sketching matrix of $\PartialR(\cdot,\cdot,a_j,b_j,\gamma_j)$ (the sketching matrix is independent of first two arguments of the input and thus we do not provide those arguments). 
		\end{enumerate}
		\item Let $\Phi^*$ be the sketching matrix of $\sprec$ to recover a $(a_t+b_t)$-sparse vector in $\mathbb{F}_q^n$. 
		\item Compute the sketches $\set{\Phi_j \cdot x}_{j=1}^{t}$ and $\Phi^* \cdot x$ for the vector representation $x$ of $N(S)$ defined earlier. 
	\end{enumerate}
\end{Algorithm}

We bound the sketch size and randomness of $\SNR$ in the following claim. 
\begin{claim}\label{clm:snr-size}
	The sketch and randomness of $\SNR$ in~\Cref{alg:snr-sketch} have  size 
	\[
	O((a + b \cdot \log{n}) \cdot \log{q}) \quad \text{and} \quad O(a \cdot \log{n})
	\]
	 bits, respectively. 
\end{claim}
\begin{proof}
	By~\Cref{lem:partialR}, we have that $\Phi_j \cdot x$ requires $O(a_j \cdot \log{(1/\gamma_j)} \cdot \log{q})$ size and $O(a_j \cdot \log{(1/\gamma_j)} \cdot \log{n})$ bits of randomness, respectively. 
	Thus, the total size of the sketches $\set{\Phi_j \cdot x}_{j=1}^{t}$ is
	\begin{align*}
		O(\log{q}) \cdot \sum_{j=1}^{t} a_j \cdot \log{(1/\gamma_j)} = O(\log{q}) \cdot \sum_{j=1}^{t} \frac{a_1}{4^j} \cdot \log{(2^j/\gamma_1)} = O(a_1 \cdot \log{(1/\gamma_1)} \cdot \log{q}),
	\end{align*}
	where the second to last equality is by explicitly computing the recursive definition of $a_j,\gamma_j$, and the last one is since the series $\sum_{j=1}^{\infty} {j}/{4^j}$ converges to some constant. Considering that $a_1 = a$ and $\gamma_1$ is a constant,
	we get that the total size of sketches  $\set{\Phi_j \cdot x}_{j=1}^{t}$ is $O(a \cdot \log{q})$. The same exact calculation, by plugging in the randomness bound of~\Cref{lem:partialR} instead, also gives an $O(a \cdot \log{n})$ bound
	on the number of random bits. 
	
	We now need to also calculate the size of $\Phi^*$ (this part of the sketch does not involve any randomness). By~\Cref{prop:sparse-rec}, size of $\Phi^* \cdot x$ is 
	\begin{align}
	O((a_t+b_t) \cdot \log{n} \cdot \log{q}). \label{eq:at-bt}
	\end{align}
	We start by bounding the value of $b_t$. By the recursive definition of $\gamma_j,a_j,b_j$, we have that for all $j \in [t]$
	\begin{align*}
		b_j &= b_{j-1} + 12 \cdot (\gamma_{j-1} \cdot b_{j-1} + \frac{b_{j-1}^2}{a_{j-1}}) = b_{j-1} \cdot \paren{1 + 12 \cdot \frac{\gamma_1}{2^{j-1}} + 12 \cdot \frac{b_{j-1}}{a_{1}}} \\
		&\leq b_{j-1} \cdot \paren{1 + 24 \cdot \frac{\gamma_1}{2^j} + 12 \cdot \frac{b_{t}}{a_{j-1}}}. \tag{as $b_j$'s are increasing}  
	\end{align*}
	As such, for $b_t$ itself, we have that, 
	\begin{align*}
		b_t &\leq b_1 \cdot \prod_{j=1}^{t}\paren{1 + 24 \cdot \frac{\gamma_1}{2^j} + 12 \cdot \frac{b_{t}}{a_{j-1}}} \tag{by the inequality above} \\
		&\leq b_1 \cdot \exp\paren{\sum_{j=1}^{t} 24 \cdot \frac{\gamma_1}{2^j} + 12 \cdot \frac{b_t}{a_{j-1}}} \tag{as $(1+x) \leq \exp(x)$ for all $x > 0$} \\
		&\leq b_1 \cdot \exp\paren{24 \cdot \gamma_1} \cdot \exp\paren{12 b_t \cdot \sum_{j=1}^{t} \frac{1}{a_{j-1}}} \tag{as $\sum_{j=1}^{\infty} 1/2^j = 1$} \\
		&= b_1 \cdot \exp\paren{24 \cdot \gamma_1} \cdot \exp\paren{12\cdot \frac{b_t}{a_t} \cdot \sum_{j=1}^{t} \frac{1}{4^{t-j+1}}} \tag{as $a_{j-1} = a_1/4^{j-1} = 4^{t-j+1} \cdot a_t$} \\
		&\leq b_1 \cdot \exp\paren{24 \cdot \gamma_1} \cdot \exp\paren{4\cdot \frac{b_t}{a_t}} \tag{as $\sum_{j=1}^{t} {1}/{4^{t-j+1}} \leq 1/3$} \\
		&\leq b_1 \cdot  e^2. \tag{as $\gamma_1 = 1/24$ and $a_t \geq 100 b_t$} 
	\end{align*}
	Thus, even though $b_j$'s are increasing, we still have $b_t = O(b_1)$.  
	
	As for the value of $a_t$, we have that $a_t = a_1/4^t$. Thus, either $t = \log\log{n}$ and so we get $a_t \leq a_1/\log{n}$, or $a_{t} < 400 \cdot b_{t+1}$, which implies $a_t = O(b_t) = O(b_1)$ in this case. This implies 
	that $a_t = O(a_1/\log{n} + b_1)$. 
	
	Plugging in the bounds on $a_t,b_t$ in~\Cref{eq:at-bt}, we get that the size for $\Phi^* \cdot x$ is   $O(a \cdot \log{q} + b \cdot \log{n} \cdot \log{q})$. This concludes the proof.  
\end{proof}

We now show how to use these sketches to perform the recovery part. 
\begin{Algorithm} \label{alg:snr-recovery} 
	The \textbf{recovery algorithm} of $\SNR(G,S)$. 
	\begin{enumerate}
		\item Let $x_1 = x$ and $T_1 = T$. For $j=1$ to $t$ iterations (where $t$ is the parameter in~\Cref{alg:snr-sketch}): 
		\begin{enumerate}
			\item Use $\Phi_j \cdot x$ to obtain the vector $y_j$ and set $ T_{j+1}:= T_{y_j}$ of $\PartialR(x_{j},T_j,a_j,b_j,\gamma_j)$ and let $x_{j+1} = x_j - y_j$ (here, and throughout the proof, $a_j,b_j,\gamma_j$ are
			as defined in~\Cref{alg:snr-sketch}). 
			
			This step can be done by exploiting the linearity
			of the sketches to compute 
			\[
			\Phi_j \cdot x_j = \Phi_j \cdot (x_{j-1} - y_{j-1}) = \Phi_j \cdot x - \sum_{k=1}^{j-1} \Phi_j \cdot y_k,
			\] as the algorithm has already calculated $\Phi_j \cdot x$, 
			and knows $y_1,\ldots,y_{j-1}$ explicitly. 
		\end{enumerate} 
			\item Use $\Phi^* \cdot x$ to compute $\Phi^* \cdot x_{t}$ as specified above and recover $x_t$ by $\sprec$. As we have all $y_1,\ldots,y_t$ also, 
			we can compute $x = x_t + \sum_{j=1}^t y_j$ which recovers  $N(S)$ accordingly. 
		
	\end{enumerate}
\end{Algorithm}

The following lemma proves the correctness of the algorithm. 

\begin{claim}\label{clm:snr-inductive}
	With probability $1-4\exp\paren{-\dfrac{b}{24\log{n}}}$, for every $j \in [t]$, $T_j$ has size at most $a_j$ and $x_j$ has at most $b_j$ non-zero entries outside $T_j$. 
\end{claim}
\begin{proof}
	We prove this by induction wherein we assume that the high probability event of~\Cref{lem:partialR} happens every time we invoke it -- we then bound the probability that this does not happen explicitly. 
	
	For $j=1$, the claim statement holds trivially by the promise of~\Cref{prob:sparse-recovery}. Now suppose this is the case for some index $j$. At this point, by~\Cref{lem:partialR}, when invoking $\PartialR(x_{j},T_j,a_j,b_j,\gamma_j)$, 
	we get that the resulting pair of vector $y_j$ and set $T_{y_j}$ have the following properties: 
	\begin{enumerate}[label=$(\roman*)$]
		\item $T_{y_j}$ is a subset of $T_j$ with size at most $a_j/4$. Hence, $T_{j+1} = T_{y_j}$ has size at most $a_{j+1} = a_j/4$ also.  
		\item $x_j-y_j$ has at most $b_j+12(\gamma_j b_j + \frac{b_j^2}{a_j})$ non-zero elements outside $T_{y_j}$. Hence, $x_{j+1} = x_j - y_j$ has at most $b_{j+1} = b_j+12(\gamma_j b_j + \frac{b_j^2}{a_j})$ non-zero elements outside of 
		$T_{j+1} = T_{y_j}$ also. 
	\end{enumerate}
	This proves the induction step. 
	
	We now need to also account for the error probability of each application of~\Cref{lem:partialR}, which, for an iteration $j \in [t]$, is at most $2\exp\paren{-\gamma_j \cdot b_j} $.  
	As such, by union bound, probability of error is at most 
	\begin{align*}
		2 \cdot \sum_{j=1}^{t} \exp\paren{-\gamma_j \cdot b_j} &\leq 2 \cdot \sum_{j=1}^{t} \exp\paren{-\frac{\gamma_1}{2^j} \cdot b_1} \tag{as $\gamma_j = \gamma_1/2^j$ and $b_j \geq b_1$} \\
		&\leq 4 \cdot \exp\paren{-\frac{\gamma_1}{2^t} \cdot b_1} \tag{as the largest term in the series dominates the sum of the rest} \\
		&\leq 4 \cdot \paren{-\frac{b}{24 \cdot \log{n}}},
	\end{align*}
	by the choice of $\gamma_1$ and since $t \leq \log\log{n}$. 
\end{proof}

Conditioned on the event of~\Cref{clm:snr-inductive}, we have that at the end of the last iteration, $x_t$ is $(a_t+b_t)$-sparse. Thus, by the guarantee of $\sprec$ in~\Cref{prop:sparse-rec}, the algorithm 
correctly recovers $x_t$. As $x = x_t + \sum_{j=1}^t y_j$ and the algorithm has already computed $y_j$'s also, it will recover $x$ correctly. This, combined with the bounds on the sketch size and randomness in~\Cref{clm:snr-size}, 
concludes the proof of~\Cref{lem:sparse-recovery}.

\begin{remark}
	The number of random bits needed by the $\SNR$ sketch is $O(a\log n)$ which is more than our budget to store individually for each sketch. But we can reuse these random bits for all copies of $\SNR$ sketches and union bound over the failure probability. Thus, we use and store at most $O(a\log n)$ random bits over \emph{all} $\SNR$ sketches that we use in our dynamic streaming algorithm for matching. 
\end{remark}

\subsection{Neighborhood-Size Tester} \label{sec:nbhd-size-test}

Before we move on from this section, we also mention a simple helper sketch that allows to approximately verify if the promises of~\Cref{prob:sparse-recovery} are satisfied for a given input. We formally define the problem as follows:

\begin{problem}[Neighborhood-Size Testing]\label{prob:nbhd-test}
	Let $a,\inside \geq 1$ be known integers such that $a \geq 16\inside$. Consider a graph $G=(V,E)$ specified in a dynamic stream and let $S \subseteq V$ be a known subset of vertices. The goal is to, given a set $T \subseteq V$ at the \underline{end} of the stream, 
	return ``Yes" if $\card{N(S) - T}\leq \inside$ and ``No" if $\card{N(S) - T}\geq 2 \inside$, assuming the following promises: 
	\begin{enumerate}[label=$(\roman*)$]
		\item size of $T$ is at most $a$;
		\item $\card{N(S) - T} \leq \inside$ or $\card{N(S) - T} \geq 2\, \inside$;
	\end{enumerate}
\end{problem}
In words, in \Cref{prob:nbhd-test}, we have a set $S$ of vertices, known at the \emph{start} of the stream, and we are interested in the size of the neighborhood \emph{outside} a given set $T$, specified at the \emph{end} of the stream. 
We guarantee that $T$ is not ``too large'' (parameter $a$) and that $a$ is slightly larger than $\inside$ and want to know the size of the neighborhood of $S$ outside $T$. If the size is between $\inside$ and $2 \inside$ then the answer can be arbitrary.

\begin{lemma}\label{lem:nbhd-size-test}
	There is a linear sketch, called $\NET(G,S)$, for~\Cref{prob:nbhd-test} that uses sketch and randomness of size
	\[
	\snet = \snet(n,a,\inside) = O\left(\frac{a}{\inside} \cdot \log^3 n\right)
	\]
	 bits and with high probability outputs the correct answer. 
\end{lemma}

The solution to this problem is standard and is included for completeness. We will solve this problem by sampling random neighbors (using $\Lsampler$) and see how many of them lie outside $T$. Note that all vertices in $T$ may not be neighbor to $S$, but we can fix that by adding artificial edges from all vertices of $T$ to $S$ only for the tester. We can then count the number of neighbors picked by the neighborhood samplers outside $T$ and get an estimate of the number of neighbors outside $T$. Formally, 
\begin{Algorithm}
	$\NET(G,S)$: A  linear sketch for~\Cref{prob:nbhd-test}. 
	\medskip
	
	\textbf{Input:} A graph $G=(V,E)$ specified via $\vect(E)$, a set $S \subseteq V$ specified at the beginning of the stream, and a set $T \subseteq V$ specified at the end of the stream.
	
	\medskip
	
	\textbf{Output:} ``Yes" if $\card{N(S)-T} \leq \inside$ and  ``No" if $\card{N(S)-T} \geq 2 \inside$.
	
	\medskip
	
	\textbf{Parameters:} Let $k := \frac{100}{99} \cdot 150 \ln{n} \cdot \left(\frac{a+\inside}{\inside}\right)$.
	
	\medskip
	
	\textbf{Sketching matrix:} 
	
	\begin{enumerate}
		\item Given $T$ at the end of the stream, add edges from all vertices in $T$ to an arbitrary vertex in $S$ to the stream for this tester.
		\item Sample $k$ copies of $\Lsampler$ for $N(S)$ with parameters $\delta_E=n^{-10}$ and $\delta_F=1/100$.
	\end{enumerate}
	
	\textbf{Recovery:} 
	
	\begin{enumerate}
		\item Go over the $k$ copies of $\Lsampler$ that do not fail and extract a vertex from them.
		\item If the number of sampled vertices outside $T$ is at most $200 \log n$, output ``Yes"; otherwise ``No". 
	\end{enumerate}
\end{Algorithm}

We now analyze the algorithm. 
Let $Y:=N(S) \setminus T$ and $y :=\card{Y}$. We consider the cases when $y \leq \inside$ and $y \geq 2\inside$ separately in the following. 

\begin{claim}
	If $y \leq \inside$ then the algorithm outputs ``Yes'' correctly with probability  at least $1-n^{-4}$. 
\end{claim}
\begin{proof}
	A sample lies outside $T$ with probability $\frac{y}{a+y} \leq \frac{\inside}{a+\inside}$ and does not fail with probability $0.99$ by the guarantee of $\Lsampler$. For any sample $i \in [k]$, 
	let $Z_i$ be an indicator random variable which is $1$ iff the sampled vertex is outside $T$. We thus have 
	\[
	\expect{Z_i} \leq 0.99 \cdot \frac{\inside}{a+\inside} = \frac{150 \log{n}}{k}. \tag{by the choice of $k$ in the algorithm}
	\]
	Let $Z := \sum_{i=1}^{k} Z_i$ denote the number of sampled vertices outside of $T$. As $Z_i$'s are independent of each other, by Chernoff bound (\Cref{prop:chernoff}) with $\mu_H = 150\,\log{n} \geq \expect{Z}$ and $\eps=1/3$, 
	we have, 
	\[
		\Pr\paren{Z > 200\log{n}} = \Pr\paren{Z > (1+\eps) \cdot \mu_H} \leq \exp\paren{-\frac{150\ln{n}}{36}} < n^{-4},
	\]
which concludes the proof, as when the number of sampled vertices that lie outside $T$ is at most $200 \log n$, the algorithm outputs ``Yes''. 
\end{proof}

We now consider the complementary case when $y \geq 2\inside$. 

\begin{claim}
	If $y \geq 2 \inside$ then the algorithm outputs ``No'' with probability at least $1-n^{-3}$.
\end{claim}
\begin{proof}
A sample lies outside $T$ with probability $\frac{y}{a+y} \geq \frac{2\inside}{a+2\inside}$ and does not fail with probability $0.99$ by the guarantee of $\Lsampler$. For any sample $i \in [k]$, 
	let $Z_i$ be an indicator random variable which is $1$ iff the sampled vertex is outside $T$. We thus have 
	\[
	\expect{Z_i} \geq 0.99 \cdot \frac{2\inside}{a+2\inside} \geq 0.99 \cdot \frac{2\inside}{a+\inside} \cdot \frac{7}{8} \geq \frac{250 \log{n}}{k}. \tag{by the choice of $k$ in the algorithm and since $a \geq 16\inside$}
	\]
	Let $Z := \sum_{i=1}^{k} Z_i$ denote the number of sampled vertices outside of $T$. As $Z_i$'s are independent of each other, by Chernoff bound (\Cref{prop:chernoff}) with $\mu_L = 250\,\log{n} \geq \expect{Z}$ and $\eps=1/5$, 
	we have, 
	\[
		\Pr\paren{Z \leq 200\log{n}} = \Pr\paren{Z \leq (1-\eps) \cdot \mu_H} \leq \exp\paren{-\frac{150\ln{n}}{50}} < n^{-3},
	\]
which concludes the proof, as when the number of sampled vertices that lie outside $T$ is more than $200 \log n$, the algorithm outputs ``No''. 
\end{proof}

Therefore, we showed that we can distinguish between the two cases. To find the error probability we can union bound over the error probabilities of both cases and the error probabilities of all copies of $\Lsampler$ and conclude that the error probability is at most $2n^{-3}$ (since there are at most $n^2$ copies of $\Lsampler$).

The algorithm uses $k$ copies of $\Lsampler$ each of which has size
$O( \log^2 n)$ bits (\Cref{prop:l0-sampler} with $\delta_E=n^{-10}$ and $\delta_F=1/100$).
Thus, the sketch size is $k \cdot O(\log^2 n) = O\left(\dfrac{a}{\inside} \log^3 n\right)$ bits, proving \Cref{lem:nbhd-size-test}.


\newcommand{\guess}{\ensuremath{{\textnormal{opt}}}\xspace}

\newcommand{\LL}{\ensuremath{\mathcal{L}}}
\renewcommand{\RR}{\ensuremath{\mathcal{R}}}

\renewcommand{\MM}{\ensuremath{\mathcal{M}}}

\newcommand{\EE}{\ensuremath{\mathcal{E}}}

\newcommand{\EErep}{\EE_{\textnormal{rep}}}
\newcommand{\EEsrep}{\EE_{\textnormal{str-rep}}}

\newcommand{\grem}{G_{\textnormal{rem}}}
\newcommand{\gspa}{G_{\textnormal{sparse}}}
\newcommand{\mgred}{M_{\textnormal{greedy}}}
\newcommand{\mrem}{M_{\textnormal{rem}}}
\newcommand{\vrem}{V_{\textnormal{rem}}}
\newcommand{\mstr}{M_{\textnormal{strong}}}

\newcommand{\Mone}{\ensuremath{M_{\textnormal{\textsc{alg-1}}}}}
\newcommand{\Mtwo}{\ensuremath{M_{\textnormal{\textsc{alg-2}}}}}
\newcommand{\Mthree}{\ensuremath{M_{\textnormal{\textsc{alg-3}}}}}

\newcommand{\Distgroup}{\ensuremath{\mathbb{D}_{\textnormal{grp}}}\xspace}

\newcommand{\Measy}{\ensuremath{M_{\textnormal{easy}}}\xspace}
\newcommand{\Geasy}{\ensuremath{\overline{G}_{\textnormal{easy}}}\xspace}

\newcommand{\ME}[1]{\ensuremath{M_{#1}}\xspace}
\newcommand{\GE}[1]{\ensuremath{\overline{G}_{#1}}\xspace}

\newcommand{\High}[1]{\ensuremath{\textnormal{High}_{#1}}\xspace}
\newcommand{\Med}[1]{\ensuremath{\textnormal{Med}_{#1}}\xspace}
\newcommand{\Low}[1]{\ensuremath{\textnormal{Low}_{#1}}\xspace}

\newcommand{\dlow}[2]{\ensuremath{\textnormal{deg}_{\textnormal{low}}(#1,#2)}\xspace}
\newcommand{\degG}[2]{\ensuremath{\textnormal{deg}_{#2}(#1)}\xspace}

\newcommand{\eventmatch}{\ensuremath{\mathcal{E}_{\textnormal{M}}}\xspace}
\newcommand{\eventsparsify}{\ensuremath{\mathcal{E}_{\textnormal{S}}}\xspace}

\newcommand{\rand}{\ensuremath{Y}}

\newcommand{\com}{\ensuremath{\textnormal{com}}\xspace}
\newcommand{\Com}{\ensuremath{\textnormal{Com}}\xspace}

\newcommand{\NL}{\ensuremath{N\!L}}

\section{Main Result and Setup}\label{sec:alg}

In this section, we present our main results for $\alpha$-approximating the maximum matching of any given graph in dynamic streams using $O(n^2/\alpha^3)$ bits of space. 
Specifically, we prove the following theorem for linear sketches which immediately gives a dynamic streaming algorithm with the same guarantees by~\Cref{prop:linear-sketch}. 

\begin{theorem}\label{thm:main}
	There is a linear sketch that given any parameter $\alpha \leq n^{1/2-\delta}$ for any constant $\delta > 0$, and any  $n$-vertex graph $G=(V,E)$ specified via $\vect(E)$, 
	with high probability outputs an $\alpha$-approximate maximum matching of $G$ using $O(n^2/\alpha^3)$ \underline{\emph{bits}} of space. 
\end{theorem}

We will make the following (more or less standard) assumptions when designing our algorithms. Both assumptions are made for simplicity of exposition and we show how to remove them 
later in this section.

\begin{assumption}[Knowledge of matching size]\label{assumption:guess}
	At the beginning of the stream, we are given an estimate $\guess$ with the promise that the maximum matching size of the input graph $G$ has size \underline{at least} $\guess$. 	The goal 
	is then to return a matching of size  $(\eta_0 \cdot \opt/\alpha)$ for some absolute constant $\eta_0 > 0$ at the end of the stream. 
\end{assumption}

\begin{assumption}[Range of parameters]\label{assumption:range}
	We assume that the parameter $\opt$ of~\Cref{assumption:guess} and approximation factor $\alpha$ satisfy the following equations: 
	\[
		\opt \geq \alpha^2 \cdot n^{\delta} \qquad \textnormal{and} \qquad \alpha > 100.
	\]
\end{assumption}

\begin{remark}\label{rem:assumption}
	While we assume~\Cref{assumption:guess} and~\Cref{assumption:range} when designing our algorithms, even if these assumptions are \emph{not} satisfied, the algorithms (with high probability) will not output an edge that does not belong to the graph, 
	but may output a matching that is not sufficiently large for our purpose. 
\end{remark}

The plan for designing our main algorithms is then to focus on the problem of~\Cref{assumption:guess} (and further assume~\Cref{assumption:range}). 
We first give a linear sketch that can handle ``easy'' graphs for this problem. In particular, we prove the following lemma. 

\begin{Lemma}[Match-Or-Sparsify Lemma]\label{lem:ms}
	There is a linear sketch that given any graph $G=(V,E)$ specified via $\vect(E)$, uses $O(\guess^2/\alpha^3)$ bits of space and with high probability outputs a matching $\Measy$ 
	that satisfies \underline{at least one} of the following conditions:  
	\begin{itemize}
		\item \textbf{Match-case:} The matching $\Measy$ has at least $(\opt/8\alpha)$ edges; 
		\item \textbf{Sparsify-case:} The induced subgraph of $G$ on vertices not matched by $\Measy$, denoted by $\Geasy$,  has at most $(20\,\opt \cdot \log^4\!{n})$ edges and a matching of size at least $3\,\opt/4$. 
	\end{itemize} 
\end{Lemma}

This lemma should be interpreted as follows: we can either find a matching of size $(\opt/8\alpha)$ (thus already solve the problem of~\Cref{assumption:guess} with $\eta_0 = 1/8$), or 
certify that we had a ``hard'' graph to work on. Our main saving in the space then comes from the subsequent algorithm that handles any input that leads to the sparsify-case of~\Cref{lem:ms}. 
We prove~\Cref{lem:ms} in~\Cref{sec:ms}. 

We note that~\Cref{lem:ms} bears  similarities to the so-called ``residual sparsity property'' of greedy matching established in~\cite{AhnCGMW15} (see also~\cite{Konrad18}). In this context, those results prove that 
if one samples $\approx \guess^2/\alpha^3$ edges of the graph uniformly at random, and compute a maximal matching of the sample greedily, then the induced subgraph of $G$ on unmatched vertices have 
maximum degree $\approx \alpha^3 \log{n}$ with high probability, thus $O(\alpha^3 \cdot n\log{n})$ edges in total. Our~\Cref{lem:ms} uses a non-uniform sampling method and exploits the fact that the resulting matching is small (otherwise we are in the match-case), to bound the total number of edges in the induced subgraph of unmatched vertices more strongly by $\approx \opt \cdot \poly\!\log{(n)}$. Finally, the non-uniform sampling method used in this lemma is inspired by prior 
work on dynamic streaming matching algorithms in~\cite{AssadiKLY16,ChitnisCEHMMV16} although the analysis of the algorithm is quite different. 

The following lemma is the heart of the proof. We emphasize that the information provided by algorithm of~\Cref{lem:ms} will only be available to the algorithm of this lemma \emph{at the end of the stream} as we have to run both algorithms in parallel in a single pass.

\begin{Lemma}[Algorithm for Sparsify-Case]\label{lem:sc}
	There is a linear sketch that given any graph $G=(V,E)$ specified via $\vect(E)$, uses $O(\opt^2/\alpha^3)$ bits of space and with high probability, given the matching $\Measy$ of~\Cref{lem:ms} in the recovery step, can recover a matching of size $(\opt/8\alpha)$ in $G$. 
\end{Lemma}

\Cref{lem:sc} gives an efficient way of solving ``hard instances'' of the dynamic streaming matching problem, namely, the ones left out by our more standard approach in~\Cref{lem:ms}. This lemma is where we use our $\SNR$ sketches in place 
of $L_0$-samplers and is the source of 
efficiency of our general algorithm. We prove~\Cref{lem:sc} in~\Cref{sec:sc}. 

\Cref{thm:main} now follows easily from~\Cref{lem:ms} and~\Cref{lem:sc}  by lifting~\Cref{assumption:guess} and~\Cref{assumption:range}.  

\begin{proof}[Proof of~\Cref{thm:main}]
	By~\Cref{lem:ms} and~\Cref{lem:sc}, we can use $O(\opt^2/\alpha^3)$ bits of space under~\Cref{assumption:guess} and~\Cref{assumption:range} and find a matching of size at least $(\opt/8\alpha)$ in $G$ with high probability.

	\paragraph{Removing~\Cref{assumption:range}.} Firstly, if $\opt < \alpha^2 \cdot n^{\delta}$, then by the promise of~\Cref{thm:main} that $\alpha < n^{1/2-\delta}$, we get that $\opt < n^{1-\delta}$. 
	At this point, even if we run an algorithm with $O((\opt^2/\alpha^3) \cdot \poly\log{(n)})$ space, it will still be $o(n^2/\alpha^3)$ bits as required by~\Cref{thm:main}. Thus, we can run any of the previously-best algorithms 
	for this problem, e.g. the ones in~\cite{AssadiKLY16,ChitnisCEHMMV16}, to solve the problem.
	
	Secondly, if $\alpha \leq 100$, we can simply maintain a counter mod two between every pairs of vertices to store all edges of $G$ in $O(n^2)$ bits of space which is permitted by~\Cref{thm:main} when $\alpha = O(1)$. 
	This allows us to solve the problem exactly.  
	
	\paragraph{Removing~\Cref{assumption:guess}.} Let $\alg(\opt,\alpha)$ be the algorithm we obtained so far under~\Cref{assumption:guess}. We simply run $\alg(o,\beta)$ for all choices of $o \in \set{2^i \mid \text{$i=0$ to $\log{n}$}}$ and 
	$\beta = (\alpha/2\eta_0)$ in parallel 
	and return the largest matching found. By~\Cref{rem:assumption}, these matchings all belong to the input graph with high probability and for the choice of $o \geq \mu(G)/2$, where $\mu(G)$ is the maximum matching size of $G$, 
	we can apply our results for $\opt=o$ to get a matching of size $\eta_0 \cdot \opt/(\beta/2) = \mu(G)/\alpha$ in the graph, which is precisely an $\alpha$-approximation as desired. Finally, the space of this new algorithm is
	\[
		O(1) \cdot \sum_{\substack{o \in \set{2^{i} \mid i \in [\log{n}]}}} \frac{o^2}{\alpha^3} = O(n^2/\alpha^3)~\text{bits}
	\]
	as the sum is forming a geometric series. This concludes the proof of~\Cref{thm:main}. 
\end{proof}

We conclude this section by making the following remark about the sketches we use. 

\begin{remark}\label{rem:l0-err}
Throughout our main algorithms in the remainder of the paper, we use at most $O(n^2)$ copies of the sketching primitives $\NET,\NES$ and $\SNR$ developed in~\Cref{sec:tools}. For all these sketches the probabilities of failure and error are
$1/100$ and $n^{-10}$, respectively.  We can simply do a union bound over all these sketches and have that with a high probability, none of 
them are going to err. Hence, we condition on this high-probability event here and do not explicitly account for the error probability of this part each time. However, we will consider the case that (some of) these sketches may output FAIL still. 
\end{remark}


\newcommand{\VV}{\ensuremath{\mathcal{V}}}

\newcommand{\HR}{\ensuremath{H_{\textnormal{rec}}}}

\newcommand{\Filter}{\ensuremath{{{F}}}\xspace}

\newcommand{\Mstar}{\ensuremath{M^\star}}

\newcommand{\WR}{\ensuremath{W\!R}}
\newcommand{\SR}{\ensuremath{S\!R}}
\newcommand{\NR}{\ensuremath{N\!R}}

\newcommand{\Mhard}{\ensuremath{M_{\textnormal{hard}}}}

\newcommand{\outside}{\tilde{b}}

\section{Main Algorithm: Handling the Sparsify-Case}\label{sec:sc} 

As the main part of our work in this paper is the algorithm in~\Cref{lem:sc}, we change the order of presentation and start with this algorithm and postpone the proof of~\Cref{lem:ms} to the next section. 

\Cref{lem:ms} allows us to find a matching $\Measy$ which is either large enough, or the subgraph induced by its unmatched vertices is sparse and has a large matching. Our task now is to handle 
the latter case efficiently, i.e., prove~\Cref{lem:sc}. We emphasize that we can only know this particular sparse subgraph of the input \emph{after} the pass over the input, and by that point we should have collected all the required information from the graph already. The following lemma is a slightly weaker version of~\Cref{lem:sc}. 

\begin{lemma}[Slightly weaker version of~\Cref{lem:sc}]\label{lem:sc2}
	There is a linear sketch that given any graph $G=(V,E)$ specified via $\vect(E)$, with high probability uses $O(\opt^2/\alpha^3)$ bits of space and given the matching $\Measy$ of~\Cref{lem:ms} in the recovery step that satisfies the sparsify-case, 
	can recover a matching of size $(\opt/8\alpha)$ in $G$ with probability at least $(1-n^{-\delta/6})$ and does not output any edge that is not in $G$ with high probability. 
\end{lemma}

Let us show that~\Cref{lem:sc2} immediately proves~\Cref{lem:sc} in its full generality. Firstly, it is without loss of generality to assume that $\Measy$ satisfies the sparsify-case as otherwise, the algorithm 
can simply return $\Measy$ itself which is of size at least $(\opt/8\alpha)$ in the match-case and satisfies the promise of~\Cref{lem:sc}. 

Secondly, to improve the success probability to a high-probability bound, we can run the algorithm of above lemma in parallel for $(60/\delta) = O(1)$ times
 and return the largest matching output by any copy. With high probability, the algorithm still does not 
output an edge not in the graph and uses $O(\opt^2/\alpha^3)$ bits of space (as $\delta=\Theta(1)$). The probability that none of these matchings are large enough is only $\paren{n^{-\delta/6}}^{60/\delta} = n^{-10}$; thus, the algorithm
also outputs a large enough matching with high probability. This proves~\Cref{lem:sc} assuming~\Cref{lem:sc2}. As such, in this section, 
we focus on proving~\Cref{lem:sc2}. 

To simplify the exposition, we present and analyze the sketching matrix and recovery step of the linear sketch in~\Cref{lem:sc2} separately. 

\subsection{The Sketching Matrix}

The sketching matrix of~\Cref{lem:sc2} is computed as follows. We create $k \approx \opt/\alpha$ groups of vertices and each group is obtained by sampling each vertex independently with probability $1/k$ (so vertices can belong to more than one group or none at all). We connect these groups using a fixed $(k/\alpha)$-regular graph and throughout the stream, we only focus on the edges appearing between vertices of connected groups. Over these edges then, 
we maintain one $\NET$ and one $\SNR$ for each group with parameters $a \approx (\opt/\alpha^2)$, $b \approx \outside \approx n^{\delta}$, and $c = \Theta(1)$. Moreover, to save space in the sketching matrices, we use the same 
set of random bits for sketching matrices of all $\SNR$ copies. This amounts to a total of $O(\opt^2/\alpha^3)$ bits of space. 

\begin{Algorithm}\label{alg:sc} 
	The sketching matrix  of~\Cref{lem:sc2}.
	
	\medskip
	
	\textbf{Input:} A graph $G=(V,E)$ specified via $\vect(E)$. 
	
	\medskip
	
	\textbf{Parameters:} Let $k := {10\opt}/{\alpha}$, $a = 2\opt/\alpha^2$, $\outside =n^{\delta/4}$, and $c=30/\delta$. 
	
	\begin{enumerate}[label=$(\roman*)$]
		\item\label{sc1} Create a collection of \textbf{groups} of vertices $\VV := (V_1,\ldots,V_k)$ as follows: For each $i \in [k]$, independently sample a $(\log^2\!{n})$-wise independent hash function $h_i: V \rightarrow [k]$ and set $V_i := \set{v \in V \mid h_i(v) = 1}$. 
		\item\label{sc2} Let $\Filter \in \set{0,1}^{k \times k}$ be the adjacency matrix of an arbitrarily fixed $(k/\alpha)$-regular graph on $[k]$ (with no self-loops or parallel edges). We say that two groups $V_i$ and $V_j$ are \textbf{neighbor}  
			whenever $\Filter(i,j) = 1$.  
		\item\label{sc3} For any group $i \in [k]$, define the subgraph $G(V_i)$ on vertices $V$ but only consisting of edges between $V_i$ and its neighbor-groups, i.e., with edges $\set{(u,v) \in E \mid u \in V_i,\, v \in V_j,\, \Filter(i,j)=1, \; \forall j}$. 
		\item\label{sc4} For every $i \in [k]$, return sketching matrices of $\SNR(G(V_i),V_i)$ with parameters $a,b=2 \outside,c$ and $\NET(G(V_i),V_i)$ with parameters $a$ and $\outside$ as the final sketching matrix -- to save space, use the \emph{same random bits} 
		for  sketching matrices of all copies of $\SNR$.  
	\end{enumerate}
\end{Algorithm}

We note that in~\Cref{alg:sc}, each vertex of $V$ may appear in more than one group of $\VV$ or no group at all. 
We start by bounding the size of the sketching matrix and the extra information stored by~\Cref{alg:sc}. 
\begin{lemma}\label{lem:sc-space}
	\Cref{alg:sc} uses $O(\opt^2/\alpha^3)$ bits of space with high probability. 
\end{lemma}
\begin{proof}
	Line~\ref{sc1} requires storing $k = O(\opt/\alpha)$ different $(\log^2{n})$-wise independent hash functions, each of which requiring $O(\log^3{n})$ bits by~\Cref{prop:k-wise}. This is bounded by $O(\opt^2/\alpha^3)$ bits by~\Cref{assumption:range}. 
	
	Line~\ref{sc2} does \emph{not} require storing $\Filter$ explicitly as it is  fixed and input-independent (we can use any standard way of generating a fixed $(k/\alpha)$-regular graph\footnote{For instance, to generate a $2d$-regular graph on $N$ vertices connect vertex $i$ to vertices in $[i-1,i-d]$ and in $[i+1,i+d]$. To generate a $2d+1$-regular graph connect $i$ to $i+N/2$ in addition to the previous vertices. All the calculations are done mod $N$. Note that $N$ has to be even in the $2d+1$ case which is okay for us because we have $10 \cdot (\opt/\alpha)$ groups.}). 
	
	Line~\ref{sc3} and the graphs it works with are deterministic functions of $\vect(E)$ in the stream and the groups stored in Line~\ref{sc1}. Note that we are \emph{not} going to store these subgraphs in the stream but rather for each update $(u,v)$ to $\vect(E)$, 	we only update all subgraphs $G(V_i)$ for $i \in [k]$ by checking whether $(u,v)$ also belongs to $G(V_i)$ for $i \in [k]$. Thus, we require no further storage in this line. 
	
	Line~\ref{sc4} stores sketching matrix of $k$ copies of $\SNR$ with the same parameters $a,b=2\outside,c$ and $k$ copies of $\NET$ with the same parameters $a,\outside$. By~\Cref{lem:sparse-recovery}, each $\SNR$ sketch will take $O((a+\outside\log{n})\log{c})$ bits 
	and by~\Cref{lem:nbhd-size-test}, each $\NET$ takes $O(\frac{a}{\outside} \cdot \log^3{n})$ bits. Both of these are $O(\opt/\alpha^2)$ bits by~\Cref{assumption:range}. As we are storing $k = O(\opt/\alpha)$ of these sketches, 
	the total space will then be $O(\opt^2/\alpha^3)$ as desired. Finally, since we share the randomness of copies of $\SNR$, we only need $O(a\log{n})$ bits \emph{in total} which is a lower-order term. \Qed{\Cref{lem:sc-space}}
	
\end{proof}

By the sparsify-case, we get a sparse graph $\Geasy$ with a large matching.
We identify edges of this matching with certain properties that make them easy to recover while accounting for a constant fraction of the matching. 
In our subsequent recovery algorithm we will show that we recover a superset of these edges.
We now analyze \Cref{alg:sc} and describe the useful properties of certain edges. To continue, we need some notation and definitions.

\paragraph{Notation.} We say an edge $e=(u,v)$ \textbf{appears between} two groups $V_i,V_j \in \VV$ iff $u \in V_i$ and $v \in V_j$, and  $V_i$ and $V_j$ are neighbor groups, i.e., $\Filter(i,j) = 1$. 
Similarly, we say $e$ \textbf{appears inside} a group $V_i \in \VV$ if there exists some group $V_j$ such that $e$ appears between $(V_i,V_j)$. We write `$e \in (V_i,V_j)$' or `$e \in V_i$' when 
$e$ appears between $(V_i,V_j)$ or inside $V_i$, respectively.

\begin{definition}[Group definitions]\label{def:group-def}
For each group $V_i \in \VV$, we say that $V_i$ is: 
\begin{itemize}
	\item[$-$] \emph{\textbf{clean}} if it does not contain any vertex of $\Measy$. 
	\item[$-$] \emph{\textbf{expanding}} if more than $\outside$ edges of $\Geasy$ appear inside $V_i$ and \emph{\textbf{non-expanding}} otherwise. 
\end{itemize}
\end{definition}

\noindent
Let $M$ be the matching of size at least $3\,\opt/4$ in $\Geasy$ as guaranteed by~\Cref{lem:ms} (recall that $\Geasy$ is the subgraph of $G$ induced on vertices not matched by $\Measy$).
We define $\Mstar$ as the following subset of $M$ on ``low-degree'' vertices of $\Geasy$, namely: 
\begin{align}
	\Mstar := \set{(u,v) \in M \mid \text{each of $u$ and $v$ has at most $(\outside/2)$ neighbors in $\Geasy$}}.  \label{eq:M*}
\end{align}
We will focus on recovering edges of $\Mstar$ (which we show are sufficiently many). 
For this, we need to define several conditions for each edge $(u,v) \in \Mstar$ that if satisfied, allows us to recover this edge via our recovery algorithm using the sketches stored by~\Cref{alg:sc}.

\begin{definition}[$\Mstar$-edges definitions]\label{def:Mstar-def}
For any edge $e$ of $\Mstar$, we say that $e$ is:  
\begin{itemize}
	\item[$-$] \emph{\textbf{weakly-represented}} by pairs of groups $V_i \neq V_j \in \VV$ iff: 
	\begin{enumerate}[label=$(\roman*)$]
	\item $e$ appears between $V_i$ and $V_j$ (this means $V_i$ and $V_j$ has to be neighbor groups),
	\item no edge of $\Geasy$ other than $e$ appears between $V_i$ and $V_j$, and
	\item both $V_i$ and $V_j$ are clean. 
	\end{enumerate}
	\item[$-$] \emph{\textbf{strongly-represented}} by pairs of groups $V_i \neq V_j \in \VV$ iff:
		\begin{enumerate}[label=$(\roman*)$]
	\item $e$ is weakly-represented by $(V_i,V_j)$, and
	\item both of $V_i$ and $V_j$ are non-expanding.
	\end{enumerate}
\end{itemize}
\end{definition}
\noindent
\Cref{fig2,fig3} give illustrations of this definition. 

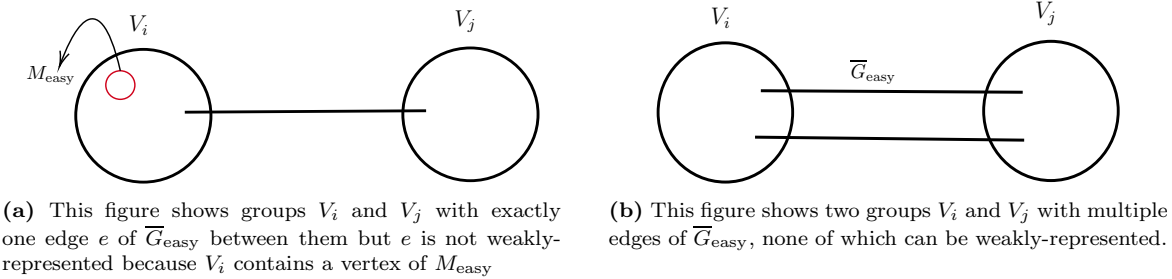
\begin{figure}[h!]
	\centering
	\subcaptionbox{This figure shows groups $V_i$ and $V_j$ with exactly one edge $e$ of $\Geasy$ between them but $e$ is not weakly-represented because $V_i$ contains a vertex of $\Measy$ \label{}}%
	[.45\linewidth]{ 	\resizebox{200pt}{80pt}{

\tikzset{every picture/.style={line width=0.75pt}} 

\begin{tikzpicture}[x=0.75pt,y=0.75pt,yscale=-1,xscale=1]
	
	\draw  [line width=1.5]  (100,164.5) .. controls (100,138.82) and (120.82,118) .. (146.5,118) .. controls (172.18,118) and (193,138.82) .. (193,164.5) .. controls (193,190.18) and (172.18,211) .. (146.5,211) .. controls (120.82,211) and (100,190.18) .. (100,164.5) -- cycle ;
	\draw  [line width=1.5]  (325,163.5) .. controls (325,137.82) and (345.82,117) .. (371.5,117) .. controls (397.18,117) and (418,137.82) .. (418,163.5) .. controls (418,189.18) and (397.18,210) .. (371.5,210) .. controls (345.82,210) and (325,189.18) .. (325,163.5) -- cycle ;
	\draw [line width=1.5]    (175,162) -- (341,161) ;
	\draw  [color={rgb, 255:red, 208; green, 2; blue, 27 }  ,draw opacity=1 ] (120.67,143) .. controls (120.67,137.48) and (125.14,133) .. (130.67,133) .. controls (136.19,133) and (140.67,137.48) .. (140.67,143) .. controls (140.67,148.52) and (136.19,153) .. (130.67,153) .. controls (125.14,153) and (120.67,148.52) .. (120.67,143) -- cycle ;
	\draw    (130.67,133) .. controls (115.15,64.45) and (95.86,110.38) .. (89.24,128.72) ;
	\draw [shift={(88.67,130.33)}, rotate = 289.44] [color={rgb, 255:red, 0; green, 0; blue, 0 }  ][line width=0.75]    (10.93,-3.29) .. controls (6.95,-1.4) and (3.31,-0.3) .. (0,0) .. controls (3.31,0.3) and (6.95,1.4) .. (10.93,3.29)   ;
	
	\draw (135,93) node [anchor=north west][inner sep=0.75pt]  [font=\large] [align=left] {$V_{i}$};
	\draw (359,89) node [anchor=north west][inner sep=0.75pt]  [font=\large] [align=left] {$V_{j}$};
	\draw (64.67,128) node [anchor=north west][inner sep=0.75pt]  [font=\normalsize] [align=left] {$\Measy$};

\end{tikzpicture}}}
	\hspace{0.4cm}
	\subcaptionbox{This figure shows two groups $V_i$ and $V_j$ with multiple edges of $\Geasy$, none of which can be weakly-represented.\label{}}%
	[.45\linewidth]{ \resizebox{180pt}{70pt}{

\tikzset{every picture/.style={line width=0.75pt}} 

\begin{tikzpicture}[x=0.75pt,y=0.75pt,yscale=-1,xscale=1]
	
	\draw  [line width=1.5]  (100,164.5) .. controls (100,138.82) and (120.82,118) .. (146.5,118) .. controls (172.18,118) and (193,138.82) .. (193,164.5) .. controls (193,190.18) and (172.18,211) .. (146.5,211) .. controls (120.82,211) and (100,190.18) .. (100,164.5) -- cycle ;
	\draw  [line width=1.5]  (325,163.5) .. controls (325,137.82) and (345.82,117) .. (371.5,117) .. controls (397.18,117) and (418,137.82) .. (418,163.5) .. controls (418,189.18) and (397.18,210) .. (371.5,210) .. controls (345.82,210) and (325,189.18) .. (325,163.5) -- cycle ;
	\draw [line width=1.5]    (167,181) -- (353,183) ;
	\draw [line width=1.5]    (171,150) -- (352,151) ;
	
	\draw (135,93) node [anchor=north west][inner sep=0.75pt]  [font=\large] [align=left] {$V_{i}$};
	\draw (359,89) node [anchor=north west][inner sep=0.75pt]  [font=\large] [align=left] {$V_{j}$};
	\draw (231.67,128) node [anchor=north west][inner sep=0.75pt]   [align=left] {$\Geasy$};

\end{tikzpicture}}}
	\caption{Illustration of edges that satisfy the first condition, but not other conditions of weakly-represented edges.}
	\label{fig2}
\end{figure}

\begin{figure}[h!]
	\centering
	{   \resizebox{180pt}{80pt}{ 

\tikzset{every picture/.style={line width=0.75pt}} 

\begin{tikzpicture}[x=0.75pt,y=0.75pt,yscale=-1,xscale=1]
	
	\draw  [line width=1.5]  (120,184.5) .. controls (120,158.82) and (140.82,138) .. (166.5,138) .. controls (192.18,138) and (213,158.82) .. (213,184.5) .. controls (213,210.18) and (192.18,231) .. (166.5,231) .. controls (140.82,231) and (120,210.18) .. (120,184.5) -- cycle ;
	\draw  [line width=1.5]  (345,183.5) .. controls (345,157.82) and (365.82,137) .. (391.5,137) .. controls (417.18,137) and (438,157.82) .. (438,183.5) .. controls (438,209.18) and (417.18,230) .. (391.5,230) .. controls (365.82,230) and (345,209.18) .. (345,183.5) -- cycle ;
	\draw [color={rgb, 255:red, 74; green, 144; blue, 226 }  ,draw opacity=1 ][line width=1.5]    (199,188) -- (360,183) ;
	\draw    (179.67,159.33) -- (257.67,91.33) ;
	\draw    (199.67,179.33) -- (277.67,111.33) ;
	\draw    (189.67,170.33) -- (267.67,102.33) ;
	\draw    (384.67,148.33) -- (338,94) ;
	\draw    (375.67,157.33) -- (329,103) ;
	
	\draw (155,113) node [anchor=north west][inner sep=0.75pt]  [font=\large] [align=left] {$\displaystyle V_{i}$};
	\draw (386,109) node [anchor=north west][inner sep=0.75pt]  [font=\large] [align=left] {$\displaystyle V_{j}$};
	\draw (200,90) node [anchor=north west][inner sep=0.75pt]  [font=\large,color={rgb, 255:red, 208; green, 2; blue, 27 }  ,opacity=1 ] [align=left] {$\displaystyle < \outside$};
	\draw (259,191) node [anchor=north west][inner sep=0.75pt]   [align=left] {$\Geasy$};
	\draw (350,88) node [anchor=north west][inner sep=0.75pt]  [font=\large,color={rgb, 255:red, 208; green, 2; blue, 27 }  ,opacity=1 ] [align=left] {$\displaystyle < \outside$};

\end{tikzpicture}}  }
	\caption{This figure shows a strongly-represented edge (in blue). There is exactly one edge $e$ of $\Geasy$ between $V_i$ and $V_j$ both of which are non-expanding and clean. Thus, $e$ is strongly-represented. \label{fig3}}
\end{figure}
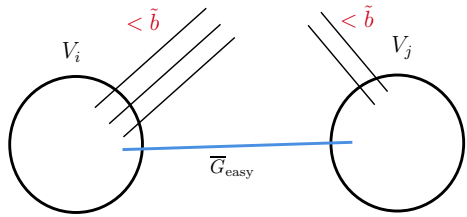

\smallskip

In this subsection, we show that $\approx \opt/\alpha$ edges of $\Mstar$ are strongly represented. Then, in the next subsection, we design our recovery algorithm in a way that can recover all strongly represented edges with high probability. Since these
edges are coming from a matching themselves, this allows us to find a large enough matching in the input graph.  We now state the main lemma for this subsection.

\begin{lemma}\label{lem:no_strong-rep}
	The number of strongly-represented edges is at least $\nicefrac{\opt}{8\alpha}$ with probability at least $1-n^{-\delta/6}$.
\end{lemma}

We start the proof with an easy claim that lower bounds the size of $\Mstar$. 

\begin{claim}\label{clm:sc-Mstar-large}
	There are at least $2\,\opt/3$ edges in $\Mstar$. 
\end{claim}
\begin{proof}
	Recall that $M$ is the $(3\,\opt/4)$-size matching of $\Geasy$ and $\Mstar$ is a subset of $M$ on vertices with degree at most $(\outside/2)$ in $\Geasy$. For any $v \in V(M)$, let $d(v)$ denote the degree of $v$ in $\Geasy$. 
	We have
	\[
		\sum_{v \in V(M)} d(v) \leq 2\cdot\card{E(\Geasy)} = 40\cdot\opt \cdot \log^4{n} \tag{by the sparsify-case  of $\Geasy$ in~\Cref{lem:ms}}. 
	\]
	Thus, the average degree of vertices in $V(M)$ is at most $(40\log^4{n}) = o(\outside)$. By Markov bound, the total number of vertices in $V(M)$ with degree more than $(\outside/2)$ is then at most $o(\opt)$.
	Removing all these vertices still leaves out $3\,\opt/4 - o(\opt) > 2\,\opt/3$ edges which all belong to $\Mstar$. \Qed{\Cref{clm:sc-Mstar-large}}
	
\end{proof}

Our goal is now to show that $\approx 1/\alpha$ fraction of edges of $\Mstar$ are strongly represented by some pairs of groups. In the following, we first bound the probability that an edge is weakly-represented and prove that the number of 
weakly-represented edges is both large enough and concentrated. We will then bound the number of these edges that will be strongly-represented also (which no longer is necessarily concentrated). 

\begin{lemma}\label{lem:sc-prob-weak}
	For any edge $e=(u,v) \in \Mstar$,
	\vspace{-0.25cm}
	\[
		\Pr\paren{\textnormal{$e$ is weakly-represented by some pairs of groups}} \geq \frac{1}{3\alpha}. 
	\] 
\end{lemma}
To prove~\Cref{lem:sc-prob-weak}, we bound the probability of each condition of being weakly-represented separately. 

The first condition is that $e$ should appear between some $(V_i,V_j)$. Given that there are $\approx k^2/\alpha$ pairs neighboring groups (by choice of $\Filter$) and $e$ can appear between each of these groups with probability $1/k^2$, 
we will get that the probability $e$ appear between a pair of groups is $\approx 1/\alpha$. Formally, 

\begin{claim}[Condition $(i)$ of weakly-represented]\label{clm:sc-weak-appear}
	\[
	\Pr\paren{\textnormal{$e$ appears between some pairs of groups}} \geq \frac{2}{5\alpha}.
	 \]
\end{claim}
\begin{proof}
	Since $\Filter$ is an adjacency matrix of a $(k/\alpha)$-regular graph on $k$ vertices, there are exactly $(k^2/2\alpha)$ pairs of neighboring groups in $\VV$ (assuming both directions of pairs are included). 
	Note that $e$ could appear between multiple groups, so we need to avoid over-counting.
	In the following, the summands are always only over neighboring pairs. By inclusion-exclusion principle, we have, 
	\begin{align*}
		\textnormal{LHS of~\Cref{clm:sc-weak-appear}}  &\geq 
		\hspace{-0.15cm}\sum_{(V_i,V_j)} \Pr\paren{e \in (V_i,V_j)} -  \hspace{-1cm}\sum_{(V_{i_1},V_{j_1})\neq (V_{i_2},V_{j_2})} \hspace{-1cm}\Pr\paren{e \in (V_{i_1},V_{j_1}) \wedge e \in (V_{i_2},V_{j_2})} \\
		&= \sum_{(V_i,V_j)}  \frac{1}{k^2} - \sum_{\substack{(V_{i_1},V_{j_1})\neq (V_{i_2},V_{j_2})\\ \card{\set{i_1,j_1} \cap \set{i_2,j_2}}=1}} \frac{1}{k^3} -  
		\sum_{\substack{(V_{i_1},V_{j_1})\neq (V_{i_2},V_{j_2})\\ \card{\set{i_1,j_1} \cap \set{i_2,j_2}}=0}} \frac{1}{k^4},
	\end{align*}
	where each term of the second inequality is because: for the first-term, the probability depends on the choice of $h_i(u)$ and $h_j(v)$ which are independent; for the second-term, the probability depends 
	on the choice of $h_{i_1}(u)$ and $h_{j_1}(v)$ and $h_{j_2}(v)$ (assuming $i_1 = i_2$, and similarly for other cases) which are independent; and for the last-term, the probability depends on 
	$h_{i_1}(u),h_{j_1}(v)$ and $h_{i_2}(u),h_{j_2}(v)$ which are all independent; as these hash functions are marginally uniform over $[k]$, we get the bound. 
	
	Moreover, there are exactly $(k^2/2\alpha)$ choices for the first summand, at most $3k^3/\alpha^2$ for the second one, and at most $k^4/\alpha^2$ for the last one. Thus, 
	\[
		\textnormal{LHS of~\Cref{clm:sc-weak-appear}}  \geq \frac{k^2}{2\alpha} \cdot \frac{1}{k^2} - \frac{3k^3}{\alpha^2} \cdot \frac{1}{k^3} - \frac{k^4}{\alpha^2}\cdot \frac{1}{k^4} > \frac{2}{5\alpha}, 
	\]
	as $\alpha > 100$ by~\Cref{assumption:range}. \Qed{\Cref{clm:sc-weak-appear}} 
	
\end{proof}

The second condition is that no other edge of $\Geasy$ should appear between $(V_i,V_j)$ (conditioned on $e$ already appearing between $V_i$ and $V_j$). Since $\Geasy$ has $\approx \opt \cdot \poly\log{(n)}$ edges only 
and each edge appears between $(V_i,V_j)$ specifically with probability only $1/k^2 \approx \alpha^2/\opt^2$, the probability that another edge appears between $V_i$ and $V_j$ is only $o(1)$. Formally, 

\begin{claim}[Condition $(ii)$ of weakly-represented]\label{clm:sc-weak-another}
	\[
	\Pr\paren{\textnormal{another edge of $\Geasy$ appear between $(V_i,V_j)$} \mid e=(u,v) \in (V_i,V_j)}= o(1).
	\]
\end{claim}
\begin{proof}
	We partition the edges of $\Geasy$ into two parts: the edges $E_1:=E_1(\Geasy,u,v)$ that are incident on either $u$ or $v$, and the remaining edges $E_2:=E_2(\Geasy,u,v)$. By union bound, 
	\begin{align*}
		\textnormal{LHS of~\Cref{clm:sc-weak-another}} &\leq \sum_{f \in E_1} \Pr\paren{f \in (V_i,V_j) \mid e \in (V_i,V_j)} + \sum_{f \in E_2} \Pr\paren{f \in (V_i,V_j) \mid e \in (V_i,V_j)} \\
		&= \card{E_1} \cdot \frac{1}{k} +  \card{E_2} \cdot \frac{1}{k^2}, 
	\end{align*}
	where each term of the second inequality is because: for the first term, assuming $f_1 = (u,w)$ for $w \neq v$ (the other case is symmetric), we need to have $h_j(w) = h_j(v)$ which  happens with probability $1/k$ as $h_j(\cdot)$ is $(>2)$-wise independent; and for the second term, assuming $f=(w,z)$ for $w\neq u$ and $z \neq v$ we need to have $h_i(w) = h_i(u)$ and $h_j(z) = h_j(v)$ which only happens with probability $1/k^2$ as both $h_i(\cdot)$ and $h_j(\cdot)$ are $(>2)$-wise independent, and also independent of each other. 
	
	Moreover, as $e \in \Mstar$, 
	we have that $\card{E_1} \leq \outside=n^{\delta/4}$ (\Cref{eq:M*}) and by sparsify-case property of~\Cref{lem:ms}, we have $\card{E_2} < 20\,\opt\cdot\log^4{n}$. Since $k= (10\opt/\alpha)$, we have, 
	\[
		\textnormal{LHS of~\Cref{clm:sc-weak-another}} \leq n^{\delta/4} \cdot \frac{\alpha}{10\,\opt} +(20\,\opt\cdot\log^4{n}) \cdot \frac{\alpha^2}{100\,\opt^2} = o(1),
	\]
	as $\opt \geq \alpha^2 \cdot n^{\delta}$ by~\Cref{assumption:guess}. \Qed{\Cref{clm:sc-weak-another}}

\end{proof}

Finally, the last condition is that both $V_i$ and $V_j$ should be clean, namely, there is no vertex of $\Measy$ inside either of them (again, conditioned on $e$ already appearing between $(V_i,V_j)$). 
There are at most $(\opt/4\alpha)$ vertices in $\Measy$ and each one appear in either group with probability $1/k \approx \alpha/\opt$, thus we can bound the probability that neither group has any vertex of $\Measy$ by some small constant. 
Formally, 

\begin{claim}[Condition $(iii)$ of weakly-represented]\label{clm:sc-weak-clean}
	\[
	\Pr\paren{\textnormal{one of $V_i$ or $V_j$ is not clean} \mid e=(u,v) \in (V_i,V_j)} < \frac{1}{20}.
	\]
\end{claim}
\begin{proof}
	Recall that $\Measy$ has less than $(\opt/4\alpha)$ vertices. By union bound, we thus have, 
	\begin{align*}
		\textnormal{LHS of~\Cref{clm:sc-weak-clean}} &\leq \sum_{w \in V(\Measy)} \Paren{\Pr\paren{h_i(w)=1 \mid e \in (V_i,V_j)} + \Pr\paren{h_j(w)=1 \mid e \in (V_i,V_j)}} \\
		&= \card{V(\Measy)} \cdot \paren{\frac{1}{k}+\frac{1}{k}} \tag{as both of $h_i(\cdot)$ and $h_j(\cdot)$ or $(>2)$-wise independent} \\
		&< \frac{\opt}{4\alpha} \cdot \frac{\alpha}{5 \cdot \opt} = \frac{1}{20},
	\end{align*}
	by the choice of $k= 10\opt/\alpha$. \Qed{\Cref{clm:sc-weak-clean}}
	
\end{proof}

We can now conclude the proof of~\Cref{lem:sc-prob-weak}.

\begin{proof}[Proof of~\Cref{lem:sc-prob-weak}]
	An edge $e \in \Mstar$ is weakly-represented iff it satisfies all the conditions $(i)$ to $(iii)$ of \Cref{def:Mstar-def}. \Cref{clm:sc-weak-appear} lower bounds the probability that $e$ satisfies condition $(i)$. 
	Conditioned on this event,~\Cref{clm:sc-weak-another,clm:sc-weak-clean} each upper bound the probability that $e$ does not satisfy conditions $(ii)$ or $(iii)$, respectively. 
	Thus, 
	\begin{align*}
		\Pr\paren{\textnormal{$e$ is weakly-represented}} \geq \frac{2}{5 \alpha} \cdot (1-o(1)- \frac{1}{20}) > \frac{1}{3\alpha}, 
	\end{align*}
	concluding the proof. \Qed{\Cref{lem:sc-prob-weak}}
	
\end{proof}

Let $\WR$ be a random variable for the number of weakly-represented edges. By~\Cref{clm:sc-Mstar-large} and~\Cref{lem:sc-prob-weak}, 
\begin{align}
	\expect{\WR} = \card{\Mstar} \cdot \Pr\paren{\textnormal{weakly-represented}} \geq \frac{2\,\opt}{3} \cdot \frac{1}{3\alpha} \geq \frac{\opt}{6\alpha}. \label{eq:WR}
\end{align}
While there is a degree of correlation between different edges of $\Mstar$ being weakly-represented, it is not too much and thus we can prove $\WR$ is also concentrated using a careful argument. 
We first need to bound the number of weakly-represented edges that can appear inside any group. 

\begin{claim}\label{clm:sc-group-weak-size}
	For any group $V_i \in \VV$, 
	\[
		\Pr\paren{\text{more than $(\log^2{n})$ weakly-represented edges appear inside $V_i$}} \ll 1/\poly{(n)}. 
	\]
\end{claim}
\begin{proof}
	As each weakly-represented edge belong to $\Mstar$, we simply upper bound the number of edges of $\Mstar$ that appear inside $V_i$. 
	For any edge $e \in \Mstar$, define an indicator random variable $X_e \in \set{0,1}$ which is $1$ iff $e \in V_i$. 
	For $X_e = 1$ to happen at least one endpoint of $e$ should be in $V_i$ and at least one endpoint of $e$ should be in one of $k/\alpha$ neighboring groups of $V_i$ (this is an ``upper bound'' because these endpoints of $e$ should be different, but 
	we are only interested in upper bounding the probability of $X_e=1$). As such, 
\begin{align*}
	\expect{X_e} &\leq \Pr\paren{\text{one end of $e$ is in $V_i$}} \cdot \Pr\paren{\text{one end of $e$ is in a neighbor of $V_i$}} \tag{by the independence of choice of vertices in different groups} \\
		&\leq \frac{2}{k} \cdot \frac{k}{\alpha} \cdot \frac{2}{k} \leq \frac{1}{\opt}. \tag{by the choice of $k=10\opt/\alpha$ and as each vertex belongs to a group with probability $1/k$}
\end{align*}
	Define $X:= \sum_{e \in \Mstar} X_e$ as the number of edges of $\Mstar$ appearing inside $V_i$. As such, $\expect{X} \leq 1$. Moreover,  $\Mstar$ is a matching (thus vertex-disjoint edges), choice of 
	vertices inside each $V_i$ is $(\log^2{n})$-wise independent, and choice of different $V_i$'s are independent. Hence, the set of variables $\set{X_e}_{e \in \Mstar}$ are $(\log^2{n})$-wise independent. 
	By concentration results for sum of $(\log^2{n})$-wise independent random variables (\Cref{prop:chernoff-limited}), 
	\[
		\Pr\paren{X > (\log^2{n})} \leq \Pr\paren{X > (\log^2{n}) \cdot \expect{X}} \leq \exp\paren{-\frac{(\log^2{n})}{6}} \ll 1/\poly{(n)}. 
	\]
	A union bound over all groups concludes the proof. \Qed{\Cref{clm:sc-group-weak-size}} 
	
\end{proof}

We now prove that the random variable $\WR$ is concentrated. 

\begin{claim}[Number of weakly-represented edges is concentrated]\label{clm:WR-conc}
	\[
	\Pr\paren{\card{\WR - \expect{\WR}} > \frac{\opt}{42 \cdot \alpha}} \ll 1/\poly{(n)}. 
	\]
\end{claim}
\begin{proof}
	$\WR$ is the sum of $\card{\Mstar}$ random variables, each determining whether a given edge in $\Mstar$ is weakly-represented or not. 
	These random variables are \emph{not} independent because an edge between $V_i$ and $V_j$ being weakly-represented restricts other edges between $V_i$ and $V_j$ from being weakly-represented.
	Thus, we cannot directly apply the Chernoff bound. Instead, we are going to use McDiarmid's inequality (\Cref{prop:mcdiarmid}). 
	
	 Define random variables $Z_1,\ldots,Z_{k}$ as the choice of vertices in each group in $\VV$, i.e., each $Z_i \subseteq V$ and is equal to the set of vertices in $V_i$. The choice of random variables $\set{Z_i}$ are independent of each other (as they are decided by different hash functions $\set{h_i(\cdot)}$). The value of random variable $\WR$  is a deterministic
	function of $Z_1,\ldots,Z_{k}$ so we can set $\WR = f(Z_1,\ldots,Z_{k})$ for some function $f$. To apply \Cref{prop:mcdiarmid}, we need to have that $f(\cdot)$ is Lipschitz (for some relatively small parameter), but this is not the case in general. We fix 
	this in the following. 
	
	Define another function $g(Z_1,\ldots,Z_k)$ as follows. Let $g$ count the number of weakly-represented edges defined by the choices of $Z_1,\ldots,Z_k$ that appear between groups with 
	at most $(\log^2{n})$ other weakly-represented edges appearing inside them. By~\Cref{clm:sc-group-weak-size}, we have that, 
	\[
		\Pr\paren{f \neq g} \ll 1/\poly{(n)} \quad \text{and thus} \quad \expect{g} \geq \expect{f} - 1.
	\]
	As such, we can prove a concentration for $g$ instead of $f$ and obtain the result for $f$ as well. We now prove that $g$ is $(2 \log^2{n})$-Lipschitz which allows us to prove its concentration. 
	Suppose we change the realization of a single variable $Z_i$ from a set of vertices $U_i$ to $U'_i$. Then, the following may happen: 
	\begin{itemize}
		\item Both $U_i$ and $U'_i$ have less than $(\log^2{n})$ weakly-represented edges appearing inside them: Thus, the change of $U_i$ to $U'_i$ can only change the value of $g$ by at most $(2\log^2{n})$. 
		\item $U_i$ has more and $U'_i$ has less than $(\log^2{n})$ weakly-represented edges appearing inside them: In $g$, none of the weakly-represented edges incident on $U_i$ were counted. In $U'_i$, at most $(\log^2{n})$ new edges
		will be counted toward $g$. This changes the value of $g$ by at most $(\log^2{n})$. 
		\item $U_i$ has less and $U'_i$ has more than $(\log^2{n})$ weakly-represented edges appearing inside them: At most all the $(\log^2{n})$ weakly-represented edges incident on $U_i$ are going to be not counted toward $g$ when
		switching to $U'_i$, thus changing the value of $g$ by at most $(\log^2{n})$. 
		\item Both $U_i$ and $U'_i$ have more than $(\log^2{n})$ weakly-represented edges appearing inside them: Neither contributed any value to $g$ so value of $g$ remains the same. 
	\end{itemize}
	Thus, $g$ is $(2\log^2{n})$-Lipschitz. Given this, we can apply McDiarmid's inequality (\Cref{prop:mcdiarmid}) and get, 
	\begin{align*}
		\Pr\paren{\card{g - \expect{g}} > \frac{\opt}{42 \cdot \alpha}} &\leq 2 \cdot \exp\paren{-\frac{2 \cdot \opt^2}{42^2 \cdot \alpha^2 \cdot k \cdot (2\log^2{n})^2}} 
		\ll 1/\poly{(n)}. \tag{as by~\Cref{assumption:range} $\opt \geq \alpha^2 \cdot n^{\delta}$} 
	\end{align*}
	This concludes the proof. \Qed{\Cref{clm:WR-conc}} 

\end{proof}

By~\Cref{clm:WR-conc} and~\Cref{eq:WR}, we have that with high probability 
\begin{align}
	\WR > \frac{\opt}{7 \cdot \alpha}. \label{eq:WR-high}
\end{align}

Let us now bound the number of strongly-represented edges, denoted by $\SR$. In the following lemma, 
we bound the number of strongly-represented edges in an indirect way by bounding how many edges among weakly-represented edges can no longer be strongly-represented. 

\begin{lemma}\label{lem:sc-strong}
	With probability $1-n^{-\delta/6}$, we have 
	$
		\SR > \WR - \dfrac{\opt}{56 \cdot \alpha}. 
	$
\end{lemma}

We need the following claim that bounds the probability that a group is expanding. We will then simply subtract the edges of all expanding groups from the weakly-represented edges to obtain the number of strongly-represented one. 
This probability should be small because $\Geasy$ only has $\approx \opt \cdot \poly\log{(n)}$ edges and each edge appear inside a group with probability roughly $\approx 1/\opt$. Formally, 

\begin{claim}\label{clm:sc-expanding}
	For any group $V_i$, 
	$
	\Pr\paren{\textnormal{$V_i$ is expanding}} <  n^{-\delta/5}.
	$
\end{claim}
\begin{proof}
	For $V_i$ to be expanding, at least $\outside$ edges of $\Geasy$ should appear inside $V_i$. Let $E_1 := E(\Geasy)$ denote the edges in $\Geasy$.  
	We have that size of $E_1$ is at most $(20\opt\cdot\log^4{n})$ by sparsify-case of~\Cref{lem:ms}.  
	
	For any $e \in E_1$, let $X_e \in \set{0,1}$ be an indicator random variable which is $1$ iff $e$ appears inside $V_i$. As we proved in~\Cref{clm:sc-group-weak-size}, 
	\[
		\expect{X_e} < \frac{1}{\opt}. 
	\]
	Let $X := \sum_{e \in E_1} X_e$ denote the number of edges of $E_1$ that appear inside $V_i$. Thus, by Markov bound (and since $\outside=n^{\delta/4}$),  
	\begin{align*}
		\Pr\paren{X > \outside} \leq \frac{\expect{X}}{\outside} \leq \frac{1}{\outside} \cdot \card{E_1} \cdot \frac{1}{\opt}  = n^{-\delta/4} \cdot (20\opt\cdot\log^4{n}) \cdot \frac{1}{\opt} \ll n^{-\delta/5}, 
	\end{align*}
	concluding the proof. \Qed{\Cref{clm:sc-expanding}}
	
\end{proof}

\begin{proof}[Proof of~\Cref{lem:sc-strong}]
	Let $X$ denote the number of expanding groups. By~\Cref{clm:sc-expanding}, combined with a Markov bound (and linearity of expectation), 
	\[
		\Pr\paren{X > \frac{\opt}{\alpha \cdot 56 \cdot (\log^2{n})}} \leq (56\,\log^2{n}) \cdot \frac{\alpha}{\opt} \cdot k \cdot n^{-\delta/5}  \ll n^{-\delta/6},
	\]
	as $k= (10 \cdot \opt/\alpha)$. Additionally, by~\Cref{clm:sc-group-weak-size}, at most $(\log^2{n})$ weakly-represented edges appear in each group with high probability. 
	Thus, with probability at least $1-n^{-\delta/6}$, the total 
	number of weakly-represented edges incident on expanding groups is at most 
	\[
		\frac{\opt}{56 \cdot \alpha \cdot \log^2{n}} \cdot \log^2{n} = \frac{\opt}{56 \cdot \alpha}. 
	\]
	All remaining weakly-represented edges will also be strongly-represented, proving the lemma.
	\Qed{\Cref{lem:sc-strong}}
	
\end{proof}

We can now bound the number of strongly represented edges.
\Cref{lem:no_strong-rep} follows by combining~\Cref{eq:WR-high} with~\Cref{lem:sc-strong}. We can conclude that, 
\begin{align}
	\Pr\paren{\textnormal{number of strongly-represented edges} \geq \frac{\opt}{8\cdot\alpha}} \geq 1-n^{-\delta/6}. \label{eq:SR-high} \qedhere
\end{align}

\subsection{The Recovery Algorithm}

We now show how to recover a large matching from the sketch computed by~\Cref{alg:sc}. The idea is to find all strongly-represented edges (or rather a superset of them). 
This is enough because the number of strongly-represented edges is at least ${\opt}/{8\alpha}$ with high enough probability (\Cref{lem:no_strong-rep}).
To do so, we first remove all groups that have a vertex from $\Measy$ inside them. Next, we run a ``weak tester'' using sketches for $\NET$ to essentially remove all expanding sets (this step is done slightly differently in the algorithm). 
Finally, we use $\SNR$ to recover the neighborhood of each group inside $\Geasy$ by setting the $T$-set of the sketches in the recovery as vertices matched by $\Measy$. Then, whenever between two groups we only recovered a single pair 
of vertices, we consider this pair as an edge and store them\footnote{$\SNR$ can find the neighbor vertex $v$ of a group $V_i$ inside another group $V_j$; however, it cannot recover an edge because it does not specify the endpoint of
neighbor $v$ inside $V_i$. This is fixed by the process mentioned above by also finding a unique neighbor inside $V_i$ from the $\SNR$ run on $V_j$ instead}. At the end, we compute a maximum matching among the stored edges. 

\begin{Algorithm}\label{alg:sc-rec} 
	The recovery algorithm  of~\Cref{lem:sc2}.
	
	\smallskip
	
	\textbf{Input:} Groups $\VV$ and sketches computed by~\Cref{alg:sc} and a matching $\Measy$ (of~\Cref{lem:ms}). 
	\smallskip
	
	\textbf{Output:} A matching $\Mhard$ in $G$. 
	
	\begin{enumerate}[label=$(\roman*)$]
		\item\label{line:sc-rec1} For every group $V_i \in \VV$, define $T_i$ as the vertices in the graph $G(V_i)$ (defined in~\Cref{alg:sc}) that also appear in $\Measy$. Run the following tests: 
		\begin{itemize}[leftmargin=15pt]
			\item \textbf{\Measy-test}: if $V_i$ has any vertex of $\Measy$ inside it, \textbf{remove} $V_i$;
			\item \textbf{Expanding-test}:  Run the recovery algorithm of $\NET(G(V_i),V_i)$ with the set $T=T_i$ to test if $V_i$ has at most $\outside$ neighbors or at least $2 \outside$ neighbors
			 out of $T_i$ in $G(V_i)$: if $2 \outside$, \textbf{remove} $V_i$. 
		\end{itemize}
		\item For any remaining group $V_i$,  run the recovery algorithm $\SNR(G(V_i),V_i)$ with set $T=T_i$ to recover $\NR(V_i)$ in the graph $G(V_i)$. 
		\item Define the following \textbf{recovered graph} $\HR$ on vertices $V$. For any two remaining groups $V_i,V_j$, if sizes of both $\NR(V_i) \cap V_j$ and $\NR(V_j) \cap V_i$ is $1$, add the edge $(u,v)$ to $H$ where $u$ and $v$ are the unique 
		vertices in the aforementioned sets. 
		 Return $\Mhard$ as a maximum matching of $H$. 
	\end{enumerate}
\end{Algorithm}

\begin{figure}[h!]
	\centering
	\subcaptionbox{This figure shows the group $V_i$ and $\NR(V_i)$ has only one element $u$ in $V_j$.   \label{fig:Spar_a}}%
	[.45\linewidth]{ 	\resizebox{160pt}{100pt}{\tikzset{every picture/.style={line width=0.75pt}} 

\begin{tikzpicture}[x=0.75pt,y=0.75pt,yscale=-1,xscale=1]

\draw   (121,179) .. controls (121,149.18) and (145.18,125) .. (175,125) .. controls (204.82,125) and (229,149.18) .. (229,179) .. controls (229,208.82) and (204.82,233) .. (175,233) .. controls (145.18,233) and (121,208.82) .. (121,179) -- cycle ;
\draw   (286,108) .. controls (286,78.18) and (310.18,54) .. (340,54) .. controls (369.82,54) and (394,78.18) .. (394,108) .. controls (394,137.82) and (369.82,162) .. (340,162) .. controls (310.18,162) and (286,137.82) .. (286,108) -- cycle ;
\draw   (319,104.5) .. controls (319,100.63) and (322.13,97.5) .. (326,97.5) .. controls (329.87,97.5) and (333,100.63) .. (333,104.5) .. controls (333,108.37) and (329.87,111.5) .. (326,111.5) .. controls (322.13,111.5) and (319,108.37) .. (319,104.5) -- cycle ;
\draw    (319,104.5) -- (258.2,131.72) -- (185,164.5) ;

\draw (145,100) node [anchor=north west][inner sep=0.75pt]   [align=left] {$V_i$};
\draw (278,55) node [anchor=north west][inner sep=0.75pt]   [align=left] {$V_j$};
\draw (336,94) node [anchor=north west][inner sep=0.75pt]   [align=left] {$u$};

\end{tikzpicture}}}
	\hspace{0.4cm}
	\subcaptionbox{This figure shows the group $V_j$ and $\NR(V_j)$ has only one element $v$ in $V_i$.  \label{fig:Spar_b}}%
	[.45\linewidth]{ \resizebox{160pt}{100pt}{\tikzset{every picture/.style={line width=0.75pt}} 

\begin{tikzpicture}[x=0.75pt,y=0.75pt,yscale=-1,xscale=1]

\draw   (121,179) .. controls (121,149.18) and (145.18,125) .. (175,125) .. controls (204.82,125) and (229,149.18) .. (229,179) .. controls (229,208.82) and (204.82,233) .. (175,233) .. controls (145.18,233) and (121,208.82) .. (121,179) -- cycle ;
\draw   (286,108) .. controls (286,78.18) and (310.18,54) .. (340,54) .. controls (369.82,54) and (394,78.18) .. (394,108) .. controls (394,137.82) and (369.82,162) .. (340,162) .. controls (310.18,162) and (286,137.82) .. (286,108) -- cycle ;
\draw   (171,164.5) .. controls (171,160.63) and (174.13,157.5) .. (178,157.5) .. controls (181.87,157.5) and (185,160.63) .. (185,164.5) .. controls (185,168.37) and (181.87,171.5) .. (178,171.5) .. controls (174.13,171.5) and (171,168.37) .. (171,164.5) -- cycle ;
\draw    (319,104.5) -- (258.2,131.72) -- (185,164.5) ;

\draw (145,100) node [anchor=north west][inner sep=0.75pt]   [align=left] {$V_i$};
\draw (278,55) node [anchor=north west][inner sep=0.75pt]   [align=left] {$V_j$};
\draw (156,168) node [anchor=north west][inner sep=0.75pt]   [align=left] {$v$};

\end{tikzpicture}}}
	\caption{Illustration of two groups $V_i$ and $V_j$ that pass the $\Measy$-test and Expanding-test. We have that $\NR(V_i)$ has only one element $u$ in $V_j$ and $\NR(V_j)$ has only one element $v$ in $V_i$. This means that $(u,v)$ must be an edge in $G$. }
	\label{fig:Spar}
\end{figure}
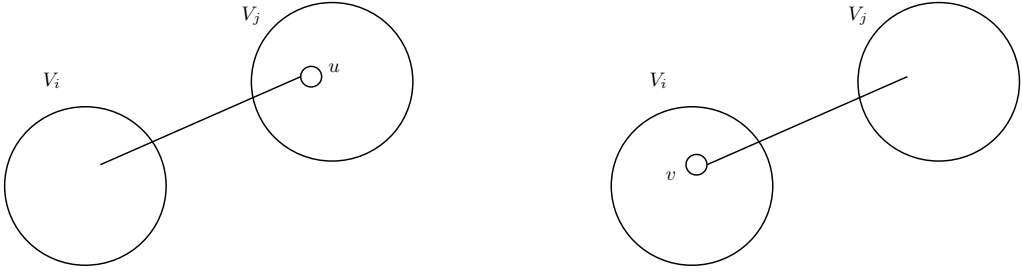

\Cref{fig:Spar} shows which edges are added to $H$.
Our goal now is to show that with high probability $\HR$ contains all strongly-represented edges and moreover it does not contain any edge that is not part of $G$. Putting these two together with~\Cref{lem:no_strong-rep}
then finalizes the proof. 

The first  step is to show that for both $\NET$ and $\SNR$ sketches run by the algorithm, the promise on the input is satisfied. We first prove that the set $T_i$
satisfies $\card{T_i} \leq a$ for all $i \in [k]$ (as required by both algorithms with given parameter $a$). 

\begin{claim}\label{clm:sc-Ti-small}
	With high probability, for every $i \in [k]$, size of $T_i$ is at most $a$. 
\end{claim}
\begin{proof}
	For any $v \in V(\Measy)$, define an indicator random variable $X_v \in \set{0,1}$ which is $1$ iff $v$ is in $T_i$, i.e., there exists an index $j$ such that $h_j(v) = 1$ and $\Filter(i,j) = 1$. 
	Given there are exactly $(k/\alpha)$ choices for $j \in [k]$ with $\Filter(i,j)=1$, we have that by union bound, 
	\[
		\expect{X_v} \leq \frac{k}{\alpha} \cdot \frac{1}{k} = \frac{1}{\alpha},
	\]
	as each $h_j(\cdot)$ is uniform over $[k]$. Let $X := \sum_{v \in V(\Measy)} X_v$ denote the size of $T_i$. As there are at most $(\opt/4\alpha)$ vertices in $\Measy$, we get
	\[
		\expect{X} \leq \card{V(\Measy)} \cdot \frac{1}{\alpha} \leq \frac{\opt}{4\alpha^2}. 
	\]
	Finally, note that random variables $\set{X_v}$ are only correlated through the choice of $(\log^2{n})$-wise independent hash functions $\set{h_j(\cdot)}$ (which are themselves independent for different $h_j$'s).
	Thus, by concentration results for sum of $(\log^2{n})$-wise independent random variables (\Cref{prop:chernoff-limited}), 
	\[
		\Pr\paren{X > a} \leq \Pr\paren{X > 8 \cdot \expect{X}} \leq \exp\paren{-(\log^2{n})/2}\ll 1/\poly{(n)}. 
	\]
	A union bound over all groups concludes the proof. \Qed{\Cref{clm:sc-Ti-small}}
	 
\end{proof}

The above claim along with $a\geq 16 \outside$ (by \Cref{assumption:range}) is enough for running $\NET$. We now show that the guarantees for $\SNR$ are also satisfied. This first requires proving that $N(V_i) - T_i$, for each $V_i$ that is not removed by the algorithm, has size at most $2 \outside$ (here, and throughout 
the rest of the analysis $N(V_i)$ is in the 
graph $G(V_i)$). 

\begin{claim}\label{clm:sc-SNR-b}
	With high probability, for every remaining group $V_i$, size of $N(V_i) - T_i$ is at most $2 \outside$. 
\end{claim}
\begin{proof}
	Given that the promise to $\NET(G(V_i),V_i)$ is satisfied by~\Cref{clm:sc-Ti-small}, with high probability, $\NET$ is going to only output $\outside$-case for a set $V_i$ when $N(V_i) - T_i$ has size less than $2 \outside$ by~\Cref{lem:nbhd-size-test}. 
	Thus, any group not removed satisfies $\card{N(V_i) - T_i} \leq 2 \outside$ as desired.  \Qed{\Cref{clm:sc-SNR-b}}
	
\end{proof}

Finally, we also need to prove that for any vertex $v \in N(V_i) - T_i$, size of $V_i \cap N(v)$ is at most $c$ (again $N(V_i)$ in the graph $G(V_i)$). This is done in a rather indirect way in the following claim. 

\begin{claim}\label{clm:sc-SNR-c}
	With high probability, for every remaining group $V_i$ and any of its neighbor group $V_j$, size of $N(V_i) \cap V_j$ is at most $c$.  
\end{claim}
\begin{proof}
	We condition on the group $V_i$ remaining after the tests. By~\Cref{clm:sc-SNR-b}, we have that there are at most $2 \outside$ vertices in $N(V_i) - T_i$. Let $O_i$ denote the set of these vertices. 
	
	Note that the randomness of this conditioning is only a function of the graph $G(V_i)$ and vertices $V_i$. Recall that the definition of graph $G(V_i)$ was the following 
	set of edges $\set{(u,v) \in E \mid u \in V_i,\, v \in V_j,\, \Filter(i,j)=1}$. This definition depends on the \emph{union} of the sets $\set{V_j}_{\Filter(i,j)=1}$ but \emph{not the partitioning} of vertices into different groups. 
	This is not true if multiple copies of a vertex exist in $O_i$ because if one copy belongs to a group then the others cannot belong to the same group. Assuming that multiple copies of a vertex could exist in a group only increases the size of $O_i \cap V_j$ which is okay for us since we only need an upper bound.
	Therefore, we can assume that each vertex of $O_i$ has an \emph{equal} probability of belonging to each of these $(k/\alpha)$ groups.  
	\begin{align*}
		\Pr\paren{\card{O_i \cap V_j} > c} &\leq {{\card{O_i}}\choose{c}} \cdot \paren{\frac{\alpha}{k}}^{c} \leq \paren{\frac{2 \outside \cdot \alpha}{k}}^c \tag{as size $O_i$ is at most $2 \outside$} \\
		&\leq  \paren{\frac{2 \cdot n^{\delta/4} \cdot \alpha^2}{10 \cdot \opt}}^{30/\delta} \tag{by the choice of $\outside=n^{\delta/4}$, $c=(30/\delta)$, and $k=(10\opt/\alpha)$} \\
		&< \paren{n^{-\delta/2}}^{30/\delta} = n^{-15}. \tag{by~\Cref{assumption:range} $\opt \geq \alpha^2 \cdot n^{\delta}$} \\
	\end{align*}
	A union bound over all choices for $V_i$ concludes the proof. 
	\Qed{\Cref{clm:sc-SNR-c}}
	
\end{proof}

We now show that~\Cref{clm:sc-SNR-c} implies that for every  $v \in N(V_i) - T_i$, size of $V_i \cap N(v)$ is at most $c$, thus satisfying the guarantee of $\SNR$. Suppose the event of \Cref{clm:sc-SNR-c} happens 
for all remaining groups. This means that among the remaining groups, between every pair of neighbor groups $V_i$ and $V_j$, there can be at most $c$ vertices in $V_i$ that have an edge to $V_j$. 
Naturally, any vertex in union of all $V_j$'s can then also only have $c$ neighbors in $V_i$. Since all of $N(V_i)-T_i$ is now a subset of these remaining $V_j$'s, we get the desired guarantee. 
To conclude, by~\Cref{clm:sc-Ti-small,clm:sc-SNR-b,clm:sc-SNR-c} we established that the guarantees required by $\NET$ and $\SNR$ are all satisfied. 

We now show how the algorithm can recovers strongly-represented edges. The first step is to show that the endpoint-groups of strongly-represented edges will not be removed by~\Cref{alg:sc-rec} with high probability. 

\begin{claim}\label{clm:sc-sr-remains}
	Suppose $e$ is a strongly-represented edge by groups $(V_i,V_j)$. Then, with high probability, neither of $V_i$ nor $V_j$ will be removed  by~\Cref{alg:sc-rec}. 
\end{claim}
\begin{proof}
	Since $e$ is strongly-represented, we know both $V_i$ and $V_j$ are clean (thus will not be removed by $\Measy$-test) and are non-expanding, thus each have at most $\outside$ other edges of $\Geasy$ appearing inside them. 
	Thus, these groups will also pass the expanding-test, and so they will not be removed as long as $\NET$ is working correctly which happens with high probability. \Qed{\Cref{clm:sc-sr-remains}}
	
\end{proof}

\begin{lemma}\label{lem:sc-sr-pick}
	With high probability, $\HR$ is a subgraph of $G$ and contains all strongly-represented edges. 
\end{lemma}
\begin{proof}
	We condition on the high probability events of~\Cref{clm:sc-Ti-small,clm:sc-SNR-b,clm:sc-SNR-c,clm:sc-sr-remains} and correctness of $\NET$ and $\SNR$. 
	
	Firstly, consider an edge $(u,v)$ added to $\HR$ and assume
	$u \in V_i$ and $v \in V_j$. For $(u,v)$ to be included in $\HR$, we need to have $N(V_i) \cap V_j = \set{v}$ and $N(V_j) \cap V_i = \set{u}$. This means there is an edge $(u,v)$ in $G$ as well, proving the second part of the lemma. 
	
	For the first part, consider any strongly-represented edge $(u,v)$ and again assume $u \in V_i$ and $v \in V_j$ (condition $(i)$ of weakly-represented edges).  These groups will not be removed
	as we conditioned on the event of~\Cref{clm:sc-sr-remains}. Moreover, there is no other edge of $\Geasy$ between $V_i$ and $V_j$ (condition $(ii)$ of weakly-represented edges) and since $V_i$ and $V_j$ do not have 
	any vertex of $\Measy$ (condition $(iii)$ of weakly-represented edges), we have that $(u,v)$ is the unique edge between $V_i$ and $V_j$. By the argument above, this means $(u,v)$ will be added to $\HR$. \Qed{\Cref{lem:sc-sr-pick}}
	 
\end{proof}

We can now conclude the proof of~\Cref{lem:sc2}. Firstly, by~\Cref{lem:sc-sr-pick}, with high probability the algorithm does not make an error. Moreover, by~\Cref{lem:no_strong-rep},  with probability at least $1-n^{-\delta/6}$, there are at least $(\opt/8\alpha)$ strongly-represented edges. Since these edges are coming from a matching $\Mstar$ and by~\Cref{lem:sc-sr-pick} we get all of those in $\HR$, 
the output matching $\Mhard$ is going to have size at least $(\opt/8\alpha)$. As we already established the bound on the space in~\Cref{lem:sc-space}, this concludes the proof.


\section{Match-Or-Sparsify Lemma}\label{sec:ms}

We prove~\Cref{lem:ms}, restated below, in this section. Informally speaking, given a graph $G$, this lemma gives an algorithm that either finds a large matching in $G$ or identifies a sparse induced subgraph of $G$ that contains a large matching.

\begin{lemma*}[Re-statement of~\Cref{lem:ms}]
	There is a linear sketch that given any graph $G=(V,E)$ specified via $\vect(E)$, uses $O(\guess^2/\alpha^3)$ bits of space and with high probability outputs a matching $\Measy$ 
	that satisfies \underline{at least one} of the following conditions:  
	\begin{itemize}
		\item \textbf{Match-case:} The matching $\Measy$ has at least $(\opt/8\alpha)$ edges; 
		\item \textbf{Sparsify-case:} The induced subgraph of $G$ on vertices not matched by $\Measy$, denoted by $\Geasy$,  has at most $(20\,\opt \cdot \log^4\!{n})$ edges and a matching of size at least $3\,\opt/4$.  
	\end{itemize} 
\end{lemma*}

The algorithm in~\Cref{lem:ms} samples $\approx \opt^2/(\alpha \cdot \log{n})^3$ edges from the graph using a non-uniform distribution as follows: 
for each sample, we first pick $\approx \opt/\alpha$ vertices $S$ uniformly at random and then use $\NES$ to sample an edge from $S$ to a vertex of $N(S)$ chosen uniformly at random. 
Given the bound of $O(\log^3{n})$ bits on the size of sketches for $\NES$, the total space of the algorithm can be bounded by $O(\opt^2/\alpha^3)$ bits.
In the recovery phase then, we compute a greedy matching over these sampled edges and return it as $\Measy$. Formally, the algorithm is as follows. 

\begin{Algorithm}\label{alg:ms} 
	The algorithm of Match-Or-Sparsify Lemma~(\Cref{lem:ms}). 
	
	\medskip
	
	\textbf{Input:} A graph $G=(V,E)$ specified via $\vect(E)$; $\qquad$ \textbf{Output:} A matching $\Measy$ in $G$. 
	
	\medskip
	
	\textbf{Parameters:} Let $k := {\opt}/{\alpha}$ and $s := \opt^2/(\alpha \cdot \log{n})^3$.  

\medskip

	 \textbf{Sketching matrix:} 
	
	\begin{enumerate}[label=$\arabic*.$, leftmargin=15pt]
	
	\item For $i=1$ to $(2s)$ steps\footnote{We  partition the steps into two \textbf{batches} of $s$ steps each in the analysis, hence the use of $2s$ for the number of steps.}: 
	\begin{enumerate}[leftmargin=15pt]
		\item \label{ms:a} Sample a pair-wise independent hash function $h_i : V \rightarrow [k]$ and set $V_i := \set{v \in V \mid h_i(v) = 1}$.
		\item \label{ms:b} Let $\Phi(V_i)$ be the sketching matrix of $\NES(G,V_i)$. 
	\end{enumerate}
	\item Return $\Phi := [\Phi(V_1); \, \cdots \, ; \Phi(V_{2s})]$ as the sketching matrix. 
	\end{enumerate}
	
\textbf{Recovery algorithm:}
	
		\begin{enumerate}[label=$\arabic*.$, leftmargin=15pt]
	
	\item For all $i \in [2s]$, run the recovery algorithm of $\NES(G,V_i)$ using $\Phi(V_i)$ and $\Phi(V_i) \cdot \vect(E)$  to get an output edge $e_i$ (we write $e_i = \perp$ if the sampler outputs FAIL).

	\item Let $\Measy \leftarrow \emptyset$ initially. For $i=1$ to $2s$ steps: greedily include  $e_i$  in $\Measy$ whenever $e_i \neq \perp$ and both its endpoints are unmatched by $\Measy$. 
	\end{enumerate}
	
\end{Algorithm}

Note that it is equivalent to think of the edges being recovered one by one and fed to the greedy matching algorithm. We will use this in our analysis.
We first bound the space complexity of this algorithm. 

\begin{lemma}\label{lem:ms-space}
	\Cref{alg:ms} uses $O(\guess^2/\alpha^3)$ bits of space with high probability. 
\end{lemma}
\begin{proof}
 
In each step, Line~\ref{ms:a} requires storing a pair-wise independent hash function which needs $O(\log n)$ bits of space by~\Cref{prop:k-wise}. Line~\ref{ms:b} requires storing an $\NES$ which needs $O(\log^3 n)$ bits by \Cref{lem:nei-edge-sampler}.
There are $2s = O(\opt^2/(\alpha \cdot \log{n})^3)$ steps, so the total space  $O(\opt^2/\alpha^3)$ bits.
\end{proof}

We now prove that the matching $\Measy$ output by~\Cref{alg:ms} satisfies the guarantees of~\Cref{lem:ms}. 
To continue, we need some notation. 

\paragraph{Notation.} For any $i \in [2s]$, let $\ME{i}$  be the set of edges included in $\Measy$ in the first $i-1$ steps of the recovery, i.e., from $\set{e_j}_{j=1}^{i-1}$, 
and $\GE{i}$ to be the subgraph of $G$ induced on \emph{unmatched} vertices of $\ME{i}$.
We use $\degG{v}{i}$ to denote the degree of each vertex in $\GE{i}$ to other vertices in $\GE{i}$. 
We partition vertices of $\GE{i}$ based on their degrees in $\GE{i}$ into \textbf{low-}, \textbf{medium-}, and \textbf{high-degree} as follows: 
\begin{align*}
& \Low{i} := \set{v: \degG{v}{i} \!<\! (\alpha \, \log^3\!{n})},  \Med{i} := \set{v : (\alpha \, \log^3\!{n}) \!\leq\! \degG{v}{i} \!<\! (\frac{\opt}{8\alpha})}, 
  \High{i} := \set{v :  \degG{v}{i} \!\geq\! \frac{\opt}{8\alpha}}.
\end{align*}
We define the following two events: 
\begin{itemize}
	\item $\eventmatch(i)$: the matching $\ME{i}$  has less than $(\opt/8\alpha)$ edges (i.e., matching-case not happened); 
	\item $\eventsparsify(i)$: the subgraph $\GE{i}$ has more than $(20\,\opt \cdot \log^4\!{n})$ edges (i.e., sparsify-case not happened). 
\end{itemize}
Finally, we say that a choice of $V_i$ in step $i \in [2s]$ is \textbf{clean} if $V_i$ does not contain any matched vertices of $\ME{i}$. 

\smallskip

We start by proving that if for some $i \in [2s]$ at least one of these events do not happen, then~\Cref{alg:ms} succeeds in outputting the desired matching of~\Cref{lem:ms}. 
The proof is straightforward. 
\begin{claim}\label{clm:event-ms}
	Suppose for some $i \in [2s]$, either of $\eventmatch(i)$ or $\eventsparsify(i)$ does not happen; then,  $\Measy$ of~\Cref{alg:ms} satisfies the guarantees of~\Cref{lem:ms}. 
\end{claim}
\begin{proof}
	Suppose first that $\eventmatch(i)$ does not happen. This means $\ME{i}$ has size at least $(\opt/8\alpha)$ and by the greedy choice of $\Measy$, we have $\card{\Measy} \geq \card{\ME{i}} \geq (\opt/8\alpha)$,  satisfying the match-case condition. 
		
	Now suppose that $\eventsparsify(i)$ does not happen. Since the number of edges of $\Geasy$ can only be smaller than that of $\GE{i}$, we have that $\Geasy$ also only has $(20\,\opt\cdot\log^4\!{n})$ edges. 
	We can also assume that size of $\Measy$ is at most $(\opt/8\alpha)$ as otherwise we will be done by the matching-case. This means that at most $(\opt/4\alpha)$ vertices of any maximum matching of $G$ are 
	incident on $\Measy$, thus $\Geasy$ still contains a matching of size at least $\opt - (\opt/4\alpha) > 3\,\opt/4$, satisfying the sparsify-case condition. 
\end{proof}

The goal at this point is to show that with high probability, for some $i \in [2s]$, one of the events $\eventmatch(i)$ or $\eventsparsify(i)$ is \emph{not} going to happen. In order to do so, we partition the steps of the algorithm into two \textbf{batches} of size $s$ 
each and analyze each one separately as follows:
\begin{itemize}[leftmargin=10pt]
	\item \textbf{First batch:} We first show that as long as $\eventmatch(i)$ and $\eventsparsify(i)$ happen for all $i \in [s]$, with high probability, the set $\High{s+1}$ (and thus $\High{j}$ for all $j \in (s,2s]$) will be empty for the second batch (a technical condition needed for our variance reduction
ideas in the next part). Formally, 

\begin{lemma}\label{lem:ms-1st-batch}
	With high probability, either at least one of $\eventmatch(i)$ and $\eventsparsify(i)$ does not happen for some step $i \in [s]$ or $\High{s+1}$ will be empty.
\end{lemma}

\item \textbf{Second batch:} We then show that whenever both $\eventmatch(i)$ and $\eventsparsify(i)$ happen in a step $i \in (s , 2s]$, there will be a probability of $\approx k/s$ in increasing 
the size of $\ME{i}$ by one in this step (this is the main part of the argument). Given that we repeat this process for $s$ steps also, this allows us to argue $\Measy$ will eventually become of size $\approx k = \opt/\alpha$, thus satisfying the matching-case 
condition (or one of the events happen along the way, and we can use~\Cref{clm:event-ms} instead). Formally, 

\begin{lemma}\label{lem:ms-2nd-batch}
	Assuming $\High{s+1}$ is empty, with high probability, at least one of the events $\eventmatch(i)$ or $\eventsparsify(i)$ does not happen for some $i \in (s:2s]$.
\end{lemma}

\end{itemize}
\noindent
\Cref{lem:ms} then follows immediately from these two lemmas combined with~\Cref{lem:ms-space,clm:event-ms}. 

Before we get to the proofs of these lemmas, we make the following important remark. 

\begin{remark}
{The actions of~\Cref{alg:ms} are clearly \emph{not} independent across different steps (in the recovery phase). However, in our upcoming probability analysis in each step $i \in [2s]$ we fix {the randomness of all prior steps} conditioned on that events $\eventmatch(i)$ and $\eventsparsify(i)$, and use \emph{only} the randomness of the choice of $(V_i,e_i)$ in this step. This randomness is independent of prior steps. 
As such, in the following, \textbf{\underline{all} our probability calculations in a step $i$ are conditioned on randomness of prior steps and events $\eventmatch(i)$ and $\eventsparsify(i)$}, without writing it explicitly each time. These 
probability calculations may not necessarily remain correct when either of these events do not happen, but we will be done by~\Cref{clm:event-ms} in those cases anyway.} 
\end{remark}

We use the following simple helper claim in the subsequent proofs (this claim would have been trivial had $h_i$ was a truly independent hash function instead of a pairwise-independent one). 

\begin{claim}\label{clm:ms-helper-1}
	Consider any step $i \in [s]$ and let $v$ be any arbitrary vertex in $\GE{i}$. Then, 
	\[
		\Pr_{V_i}\Paren{v \in V_i ~\textnormal{and $V_i$ is clean}} \geq \frac{3}{4k}. 
	\]
\end{claim}
\begin{proof}
	Recall that there are at most $\opt/4\alpha = k/4$ vertices matched by $\ME{i}$. 	We have, 
	\begin{align*}
		\Pr_{V_i}\Paren{v \in V_i ~\textnormal{and $V_i$ is clean}} &= \Pr\paren{h_i(v) = 1} \cdot \Pr\paren{\text{$V_i$ is clean} \mid h_i(v)=1} \tag{$v \in V_i$ iff $h_i(v)=1$} \\
		&=  \frac1k \cdot \Paren{1-\Pr\paren{\text{$V_i$ is not clean}  \mid h_i(v)=1}} \tag{as $h_i(v) = 1$ w.p. $1/k$} \\
		&\geq \frac1k \cdot \paren{1-\sum_{u \in V(\ME{i})} \Pr\paren{h_i(u)=1 \mid h_i(v)=1}} \tag{by union bound and since $V_i$ is not clean iff $h_i(u)=1$ for some $u \in V(\ME{i})$} \\
		&= \frac1k \cdot \paren{1-\sum_{u \in V(\ME{i})} \frac1k} \tag{$h_i(\cdot)$ is a pairwise-independent hash function} \\
		&= \frac1k \cdot \paren{1-\frac{k}{4} \cdot \frac{1}{k}} \tag{as $h_i(u)=1$ w.p. $1/k$ and there are at most $k/4$ choices for matched vertices}, 
	\end{align*}
	which is at least $3/4k$ as desired. 
\end{proof}

\subsection{First Batch: Proof of~\Cref{lem:ms-1st-batch}}

	Let $v$ be any vertex in $V$ and consider any step $i \in [s]$. If $\degG{v}{i} < (\opt/8\alpha)$, then $v$ cannot be part of $\High{i}$ and subsequently $\High{s+1}$ since $\GE{s+1}$ is a subgraph of $\GE{i}$. 
	In the following, we consider the case where $\degG{v}{i} \geq (\opt/8\alpha)$ and prove that there is a non-trivial chance of ``progress'' (to be defined later) in each step. 
	We first bound the probability of the following useful event for our analysis. 
	
	\begin{claim}\label{clm:ms-1st-v}
		In step $i$, if $\degG{v}{i} \geq \opt/8 \alpha$, we have $\Pr_{V_i}\paren{v \in N(V_i) ~\textnormal{and $V_i$ is clean} } \geq \dfrac{1}{16}$.
	\end{claim}
	\begin{proof}
		Let $d(v) := (\opt/8\alpha)$ and $D(v)$ be a set of $d$ arbitrary neighbors of $v$ in $\GE{i}$. We know that $v$ will be included in $N(V_i)$ if any of vertices in $D(v)$ is sampled in $V_i$. We have, 
		\begin{align*}
			\Pr_{V_i}\paren{v \in N(V_i) ~\textnormal{and $V_i$ is clean} } &\geq \Pr\paren{D(v) \cap V_i \neq \emptyset ~\text{and $V_i$ is clean}} \tag{$D(v) \subseteq N(v)$} \\
			&\geq \sum_{u \in D(v)} \Pr\paren{u \in V_i ~\text{and $V_i$ is clean}}  - \sum_{u\neq w \in D(v)} \Pr\paren{u,w \in V_i} \tag{by inclusion-exclusion principle and bounding 
			$\Pr\paren{u,w \in V_i} \geq \Pr\paren{u,w \in V_i~\text{and $V_i$ is clean}}$} \\
			&> \frac{3d(v)}{4k}  - \frac{d(v)^2}{k^2} \tag{by~\Cref{clm:ms-helper-1} and as $h_i(\cdot)$ is a pair-wise independent hash function with range $[k]$} \\
			&= \frac{(\opt/8\alpha)}{(\opt/\alpha)} \cdot \paren{\frac{3}{4} - \frac{(\opt/8\alpha)}{(\opt/\alpha)}}, \tag{as $d(v)=\opt/8\alpha$ and $k=\opt/\alpha$} 
		\end{align*}
		which is at least $1/16$ as desired. \Qed{\Cref{clm:ms-1st-v}}
		
	\end{proof}
	Let us now condition on the choice of $V_i$ and assume the event of~\Cref{clm:ms-1st-v} has happened. We say that this step $i$ is a \textbf{matching-step} if $N(V_i) > (\opt/2\alpha)$; otherwise, 
	we call this step a \textbf{vertex-step}. We argue that in a matching-step we have a constant probability of increasing the size of $\ME{i}$ by one and in a vertex-step we have a probability $\approx \alpha/\opt$ of matching the vertex $v$ and thus no longer  
	including it in $\GE{i+1}$ and $\High{i+1}$. We formalize this in the following. 
	
	\begin{claim}\label{clm:ms-1st-matching-step}
		Fix $V_i$ and suppose step $i$ is a \underline{matching-step} and the event of~\Cref{clm:ms-1st-v} has happened. Then, 
		\[
			\Pr_{e_i}\paren{e_i \in \ME{i+1} \mid V_i} \geq \frac{1}{3}. 
		\] 
	\end{claim}
	\begin{proof}
		As $N(V_i)$ contains more than $(\opt/2\alpha)$ vertices (as this is matching-step) while $\Measy$ has at most $(\opt/4\alpha)$ vertices (as $\eventmatch(i)$ has happened), we know that at least half the vertices in $N(V_i)$ are unmatched. Given that all of $V_i$ is also unmatched, 
		if $\NES(G,V_i)$ samples $e_i$ to any of the unmatched vertices in $N(V_i)$, we can include $e_i$ in $\ME{i+1}$ greedily. As the choice of $\NES(G,V_i)$ is uniform over $N(V_i)$, this event 
		happens with probability at least $(1/2 - \delta_F) > 1/3$, as desired (since $\delta_F=1/100$).  \Qed{\Cref{clm:ms-1st-matching-step}}
		
	\end{proof}
	
	\begin{claim}\label{clm:ms-1st-vertex-step}
		Fix $V_i$ and suppose step $i$ is a \underline{vertex-step} and the event of~\Cref{clm:ms-1st-v} has happened. Then, 
		\[
			\Pr_{e_i}\paren{v \in V(\ME{i+1}) \mid V_i} \geq \frac{\alpha}{\opt}. 
		\] 
	\end{claim}
	\begin{proof}
		We know $v \in N(V_i)$ and that size of $N(V_i)$ is at most $(\opt/2\alpha)$. At the same time, since $V_i$ is clean, if $v$ is sampled as an endpoint of $e_i$ by $\NES(G,V_i)$, the edge $e_i$ will join the matching greedily and thus 
		$v$ will be matched. As the choice of $\NES(G,V_i)$ is uniform over $N(V_i)$ and $\delta_F =1/100$, 
		\[
			\Pr_{e_i}\paren{v \in V(\ME{i+1}) \mid V_i} \geq (1-\delta_F) \cdot \frac{1}{\card{N(V_i)}}  > \frac{\alpha}{\opt}. \Qed{\Cref{clm:ms-1st-vertex-step}}
		\]
		
	\end{proof}

	We can now conclude the proof of~\Cref{lem:ms-1st-batch} as follows. We have that at least half the steps are matching-steps or half of them are vertex-steps. We consider each case as follows. 
	
	\paragraph{When half the steps are matching-steps.} In this case, each matching-step $i$ increases size of $\ME{i}$ by one 
	with probability at least $(1/48)$ by~\Cref{clm:ms-1st-v,clm:ms-1st-matching-step}. Thus, 
	\[
		\Exp\card{\ME{s+1}} \geq (\frac s2) \cdot \frac{1}{48} = \frac{\opt^2}{(\alpha \cdot \log{n})^3} \cdot \frac{1}{48} \gg \opt/\alpha,
	\]
	given that $\opt \gg \alpha^2$ by~\Cref{assumption:guess}. Moreover, the distribution of $\ME{s+1}$ statistically dominates sum of $(s/2)$ Bernoulli random variables with mean $(1/48)$. 
	As such, by the Chernoff bound (\Cref{prop:chernoff}),
	\[
	\Pr\paren{\ME{s+1} < (\opt/8\alpha)} < \exp\paren{-\opt/\alpha} \ll 1/\poly{(n)},
	\]
	as $\opt \gg \alpha$ by~\Cref{assumption:guess}. This implies that $\eventmatch(s+1)$ happens, proving~\Cref{lem:ms-1st-batch} in this case. 
	
	\paragraph{When half the steps are vertex-steps.} In this case, each vertex-step $i$ can independently match the vertex $v$ with probability at least $(\alpha/16\,\opt)$ by~\Cref{clm:ms-1st-v,clm:ms-1st-vertex-step}. Thus, 
	\[
		\Pr\paren{v \in \High{s+1}} \leq (1-\frac{\alpha}{16\,\opt})^{s/2} \leq \exp\paren{-\frac{\alpha}{16\,\opt} \cdot  \frac{\opt^2}{2 \cdot (\alpha \cdot \log{n})^3}} < \exp\paren{-\frac{n^{\delta}}{32}} \ll 1/\poly{(n)},
	\]
	where we use $\opt \geq \alpha^2 \cdot n^{\delta}$ by~\Cref{assumption:range}. 
	Thus, with high probability $v$ will not be part of $\High{s+1}$. A union bound over all the vertices $v \in V$ then ensures that $\High{s+1}$ will be empty with high probability, thus proving~\Cref{lem:ms-1st-batch} in this case too.
	
	\smallskip
	
	\emph{Remark:} We note that the definition of matching-steps and vertex-steps are 
	tailored to individual vertices in $V$; however, even if one vertex leads to having at least half of the steps as matching-steps, we can apply the argument of first part and conclude the proof. Thus, when applying the second part of the argument, we can assume that all vertices lead to half of the steps being vertex-steps, and so we can union bound over all of them.

\subsection{Second Batch: Proof of~\Cref{lem:ms-2nd-batch}}

We now prove~\Cref{lem:ms-2nd-batch}. In the following, we condition on the event that $\High{s+1}$ (and $\High{i}$ for every $i \in (s,2s]$) is empty. 
Our goal is then to prove that at some step $i \in (s,2s]$, one of the events $\eventmatch(i)$ or $\eventsparsify(i)$ is not going to happen. 
The key to the proof of~\Cref{lem:ms-2nd-batch} (and~\Cref{lem:ms} itself) is the following. 

\begin{lemma}\label{lem:ms-increase}
	For any $i \in (s\,,\,2s]$, 
	\[
		\Pr_{(V_i,e_i)}\Paren{\ME{i+1} > \ME{i}} \geq \frac{\alpha^2 \cdot \log^3\!{n}}{4 \cdot \opt}.
	\] 
\end{lemma}
 
We first identify a simple structure in the graph $\GE{i}$. The following claim is based on a standard low-degree orientation of the graph plus geometric grouping 
of degrees of vertices.

\begin{claim}\label{clm:ms-geasy-large}
	\underline{At least one} of the following two conditions is true about $\GE{i}$: 
	\begin{enumerate}[label=$(\roman*)$]
		\item for some $d \in \left[\alpha \cdot \log^3\!{n},\dfrac{\opt}{8\alpha}\right)$, there are $\left( \dfrac{\opt \cdot \log^3{n}}{2d} \right)$ vertices $v$ in $\Med{i}$ with $\degG{v}{i} \geq d$;  
		\item for some $d \in [1,\alpha \cdot \log^3\!{n})$, there are $\left(\dfrac{19\,\opt\cdot\log^3{n}}{2d} \right)$ vertices in $\Low{i}$ with at least $d$ neighbors in $\Low{i}$. 
	\end{enumerate}
\end{claim}
\begin{proof}
	Given that $\High{i}$ is empty, any edge in $\GE{i}$ is either incident on $\Med{i}$ or is between two vertices in $\Low{i}$. Consequently, given that by $\eventsparsify(i)$, we have at least 
	$(20\,\opt\cdot\log^4\!{n})$ edges in $\GE{i}$, there are either at least $(\opt \cdot \log^4\!{n})$ edges incident on $\Med{i}$ or $(19 \cdot \opt \cdot \log^4\!{n})$ edges entirely inside $\Low{i}$. 
	We prove that each case corresponds to one of the conditions in the claim.
	
	\paragraph{When $\geq (\opt \cdot \log^4\!{n})$ edges are incident on $\Med{i}$.} We partition vertices in $\Med{i}$ into sets $\set{D_j}$ where each $D_j$ contains  vertices $v$ with $\degG{v}{i} \in [2^j , 2^{j+1})$. 
	As such, 
	\begin{align*}
		\sum_{j} \card{D_j} \cdot 2^{j+1} \geq \text{\# edges incident on $\Med{i}$} \geq \opt \cdot \log^4\!{n}.
	\end{align*}
	As there are at most $\log{n}$ choices for $j$ in the summation above,  we should have some $D_{j^*}$ with 
	\[
	\card{D_{j^*}} \geq \dfrac{\opt \cdot \log^3\!{n}}{2^{j^*+1}}.
	\]
	Setting $d = 2^{j^*}$ and returning (a subset of) $D_{j^*}$ satisfies the bound in part $(i)$ of the claim: all vertices in $D_{j^*} \subseteq \Med{i}$ have $\degG{\cdot}{i}$ in $[\alpha \cdot \log^3\!{n},\dfrac{\opt}{8\alpha})$ by definition of $\Med{i}$, and we can pick
	a subset of $D_{j^*}$ with size prescribed by the claim as all vertices in $D_{j^*}$ have degree $d$ at least. 

	\paragraph{When $\geq (19\,\opt \cdot \log^4\!{n})$ edges are entirely inside $\Low{i}$.} The argument is almost identical to the above part by counting the degree of vertices in $\Low{i}$ but 
	only in $\Low{i}$ (instead of all of $\degG{\cdot}{i}$ as in the previous part). We partition vertices in $\Low{i}$ into sets $\set{D_j}$ where each $D_j$ contains all vertices with number of neighbors in $\Low{i}$ in $[2^j,2^{j+1})$. 
	As such, 
	\begin{align*}
		\sum_{j} \card{D_j} \cdot 2^{j+1} \geq \text{\# edges entirely inside $\Low{i}$} \geq 19\,\opt \cdot \log^4\!{n}.
	\end{align*}
	As there are at most $\log{n}$ choices for $j$ in the summation above, we should have some $D_{j^*}$ with 
	\[
	\card{D_{j^*}} \geq \dfrac{19\,\opt \cdot \log^3\!{n}}{2^{j+1}}.
	\]
	Setting $d = 2^{j^*}$ and returning (a subset of) $D_{j^*}$ satisfies the bound in part $(ii)$ of the claim: all vertices in $D^* \subseteq \Low{i}$ have degree less than $(\alpha \cdot \log^3\!{n})$ by the definition of $\Low{i}$ 
	(even in $\GE{i}$ and so between $\Low{i}$ also) and 
	we can pick a subset of $D_{j^*}$ with the required size as vertices in $D_{j^*}$ have degree $d$ at least. \Qed{\Cref{clm:ms-geasy-large}}
	
\end{proof}

In the following, we refer to a step $i \in (s,2s]$ as a \textbf{$\bm{V_i}$-step} whenever case $(i)$ of~\Cref{clm:ms-geasy-large} happens and a \textbf{$\bm{N(V_i)}$-step} otherwise. We will show that: 
\begin{itemize}
	 \item In a \textbf{$\bm{V_i}$-step}, we have ``enough'' large degree vertices and even if we sample one of them in $V_i$ it will make the intersection of $N(V_i)$ and $\GE{i}$ large;
	\item In a \textbf{$\bm{N(V_i)}$-step}, we have ``so many'' low degree vertices in $\GE{i}$ that many of them will appear in $N(V_i)$ and thus there is a large intersection 
between $N(V_i)$ and $\GE{i}$ again. 
\end{itemize}
In each case, we can  finalize the proof by showing that having $N(V_i)$ intersect largely with $\GE{i}$ allows us to recover an edge $e_i$ via $\NES(G,V_i)$ 
that can increase size of $\ME{i}$ with sufficiently large probability. 

\subsubsection*{Case $(i)$ of~\Cref{clm:ms-geasy-large}: \textbf{$\bm{V_i}$-steps}} 
Let 
\begin{align}
d \in [\alpha\log^3\!{n},\dfrac{\opt}{8\alpha}) \quad \text{and} \quad D \subseteq \Med{i} \quad \text{with} \quad \card{D} = \frac{\opt\cdot\log^3{n}}{2d} \label{eq:ms-Vsteps-1}
\end{align}
be, respectively, the degree-parameter and corresponding set guaranteed by Case $(i)$ of~\Cref{clm:ms-geasy-large}. The following claim lower bounds the probability that $V_i$ is both clean and  samples a vertex from $D$. 

\begin{claim}\label{clm:ms-Vstep-1}
	$\Pr_{V_i}\paren{V_i \cap D \neq \emptyset~\textnormal{and $V_i$ is clean} } \geq \dfrac{\alpha \cdot \log^3\!{n}}{8d}$. 
\end{claim}
\begin{proof}
	We have, 
\begin{align*}
	\Pr\paren{V_i \cap D \neq \emptyset ~\textnormal{and $V_i$ is clean}} &\geq \sum_{v \in D} \Pr\paren{v \in V_i ~\textnormal{and $V_i$ is clean}} \, - \sum_{u \neq w \in D} \Pr\paren{u,w \in V_i} 
	\tag{by inclusion-exclusion principle and dropping the `intersection' from the second event} \\
	&\geq \card{D} \cdot \frac{3}{4k}  - \card{D}^2 \cdot \frac{1}{k^2} \tag{by~\Cref{clm:ms-helper-1} and as $h_i(\cdot)$ is a pair-wise independent hash function with range $[k]$} \\
	&= \frac{\opt\cdot\log^3{n}}{2d} \cdot \frac{\alpha}{\opt} \cdot \paren{\frac{3}{4} - \frac{\opt\cdot\log^3{n} \cdot \alpha}{2d \cdot \opt}} \tag{by the choice of $k=\opt/\alpha$ and size of $D$ in~\Cref{eq:ms-Vsteps-1}} \\
	&\geq \frac{\alpha\cdot\log^3{n}}{8d},
\end{align*}
as $d \geq \alpha\log^3\!{n}$ by~\Cref{eq:ms-Vsteps-1}.  \Qed{\Cref{clm:ms-Vstep-1}}

\end{proof}

Let us now condition on the choice of $V_i$ and assume the event of~\Cref{clm:ms-Vstep-1} happens. Given that any vertex in $D$ already has $d$ neighbors in $\GE{i}$, 
we have that $N(V_i) \cap \GE{i}$ has size at least $d$ in this case. On the other hand, $N(V_i)$ can have at most $(\opt/4\alpha)$ neighbors outside $\GE{i}$ by the bound on the total number of 
matched vertices by~$\eventmatch(i)$. As the choice of $e_i$ from $\NES(G,V_i)$ is uniform over $N(V_i)$, we have, 
\begin{align*}
	\Pr_{e_i}\paren{\text{$e_i$ is from $V_i$ to $N(V_i) \cap \GE{i}$} \mid V_i} &\geq (1-\delta_F) \cdot \frac{\card{N(V_i) \cap \GE{i}}}{\card{N(V_i)}} \\
	&\geq (1-\delta_F) \cdot \frac{d}{(\opt/4\alpha) + d} \\
	&\geq (1-\delta_F) \cdot \frac{d \cdot 8\alpha}{3\opt} \tag{as $d \leq (\opt/8\alpha)$ in~\Cref{eq:ms-Vsteps-1}} \\
	&\geq \frac{2d \cdot \alpha}{\opt},
\end{align*}
as $\delta_F < 1/4$. Given that all of $V_i$ is also unmatched (as $V_i$ is clean by conditioning on the event of~\Cref{clm:ms-Vstep-1}), 
we can include $e_i$ in $\ME{i+1}$ greedily whenever $e_i$ is between $V_i$ and $N(V_i) \cap \GE{i}$.

Consequently, combining the two events above, we have, 
\begin{align*}
	\Pr_{(V_i,e_i)}\Paren{\ME{i+1} > \ME{i}} &\geq \Pr_{V_i}\paren{V_i \cap D \neq \emptyset~\textnormal{and $V_i$ is clean} } \cdot \Pr_{e_i}\paren{\text{$e_i$ is from $V_i$ to $N(V_i) \cap \GE{i}$} \mid V_i} \\
	&\geq  \dfrac{\alpha\log^3\!{n}}{8d} \cdot \frac{2d \cdot \alpha}{\opt} = \frac{\alpha^2\log^3\!{n}}{4\,\opt}.
\end{align*}
This concludes the proof of~\Cref{lem:ms-increase} in this case.

\subsubsection*{Case $(ii)$ of~\Cref{clm:ms-geasy-large}: \textbf{$\bm{N(V_i)}$-steps}}

 Let 
\begin{align}
d \in [1,\alpha\log^3\!{n}) \quad \text{and} \quad D \subseteq \Low{i} \quad \text{with} \quad \card{D} = \frac{19\,\opt\cdot\log^3{n}}{2d} \label{eq:ms-NVsteps-1}
\end{align}
be, respectively, the degree-parameter and corresponding set guaranteed by Case $(ii)$ of~\Cref{clm:ms-geasy-large}. 
For the rest of this analysis, we focus only on the subgraph of $\GE{i}$ induced on vertices of $\Low{i}$ and for each  $v \in D$, we pick exactly $d$ (arbitrary) neighbors from $\Low{i}$ and 
denote them by $\NL(v)$.
Our goal is to show that $N(V_i)$ and $\GE{i}$  intersect largely.
We will do so by counting the elements in $D$ that have neighbors in $V_i$. This works because $D \subseteq \GE{i}$ and having neighbors in $V_i$ means that the vertex itself is in $N(V_i)$.

For any vertex $v \in D$, define an indicator random variable $X_v \in \set{0,1}$ which is $1$ iff $\NL(v) \cap V_i \neq \emptyset$ (see \Cref{fig:MorS_a}). Notice that $X= \sum_{v \in D} X_v$ is a random variable that denotes the number of vertices $v$ in $D$ that have a neighbor in $\NL(v)$ that belongs to $V_i$.
Note that we do not consider all neighbors of $v$ in $\Low{i}$, only the ones in $NL(v)$; this is okay since we just need a lower bound on $\card{N(V_i) \cap \GE{i}}$.
It is easy to see that $X \leq \card{N(V_i) \cap \GE{i}}$ since $v$ contributes to $\card{N(V_i) \cap \GE{i}}$ if $X_v=1$ (see \Cref{fig:MorS_b}).
We first bound the probability of the event $\NL(v) \cap V_i \neq \emptyset$.

\begin{figure}[h!]
	\centering
	\subcaptionbox{This figure shows the set of vertices $\Low{i}$ and its subset $D$. $V_i$ (in blue) is the set of sampled vertices in step $i$. For a vertex $u$ in $D$ we define a set of neighbors $\NL(u)$ which if intersects with $V_i$ then we have random variable $X_u=1$.  \label{fig:MorS_a}}%
	[.45\linewidth]{ 	\resizebox{160pt}{110pt}{\tikzset{every picture/.style={line width=0.75pt}} 

\begin{tikzpicture}[x=0.75pt,y=0.75pt,yscale=-1,xscale=1]

\draw   (217,17) -- (408,17) -- (408,187) -- (217,187) -- cycle ;
\draw  [color={rgb, 255:red, 25; green, 2; blue, 208 }  ,draw opacity=1 ] (266,156.5) .. controls (266,125.85) and (313.23,101) .. (371.5,101) .. controls (429.77,101) and (477,125.85) .. (477,156.5) .. controls (477,187.15) and (429.77,212) .. (371.5,212) .. controls (313.23,212) and (266,187.15) .. (266,156.5) -- cycle ;
\draw   (231,86) .. controls (231,66.12) and (247.12,50) .. (267,50) .. controls (286.88,50) and (303,66.12) .. (303,86) .. controls (303,105.88) and (286.88,122) .. (267,122) .. controls (247.12,122) and (231,105.88) .. (231,86) -- cycle ;
\draw   (283,98) .. controls (283,95.51) and (285.01,93.5) .. (287.5,93.5) .. controls (289.99,93.5) and (292,95.51) .. (292,98) .. controls (292,100.49) and (289.99,102.5) .. (287.5,102.5) .. controls (285.01,102.5) and (283,100.49) .. (283,98) -- cycle ;
\draw   (306.5,98) .. controls (306.5,85.3) and (312.66,75) .. (320.25,75) .. controls (327.84,75) and (334,85.3) .. (334,98) .. controls (334,110.7) and (327.84,121) .. (320.25,121) .. controls (312.66,121) and (306.5,110.7) .. (306.5,98) -- cycle ;
\draw    (287.5,102.5) -- (320.25,121) ;
\draw    (287.5,93.5) -- (320.25,75) ;

\draw (274,82) node [anchor=north west][inner sep=0.75pt]   [align=left] {$u$};
\draw (337,71) node [anchor=north west][inner sep=0.75pt]   [align=left] {$\NL(u)$};
\draw (477,169) node [anchor=north west][inner sep=0.75pt]   [font=\large,color={rgb, 255:red, 25; green, 2; blue, 208 }  ,opacity=1 ] [align=left] {$V_i$};
\draw (415,25) node [anchor=north west][inner sep=0.75pt]   [align=left] {$\Low{i}$};
\draw (235,40) node [anchor=north west][inner sep=0.75pt]   [align=left] {$D$};

\end{tikzpicture}}}
	\hspace{0.4cm}
	\subcaptionbox{This figure shows $X \leq \card{N(V_i) \cap \GE{i}}$. The vertices $u_j$ are in $\Low{i}$. Notice that $X_{u_1}=1,X_{u_2}=1,X_{u_3}=0$ and $X_{u_4}$ is not defined so we have $X=2$. But $\card{N(V_i) \cap \GE{i}}=~3$ since $u_1,u_2$ and $u_4$ contribute to it.  \label{fig:MorS_b}}%
	[.45\linewidth]{ \resizebox{200pt}{100pt}{\tikzset{every picture/.style={line width=0.75pt}} 

\begin{tikzpicture}[x=0.75pt,y=0.75pt,yscale=-1,xscale=1]

\draw   (153,20) -- (441,20) -- (441,190) -- (153,190) -- cycle ;
\draw [color={rgb, 255:red, 25; green, 2; blue, 208 }  ,draw opacity=1 ]   (265,131.5) .. controls (265,100.85) and (312.23,76) .. (370.5,76) .. controls (428.77,76) and (476,100.85) .. (476,131.5) .. controls (476,162.15) and (428.77,187) .. (370.5,187) .. controls (312.23,187) and (265,162.15) .. (265,131.5) -- cycle ;
\draw   (206.5,82.25) .. controls (206.5,55.6) and (228.1,34) .. (254.75,34) .. controls (281.4,34) and (303,55.6) .. (303,82.25) .. controls (303,108.9) and (281.4,130.5) .. (254.75,130.5) .. controls (228.1,130.5) and (206.5,108.9) .. (206.5,82.25) -- cycle ;
\draw   (283,71) .. controls (283,68.51) and (285.01,66.5) .. (287.5,66.5) .. controls (289.99,66.5) and (292,68.51) .. (292,71) .. controls (292,73.49) and (289.99,75.5) .. (287.5,75.5) .. controls (285.01,75.5) and (283,73.49) .. (283,71) -- cycle ;
\draw   (306.5,71) .. controls (306.5,58.3) and (312.66,48) .. (320.25,48) .. controls (327.84,48) and (334,58.3) .. (334,71) .. controls (334,83.7) and (327.84,94) .. (320.25,94) .. controls (312.66,94) and (306.5,83.7) .. (306.5,71) -- cycle ;
\draw    (287.5,75.5) -- (320.25,94) ;
\draw    (287.5,66.5) -- (320.25,48) ;
\draw   (248,102) .. controls (248,99.51) and (250.01,97.5) .. (252.5,97.5) .. controls (254.99,97.5) and (257,99.51) .. (257,102) .. controls (257,104.49) and (254.99,106.5) .. (252.5,106.5) .. controls (250.01,106.5) and (248,104.49) .. (248,102) -- cycle ;
\draw   (276.41,98.54) .. controls (281.61,93.04) and (288.64,91.24) .. (292.12,94.53) .. controls (295.6,97.81) and (294.2,104.94) .. (289,110.44) .. controls (283.8,115.94) and (276.77,117.73) .. (273.3,114.45) .. controls (269.82,111.16) and (271.21,104.04) .. (276.41,98.54) -- cycle ;
\draw    (252.5,106.5) -- (273.3,114.45) ;
\draw    (252.5,97.5) -- (288,91.5) ;
\draw   (222,75) .. controls (222,72.51) and (224.01,70.5) .. (226.5,70.5) .. controls (228.99,70.5) and (231,72.51) .. (231,75) .. controls (231,77.49) and (228.99,79.5) .. (226.5,79.5) .. controls (224.01,79.5) and (222,77.49) .. (222,75) -- cycle ;
\draw   (167.5,83) .. controls (167.5,70.3) and (173.66,60) .. (181.25,60) .. controls (188.84,60) and (195,70.3) .. (195,83) .. controls (195,95.7) and (188.84,106) .. (181.25,106) .. controls (173.66,106) and (167.5,95.7) .. (167.5,83) -- cycle ;
\draw    (226.5,79.5) -- (181.25,106) ;
\draw    (226.5,70.5) -- (181.25,60) ;
\draw   (392.06,45.46) .. controls (394.55,45.5) and (396.54,47.54) .. (396.5,50.02) .. controls (396.47,52.51) and (394.43,54.5) .. (391.94,54.46) .. controls (389.46,54.43) and (387.47,52.39) .. (387.5,49.9) .. controls (387.54,47.42) and (389.58,45.43) .. (392.06,45.46) -- cycle ;
\draw   (391.75,68.96) .. controls (404.45,69.13) and (414.67,75.42) .. (414.57,83.02) .. controls (414.47,90.61) and (404.09,96.63) .. (391.39,96.46) .. controls (378.69,96.29) and (368.47,90) .. (368.57,82.41) .. controls (368.67,74.81) and (379.05,68.79) .. (391.75,68.96) -- cycle ;
\draw    (387.5,49.9) -- (368.57,82.41) ;
\draw    (396.5,50.02) -- (414.57,83.02) ;

\draw (269,48) node [anchor=north west][inner sep=0.75pt]   [align=left] {$u_2$};
\draw (474,150) node [anchor=north west][inner sep=0.75pt]  [font=\large,color={rgb, 255:red, 25; green, 2; blue, 208 }  ,opacity=1 ]  [align=left] {$V_i$};
\draw (451,22) node [anchor=north west][inner sep=0.75pt]   [align=left] {$\Low{i}$};
\draw (284,23) node [anchor=north west][inner sep=0.75pt]   [align=left] {$D$};
\draw (236,109) node [anchor=north west][inner sep=0.75pt]   [align=left] {$u_1$};
\draw (221,50) node [anchor=north west][inner sep=0.75pt]   [align=left] {$u_3$};
\draw (403,33) node [anchor=north west][inner sep=0.75pt]   [align=left] {$u_4$};

\end{tikzpicture}}}
	\caption{Illustration of random variables $X_v$.}
	\label{fig:MorS}
\end{figure}
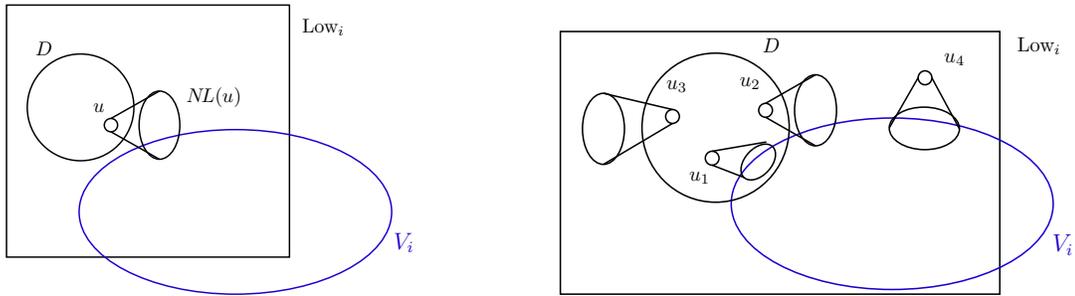

\begin{claim}\label{clm:ms-NVsteps-1}
	For any $v \in D$, 
	\[
	(1-o(1)) \cdot \frac{d \cdot \alpha}{\opt} \leq \Pr\paren{X_v = 1} \leq \frac{d \cdot \alpha}{\opt}.
	\]
\end{claim}
\begin{proof}
	$X_v = 1$ iff one of the neighbors of $v$ in $\NL(v)$ belongs to $V_i$. For the upper bound, by union bound, 
	\begin{align*}
		\Pr\paren{X_v = 1} \leq \sum_{u \in \NL(v)} \Pr\paren{u \in V_i} = \card{\NL(v)} \cdot \frac{1}{k} = \frac{d \cdot \alpha}{\opt}. \tag{as $\card{\NL(v)} = d$ and $k = \opt/\alpha$}
	\end{align*}
	For the lower bound, by inclusion-exclusion principle, 
	\begin{align*}
		\Pr\paren{X_v = 1} &\geq \sum_{u \in \NL(v)} \Pr\paren{u \in V_i}  - \sum_{u \neq w \in \NL(v)} \Pr\paren{u,w \in V_i} \\
		&\geq \card{\NL(v)} \cdot \frac{1}{k} - \card{\NL(v)}^2 \cdot \frac{1}{k^2} \\
		&= \frac{d \cdot \alpha}{\opt} \cdot \paren{1-\frac{d \cdot \alpha}{\opt}} \tag{as $\card{\NL(v)} = d$ and $k = \opt/\alpha$} \\
		&\geq (1-o(1)) \cdot \frac{d \cdot \alpha}{\opt},
	\end{align*}
	as $d < \alpha\log^3{n}$ by~\Cref{eq:ms-NVsteps-1} and $\opt \geq \alpha^2 \cdot n^{\delta}$ by~\Cref{assumption:range}. \Qed{\Cref{clm:ms-NVsteps-1}}
	
\end{proof}

By~\Cref{clm:ms-NVsteps-1} and the size of $D$ in~\Cref{eq:ms-NVsteps-1}, we have, 
\begin{align}
	(1-o(1)) \cdot \frac{19}{2} \cdot \alpha \cdot \log^3{n} \leq \expect{X} \leq  \frac{19}{2} \cdot \alpha \cdot \log^3{n}. \label{eq:ms-NVsteps-exp}
\end{align}
Our goal now is to prove that $X$ is  concentrated. This requires a non-trivial proof as the variables $\set{X_v}_{v \in D}$ are correlated through their shared neighbors in $V_i$. But the fact that the subgraph 
induced on $\Low{i}$ is low-degree allows us to bound the variance of $X$ using a combinatorial argument in the following claim.  

\begin{claim}\label{clm:ms-NVsteps-var}
	$\var{X} \leq (1/8) \cdot \expect{X}^2$. 
\end{claim}
\begin{proof}
For any two vertices $u \neq v \in D$, define $\Com(u,v) := \NL(u) \cap \NL(v)$ as the set of common
neighbors of $u$ and $v$ in subgraph of $\Low{i}$ defined by $\NL(\cdot)$ and let $\com(u,v) = \card{\Com(u,v)}$. We have, 
\begin{align}
	\var{X} &= \sum_{v \in D} \var{X_v} + \sum_{u \neq v \in D} \cov{X_u,X_v} \leq \expect{X} + \sum_{u \neq v \in D} \cov{X_u,X_v},
	\label{eq:ms-NVsteps-var}
\end{align}
as $X_v$ is an indicator random variable and thus $\var{X_v} \leq \expect{X_v}$. 
We thus need to bound the covariance-terms only. Recall that 
\[
\cov{X_u,X_v} = \expect{X_u \cdot X_v} - \expect{X_u}\expect{X_v} = \Pr\paren{\NL(u) \cap V_i \neq \emptyset \wedge \NL(v) \cap V_i \neq \emptyset} - \Pr\paren{X_u=1} \cdot \Pr\paren{X_v=1}.
\] 
We can bound the second part using~\Cref{clm:ms-NVsteps-1} for each probability-term. For the first part, notice that  for
$\NL(u) \cap V_i \neq \emptyset$ and $\NL(v) \cap V_i \neq \emptyset$ one of the following two things should happen: at least one of the shared neighbors of $u,v$ in $\Com(u,v)$ is chosen in $V_i$ or each of them separately have 
a neighbor in $\NL(u) - \Com(u,v)$ and $\NL(v) - \Com(u,v)$ those join $V_i$ (as $h_i(\cdot)$ is a pair-wise independent hash function, the probability of these two distinct vertices joining $V_i$ is independent). Thus, 
\begin{align*}
	\Pr\paren{\NL(u) \cap V_i \neq \emptyset \wedge \NL(v) \cap V_i \neq \emptyset} &\leq \sum_{w \in \Com(u,v)} \hspace{-0.5cm}\Pr\paren{w \in V_i} + 
	{\hspace{-0.5cm} \sum_{\substack{z_u \in \NL(u) - \Com(u,v)\\ z_v \in \NL(v) - \Com(u,v)}} \hspace{-1cm} \Pr\paren{z_u \in V_i} \cdot \Pr\paren{z_v \in V_i}} \\
	&\leq \com(u,v) \cdot \frac{1}{k} + d^2 \cdot \frac{1}{k^2} \tag{as $h_i(\cdot)$ is uniform over $[k]$ and $u,v \in D$ and each vertex in $D$ has exactly $d$ neighbors in $\NL(\cdot)$} \\
	&= \com(u,v) \cdot \frac{\alpha}{\opt} + \frac{d^2 \cdot \alpha^2}{\opt^2} \tag{as $k = \opt/\alpha$}. 
\end{align*}
Plugging in this for the first term of covariance and the bounds in~\Cref{clm:ms-NVsteps-1} for the second terms, we have, 
\[
	\cov{X_u,X_v} \leq \com(u,v) \cdot \frac{\alpha}{\opt} + \frac{d^2 \cdot \alpha^2}{\opt^2} - (1-o(1)) \cdot \frac{d^2 \cdot \alpha^2}{\opt^2} = \com(u,v) \cdot \frac{\alpha}{\opt} + o(1) \cdot \frac{d^2 \cdot \alpha^2}{\opt^2}. 
\]
By plugging in further in the RHS of~\Cref{eq:ms-NVsteps-var}, we get that, 
\begin{align*}
	\var{X} &\leq \expect{X} + \card{D^2} \cdot o(1) \cdot \frac{d^2 \cdot \alpha^2}{\opt^2} +  \sum_{u \neq v \in D} \com(u,v) \cdot \frac{\alpha}{\opt} \\
	&\leq \expect{X} + \paren{\frac{19 \cdot \opt \cdot \log^3{n}}{2d}}^2 \cdot o(1) \cdot \frac{d^2 \cdot \alpha^2}{\opt^2} + \sum_{u \neq v \in D} \com(u,v) \cdot \frac{\alpha}{\opt} \tag{by the bound on size of $D$ in~\Cref{eq:ms-NVsteps-1}} \\
	&\leq o(1) \cdot \expect{X}^2 +  \frac{\alpha}{\opt} \cdot \sum_{u \neq v \in D} \com(u,v). \tag{by the lower bound on $\expect{X} \geq (1-o(1)) \cdot (19/2) \cdot \alpha \cdot \log^3{n}$ in~\Cref{eq:ms-NVsteps-exp}}
\end{align*}
The remaining part is then to compute the summation in the RHS which we do below using a double-counting argument. 
Note that $\sum_{u \neq v \in D} \com(u,v)$ counts the number of common neighbors inside $\NL(\cdot)$-subgraph of $\Low{i}$ for each pair of vertices in $D$. This can be alternatively 
counted by going over vertices in $\Low{i}$ that are neighbor to $D$ and count the number of pairs of neighbors (in $\NL(\cdot)$) they have in $D$. 
\begin{align*}
	\sum_{u \neq v \in D} \com(u,v) &= \sum_{z \in \NL(D)} {{\card{\NL(z)}}\choose{2}} \leq \sum_{z \in \NL(D)} \card{\NL(z)}^2 \\
	&\leq (\alpha \cdot \log^3{n}) \cdot \sum_{z \in \NL(D)} \card{\NL(z)} \tag{each vertex in $\Low{i}$ has degree at most $(\alpha\cdot\log^3{n})$ to $\Low{i}$ and $z \in \Low{i}$ as it is in $\NL(D)$} \\
	&\leq (\alpha \cdot \log^3{n}) \cdot \card{D} \cdot d \tag{as the sum-term counts the number of edges between $D$ and $\NL(D)$ which is at most $\card{D} \cdot d$}\\
	&\leq (\alpha \cdot \log^3{n}) \cdot (\frac{19}{2} \cdot \opt \cdot \log^3{n}). \tag{by~\Cref{eq:ms-NVsteps-1} on the size of $D$}
\end{align*}
Plugging in this bound in the upper bound on $\var{X}$ in the earlier equation, we have, 
\begin{align*}
	\var{X} &\leq o(1) \cdot \expect{X}^2 +  \frac{\alpha}{\opt} \cdot \paren{\frac{19}{2} \cdot \opt \cdot \alpha \cdot \log^6{n}} \\
	&= o(1) \cdot \expect{X}^2 +  \frac{19}{2} \cdot \alpha^2 \cdot \log^6{n} \\
	&< \frac{1}{8} \cdot \expect{X}^2,  
\end{align*}
as $\expect{X} \geq  (1-o(1)) \cdot \dfrac{19}{2} \cdot (\alpha\cdot\log^3{n})$ by~\Cref{eq:ms-NVsteps-exp}.  \Qed{\Cref{clm:ms-NVsteps-var}} 

\end{proof}

Recall that 
\[
\card{N(V_i) \cap \GE{i}} \geq X.
\]
Given the bound on expectation and variance of $X$ in~\Cref{eq:ms-NVsteps-exp} and~\Cref{clm:ms-NVsteps-var}, respectively, we can now 
apply Chebyshev's inequality and get that,
\begin{align*}
	\Pr\paren{\card{N(V_i) \cap \GE{i}} < 2\alpha \cdot \log^3{n}} \leq \Pr\paren{\card{X-\expect{X}} \geq \frac{15}{19} \cdot \expect{X}} \leq \frac{19^2 \cdot \var{X}}{15^2\cdot\expect{X}^2} \leq \frac{19^2}{15^2 \cdot 8} < \frac{1}{4}. 
\end{align*}
Additionally, we also have that the probability that $V_i$ is not clean is at most, 
\[
	\Pr\paren{\text{$V_i$ is not clean}} \leq \sum_{u \in V(\ME{i})} \Pr\paren{u \in V_i} < \frac{\opt}{4\alpha} \cdot \frac{\alpha}{\opt} = \frac{1}{4}. 
\]
By a union bound on the two equations above, we have, 
\begin{align}
\Pr\paren{\card{N(V_i) \cap \GE{i}} \geq 2\alpha \cdot \log^3{n} ~\text{and $V_i$ is clean}} \geq \frac12. 
\label{eq:ms-NVsteps-cond}
\end{align}

The rest of the proof is similar to that of \textbf{$\bm{V_i}$-steps}. We condition the choice of $V_i$ and assume the event of~\Cref{eq:ms-NVsteps-cond} has happened. 
Thus, we have that both $V_i$ is clean and $N(V_i)$ has at least $2\alpha \cdot \log^3{n}$ vertices in $\GE{i}$. 
Moreover, $N(V_i)$ can have at most $(\opt/4\alpha)$ neighbors outside $\GE{i}$ by the bound on the total number of 
matched vertices by~$\eventmatch(i)$. As the choice of $e_i$ from $\NES(G,V_i)$ is uniform over $N(V_i)$, we have, 
\begin{align*}
	\Pr_{e_i}\paren{\text{$e_i$ is from $V_i$ to $N(V_i) \cap \GE{i}$} \mid V_i} &\geq (1-\delta_F) \cdot \frac{\card{N(V_i) \cap \GE{i}}}{\card{N(V_i)}} \\
	&\geq (1-\delta_F) \cdot \frac{2\alpha \cdot \log^3{n}}{(\opt/4\alpha) + 2\alpha \cdot \log^3{n}} \\
	&> \frac{3 \cdot \alpha^2 \cdot \log^3{n}}{\opt},
\end{align*}
as $\opt \geq \alpha^2 \cdot n^{\delta}$ by~\Cref{assumption:guess} and $\delta_F < 1/2$. Given that all of $V_i$ is also unmatched (as $V_i$ is clean), 
we can include $e_i$ in $\ME{i+1}$ greedily whenever the event of the LHS above  happens. 

Consequently, combining the two events above, we have, 
\begin{align*}
	\Pr_{(V_i,e_i)}\Paren{\ME{i+1} > \ME{i}} &\geq \Pr_{V_i}\paren{\card{N(V_i) \cap \GE{i}} \geq 2\alpha \cdot \log^3{n} ~\text{and $V_i$ is clean}} \cdot \Pr_{e_i}\paren{\text{$e_i$ is from $V_i$ to $N(V_i) \cap \GE{i}$} \mid V_i} \\
	&\geq  \frac{1}{2} \cdot \frac{3 \cdot \alpha^2 \cdot \log^3{n}}{\opt} > \frac{\alpha^2\log^3\!{n}}{\opt}.
\end{align*}
This concludes the proof of~\Cref{lem:ms-increase} in this case also. 

\subsubsection*{Concluding the Proof of~\Cref{lem:ms-2nd-batch}}

 By~\Cref{lem:ms-increase}, assuming the events $\eventmatch(i),\eventsparsify(i)$ hold for every $i\in (s,2s]$, size of $\ME{2s}$ 
statistically dominates sum of $s$ independent Bernoulli random variables $\set{Z_i}_{i=1}^s$ with mean $({\alpha^2 \cdot \log^3\!{n}}/4\opt)$ (RHS of~\Cref{lem:ms-increase}). 
Let $Z = \sum_{i=1}^s Z_i$. Thus, by the choice of $s$ in~\Cref{alg:ms},
\[
\expect{Z} = \frac{\opt^2}{(\alpha \cdot \log{n})^3} \cdot \frac{\alpha^2 \cdot \log^3\!{n}}{4 \cdot \opt} = \frac{\opt}{4 \cdot \alpha},
\]
 and by the Chernoff bound (\Cref{prop:chernoff}), 
\begin{align*}
	\Pr\paren{Z< \frac{\opt}{8\alpha}} < \Pr\paren{Z < \frac{1}{2} \cdot \expect{Z}} \leq \exp\paren{-\frac{\opt}{12\alpha}} \ll 1/\poly{(n)}, 
\end{align*}
where the final bound is by~\Cref{assumption:range} as $\opt \geq n^{\delta} \cdot \alpha^2$. This means that as long as $\eventmatch(i),\eventsparsify(i)$ happen for all $i \in (s:2s]$, with high probability we are going 
to end up with a matching $\Measy$ of size at least $(\opt/8\alpha)$, which means $\eventmatch(2s)$ does not happen as desired. 

This concludes the proof of~\Cref{lem:ms}, and combined with~\Cref{lem:sc}, the entire proof of~\Cref{thm:main}.

\subsection*{Acknowledgements} 

We are grateful to Christian Konrad for helpful discussions and to the anonymous reviewers of ITCS 2022 for the valuable comments that helped with the presentation of this paper. 

\clearpage

\bibliographystyle{alpha}
\bibliography{new}

\clearpage

\appendix

\section{Sparse-Neighborhood Recovery via Exhaustive-Search} \label{app:sparse-recovery} 

We give an alternative and simpler sketch for~\Cref{prob:sparse-recovery} with the caveat that it uses too much randomness to store efficiently\footnote{While one can use the heavy-machinery yet standard PRG ideas to reduce this randomness (see, e.g.~\cite{KapralovLMMS14}), the use of PRGs will lead to an $O(\log{n})$
space overhead that will break the asymptotic optimality of the algorithm.}  and also requires exponential time. Even though this sketch does not work for the purpose of our 
algorithm, given that it is much simpler than our $\SNR$ sketch, we present it here as a warm-up. Formally, we prove the following lemma. 

\begin{lemma}\label{lem:inefficient} 
	 There is a linear sketch for~\Cref{prob:sparse-recovery} that uses sketch and randomness of size, respectively, 
	\[
	s_0 = s_0(n,a,b,c) = O(a \cdot \log{c} + b \cdot \log{n} \cdot \log{c})  \quad \textnormal{and} \quad O(n \cdot s_0) = O(n \cdot (a \cdot \log{q} + b \cdot \log{n} \cdot \log{q}))
	\]
	bits and outputs a wrong answer with probability at most $n^{-10}$. The algorithm requires exponential time (in parameters $a,b$ and $c$).  
\end{lemma}

We use the same vector representation of the problem defined in~\Cref{eq:x-GS} in~\Cref{sec:snr}. 
Our approach is essentially to run equality-test from communication complexity on the vector $x = x(G,S)$; at the end of the stream, once we know the set $T$, we can search over all possible choices for $x$, given the promises in~\Cref{prob:sparse-recovery} and return the one that passes the equality-test. Given that we can bound the number of choices for $x$, we can limit the number of equality-tests we need to run. 

We will be working in the field $\mathbb{F}_q$ throughout this subsection where $q$ is the smallest prime larger than $c$. In particular, all computations are in $\mathbb{F}_q$. As $q \geq c$, recovery of coordinates of $x$ outside $T$ 
under $\mathbb{F}_q$ is the same as recovery over the integers. The algorithm is as follows. 

\begin{Algorithm}\label{alg-sr-1}
	A simple but (somewhat) inefficient algorithm for~\Cref{prob:sparse-recovery}. 
	
	\medskip
	
	\textbf{Input:} A graph $G=(V,E)$ specified via $\vect(E)$ and a set $S \subseteq V$, defining the corresponding vector $x(G,S)$ in~\Cref{eq:x-GS}. A set $T \subseteq V$ specified at the \emph{end} of the stream. 
	
	\medskip
	
	\textbf{Output:} The set of neighbors of $S$ outside $T$, i.e., $N(S) - T$. 
	
	\medskip
	
	\textbf{Sketching Matrix:}
	\begin{enumerate}
		\item Let $s = 2 \cdot \log{\Paren{2^a \cdot q^a \cdot n^b \cdot q^b \cdot n^{10}}} \cdot \frac{1}{\log q}$.  
		\item For $i=1$ to $s$: compute $z_i = a_i \cdot x$ where $a_i$ is a vector sampled uniformly from $\mathbb{F}_q^n$.
	\end{enumerate}
	\textbf{Recovery:}
	\begin{enumerate}
		\item Given $T$, go over all possible vectors $y \in \mathbb{F}_q^n$, that satisfy the promises of~\Cref{prob:sparse-recovery} for $T$.   
		\item For each guess $y$, check if $a_i \cdot y = z_i$  for all $i \in [s]$.
		If it is equal for all $i \in [s]$ then go over $y$ and output its non-zero indices (vertices corresponding to the indices) which are not in $T$.
	\end{enumerate}
\end{Algorithm}
We now analyze the correctness of the algorithm. 
\begin{claim} \label{clm:rec-one-itr-failure}
	For any vector $y \neq x$ considered in~\Cref{alg-sr-1} and an iteration $i \in [s]$, we have $a_i \cdot y  \neq z_i$ with probability at least $1-{1}/{q}$.
\end{claim}
\begin{proof}
	As $y \neq x$, there should be a coordinate $j \in [n]$ where they have different values. The only way for $a_i \cdot y$ to become equal to $a_i \cdot x$ is if $(a_i)_j \cdot (x_j - y_j)$ is equal to $-\sum_{j' \neq j} (a_i)_{j'} \cdot (x_{j'}-y_{j'})$. 
	Since $(x_j - y_j) \neq 0$, there is only one choice of $(a_i)_j$ that can make this equality happen even conditioned on the rest of $a_i$. Thus, with probability at least $1-{1}/{q}$, 
	we have $a_i \cdot y  \neq z_i$. 
\end{proof}

We run $s$ independent iterations thus the failure probability over all iteration is at most $1/q^s$ by~\Cref{clm:rec-one-itr-failure}.
We now count the number of possible choices for $y$ in \Cref{alg-sr-1} given the promises in \Cref{prob:sparse-recovery}.
\begin{claim} \label{clm:rec-choices-x}
	The number of choices for vector $y$ is at most $2 \cdot 2^a \cdot q^a \cdot n^b\cdot q^b$.
\end{claim}
\begin{proof}
	$N(S)$ contains a subset of $T$ of size at most $a$, thus there are $2^a$ possible choices for elements of $N(S)$ within $T$. $N(S)$ contains at most $b$ elements outside $T$, thus there are $\sum_{i=0}^{b} \binom{n-a}{i} \leq \sum_{i=0}^{b} n^i \leq 2 n^b$ choices for elements of $N(S)$ outside $T$. Also, each element can take values between $0$ and $q-1$ in $x$. Thus, the total number of choices for $x$ is at most $2 \cdot 2^a \cdot q^a \cdot n^b\cdot q^b$.
\end{proof}

A union bound over all choices for vector $y$ in~\Cref{clm:rec-choices-x}, using the fact that the probability that each one is mistaken for $x$ is only $1/q^s$, implies that the output will be wrong with probability at most 
\[
	2 \cdot 2^a \cdot q^a \cdot n^b\cdot q^b \cdot \frac{1}{q^s} = n^{-10}. 
\]
This concludes the correctness of the algorithm. 

The sketch size is also $O(\log q)=O(\log c)$ bits for each $z_i$ and thus $s \cdot O(\log q) = O(a \log c + b \log n) = s_0$ bits over all as desired. 
The number of random bits needed however is $O(n \log q)$ in each iteration implying $O(n \cdot (a \log c + b \log n))$ random bits in total. This concludes the proof of~\Cref{lem:inefficient}. 

We again note that the randomness used by this algorithm is too much for our final algorithm to be able to store and on top of that the algorithm requires exponential time for its recovery.

\end{document}